%% file: sh_master_file_part_three.tex
\documentclass[a4paper,11pt,oneside,reqno]{amsart}
\usepackage[utf8]{inputenc}
\usepackage{amsmath,amsthm,amsfonts,latexsym,amssymb,bm,enumerate}
\usepackage{ae}
\usepackage{cite}
\usepackage{float}
\usepackage{lmodern}
\usepackage[T1]{fontenc}

\usepackage{color}

\usepackage[colorlinks=true]{hyperref}

\definecolor{dark-red}{rgb}{.54,.0,.0}
\definecolor{dark-green}{rgb}{.0,.4,.0}
\definecolor{dark-blue}{rgb}{.04,.04,.4}

\hypersetup{linkcolor=dark-red, urlcolor=dark-blue, citecolor=dark-green}

\usepackage[dvips]{graphicx}
\usepackage{psfrag}
\DeclareGraphicsExtensions{.eps,.art,.ART,.ps}

\usepackage{rotating}

\newcounter{mnotecount}[section]

\usepackage[dvips]{graphicx}
\usepackage{psfrag}
\DeclareGraphicsExtensions{.eps,.art,.ART,.ps}

\DeclareFontFamily{OT1}{pzc}{}
\DeclareFontShape{OT1}{pzc}{m}{it}%
              {<-> s * pzcmi8t}{}
\DeclareMathAlphabet{\mathpzc}{OT1}{pzc}%
                                {m}{it}

\newtheorem{Thm}{Theorem}[section]

\newtheorem{Lem}[Thm]{Lemma}
\newtheorem{Cor}[Thm]{Corollary}
\newtheorem{Prop}[Thm]{Proposition}

\newtheorem{Rmk}[Thm]{Remark}

\newcommand{\cal}{\mathcal}

\renewcommand{\Psi}{\rho}

\newcommand{\R}{\mathbb{R}}

\newcommand{\eps}{\varepsilon}

\newcommand{\vlinha}{\frac{e^2}{r}+\frac{\Lambda}{3}r^3-\omega}
\newcommand{\vlinhazz}{\frac{e^2}{r}+\frac{\Lambda}{3}r^3-\omega_0}
\newcommand{\vlinhaz}{\frac{e^2}{{r}_0(0)}+\frac{\Lambda}{3}{r}_0^3(0)-\omega_0}
\newcommand{\muu}{1-\frac{2\omega}{r}+\frac{e^2}{r^2}-\frac{\Lambda}{3}r^2}
\newcommand{\muuz}{1-\frac{2\omega_0}{\tilde{r}_0}+\frac{e^2}{\tilde{r}_0^2}-\frac{\Lambda}{3}\tilde{r}_0^2}
\newcommand{\truemu}{\frac{2\omega}{r}-\frac{e^2}{r^2}+\frac{\Lambda}{3}r^2}
\newcommand{\mysigma}{{\cal P}}
\renewcommand{\omega}{\varpi}
\newcommand{\myumax}{U}

\newcommand{\ru}{\partial_r(1-\mu)(r_+,\varpi_0)}
\newcommand{\ckrm}{\check r_-}
\newcommand{\ckrp}{\check r_+}

\newcommand{\cg}{\Gamma}
\newcommand{\vgr}{v_{\ckrm}}
\newcommand{\ugr}{u_{\ckrm}}
\newcommand{\ugrp}{u_{\ckrp}}
\newcommand{\vgrp}{v_{\ckrp}}
\newcommand{\ug}{u_r}
\newcommand{\vg}{v_r}
\newcommand{\ugs}{u_s}
\newcommand{\vgs}{v_s}
\newcommand{\vgsd}{v_{s_2}}
\newcommand{\ugz}{u_{r_0}}
\newcommand{\vgz}{v_{r_0}}
\newcommand{\ckr}{\check{r}}
\newcommand{\uckr}{u_{\ckr}}
\newcommand{\vckr}{v_{\ckr}}
\newcommand{\udz}{u_{r_+-\delta}}
\newcommand{\vdz}{v_{r_+-\delta}}
\newcommand{\udp}{u'_{r_+-\delta}}
\newcommand{\gam}{\gamma}
\newcommand{\ugam}{u_{\gam}}
\newcommand{\vgam}{v_{\gam}}

\newcommand{\rmu}{\partial_r(1-\mu)}
\newcommand{\dMenos}{2k_-}
\newcommand{\dMais}{2k_+}
\newcommand{\cD}{{\cal D}}

\begin{document}

\newcounter{enumii_saved}

\iffalse
$\ckr$\ \ \ ckr

$\ckrm$\ \ \ ckrm

$\ckrp$\ \ \ ckrp

$\cg$\ \ \ cg

$\ugr$\ \ \ ugr

$\vgr$\ \ \ vgr

$\ugrp$\ \ \ ugrp

$\vgrp$\ \ \ vgrp

$\ug$\ \ \ ug

$\vg$\ \ \ vg

$\ugs$\ \ \ ugs

$\vgs$\ \ \ vgs

$\vgsd$\ \ \ vgsd

$\ugz$\ \ \ ugz

$\vgz$\ \ \ vgz

$\uckr$\ \ \ uckr

$\vckr$\ \ \ vckr

$\udz$\ \ \ udz

$\vdz$\ \ \ vdz

$\udp$\ \ \ udp

$\gam$\ \ \ gam

$\ugam$\ \ \ ugam

$\vgam$\ \ \ vgam

$\vlinha$\ \ \ vlinha

$\vlinhaz$\ \ \ vlinhaz

$\vlinhazz$\ \ \ vlinhazz

$\muu$\ \ \ muu

$\muuz$\ \ \ muuz

$\truemu$\ \ \ truemu

$\omega$\ \ \ omega

$\ru$\ \ \ ru
\fi

% PART 3
\addtocounter{part}{+2}

\title[Global uniqueness with a cosmological constant - Part 3]{On the global uniqueness for the Einstein-Maxwell-scalar field system with a cosmological constant \\ \vspace{.2cm}
{\small Part 3. Mass inflation and extendibility of the solutions}}

\author{Jo\~ao L.~Costa}
\author{Pedro M.~Gir\~ao}
\author{Jos\'{e} Nat\'{a}rio}
\author{Jorge Drumond Silva}

\address{Jo\~ao L.~Costa: 
ISCTE - Instituto Universitário de Lisboa, Portugal
and 
Center for Mathematical Analysis, Geometry and Dynamical Systems,
Instituto Superior T\'ecnico, Universidade de Lisboa, Portugal
}
\email{jlca@iscte.pt}

\address{Pedro M.~Gir\~ao, Jos\'{e} Nat\'{a}rio and Jorge Drumond Silva: 
Center for Mathematical Analysis, Geometry and Dynamical Systems,
Instituto Superior T\'ecnico, Universidade de Lisboa, Portugal}
%Instituto Superior T\'ecnico, Universidade de Lisboa, Av.\ Rovisco Pais, 1049-001 Lisboa, Portugal
\email{pgirao@math.ist.utl.pt}
\email{jnatar@math.ist.utl.pt}
\email{jsilva@math.ist.utl.pt}

\subjclass[2010]{Primary 83C05; Secondary 35Q76, 83C22, 83C57, 83C75}
\keywords{Einstein equations, black holes, strong cosmic censorship, Cauchy horizon, scalar field, spherical symmetry}
\thanks{Partially funded by FCT/Portugal through project PEst-OE/EEI/LA0009/2013.
P.~Girão and J.~Silva were also partially funded by FCT/Portugal through grants PTDC/MAT114397/2009 and UTA$\underline{\ }$CMU/MAT/0007/2009.}

\maketitle

\input{sh_abstract_part_three}
{
\setcounter{tocdepth}{1}
\tableofcontents
}
\input{sh_introduction_for_part_three}
\input{sh_formulas_for_part_three}

\input{sh_part_three}
\input{sh_appendix_for_part_three}

\end{document}

%% file: sh_abstract_part_three.tex
%\vspace{-.7cm}

\begin{center}
{\bf Abstract}
\end{center}

This paper is the third part of a trilogy dedicated to the following problem: 
given spherically symmetric characteristic
initial data for the Einstein-Maxwell-scalar
field system with a cosmological constant $\Lambda$,
with the data on the outgoing initial null
hypersurface given by a subextremal Reissner-Nordstr\"{o}m black
hole event horizon, study the future extendibility of the corresponding
maximal globally hyperbolic development as a  ``suitably regular'' Lorentzian manifold. 

In the first part~\cite{relIst1} of this series we established the well posedness of the characteristic problem, whereas
in the second part~\cite{relIst2} we studied the stability of the radius function at the Cauchy horizon. 

In this third and final paper we show that, depending on the decay rate of the initial data, mass inflation may or may not occur.
When the mass is controlled, it is possible to obtain continuous extensions of the metric across the Cauchy horizon with square integrable Christoffel symbols. 
Under slightly stronger conditions, we can bound the gradient of the scalar field. This allows the construction of (non-isometric) extensions 
of the maximal development which are classical solutions of the Einstein equations. 
Our results provide evidence against the validity of the strong cosmic censorship conjecture when $\Lambda>0$.

%\vspace{-.5cm}

\newpage

%% file: sh_introduction_for_part_three.tex
\section{Introduction}

The Einstein equations are a covariant system of equations relating the geometry of spacetime to its energy and matter content. They are written in appropriate units as
\begin{equation}\label{Einstein}
R_{\mu\nu} - \frac12 R g_{\mu\nu} + \Lambda g_{\mu\nu} = 2 T_{\mu\nu},
\end{equation}
where $g_{\mu\nu}$ is the spacetime metric, $R_{\mu\nu}$ is the corresponding Ricci tensor, $R$ is the scalar curvature, $\Lambda$ is the cosmological constant and $T_{\mu\nu}$ is the energy-momentum tensor. In vacuum, $T_{\mu\nu}=0$ and \eqref{Einstein} is then a closed system of partial differential equations for the metric $g_{\mu\nu}$. In general, however, $T_{\mu\nu}$ depends on nonvanishing matter fields, and so \eqref{Einstein} must be coupled to other equations governing the matter dynamics.

In suitable coordinates, the Einstein equations become a system of quasi-linear wave equations, naturally leading to an initial value problem. Nonetheless, unlike in other evolution equations of Mathematical Physics, the spacetime geometry is not known {\em a priori}. This enables the occurrence of unexpected phenomena, in particular the failure of uniqueness of global solutions without loss of regularity.

An example is provided by the Reissner-Nordstr\"om family of solutions\footnote{Throughout this work we will simply use ``Reissner-Nordst\"om'' to mean any of the anti-de Sitter ($\Lambda<0$), the asymptotically flat ($\Lambda=0$), or the de Sitter ($\Lambda>0$) Reissner-Nordstr\"{o}m solutions.}, corresponding to static charged black holes (see for instance \cite{HE95}).
%(we are mainly interested in this solution as a spherically symmetric proxy for the Kerr solution, the fundamental metric of black hole physics). 
Regarding these solutions as arising from appropriate initial value problems, we encounter the surprising phenomenon that the maximal globally hyperbolic developments (informally, the largest Lorentzian manifolds determined from the initial data, via the Einstein equations) are smoothly extendible, as solutions, in a highly non-unique way. Therefore, global uniqueness fails for the Einstein equations. This puts into question the deterministic character of General Relativity, since what happens in such extensions is not uniquely determined by initial data.

The boundary of the maximal globally hyperbolic development in any given extension is known as the Cauchy horizon, and signals the breakdown of global uniqueness; in the Reissner-Nordstr\"om solutions this horizon occurs in the black hole interior. In \cite{SimpsonInternal}, Penrose and Simpson proposed a heuristic mechanism, known as the blue-shift effect, by which arbitrarily small perturbations of the black hole exterior can be infinitely amplified along the Cauchy horizon, turning it into a ``singularity'' beyond which no ``suitably regular'' extensions exist. Therefore, the previously discussed pathological features of the Reissner-Nordstr\"om family would be artifacts of those particular solutions, unstable under perturbations, and therefore devoid of physical meaning. This reinstated the belief in global uniqueness as a generic property of reasonable initial value problems for the Einstein equations, an idea substantiated in Penrose's {\em strong cosmic censorship conjecture} (SCCC), see~\cite{PenroseSingularities}, \cite{ChruscielSCC} and \cite{ChristodoulouGlobalnew}.

Making the notion of ``suitably regular'' precise is remarkably subtle, and has evolved considerably during the last decades: the original expectation was that the nonlinearities of the Einstein equations would suffice to turn the Cauchy horizon of the Reissner-Nordstr\"om solution, under arbitrarily small perturbations, into a Schwarzschild-like singularity, across which not even continuous extensions of the metric are possible (see \cite{Sbierski}). This would completely settle the question of the SCCC in the affirmative. On the other hand, the existence of generic $C^2$ extensions as solutions of the Einstein equations would completely falsify this conjecture. Nonetheless, $C^2$-inextendibility does not necessarily provide a compelling argument in favor of global uniqueness, since there are relevant solutions of the Einstein equations whose regularity is well below this threshold. Thus a formulation of the SCCC in terms of the generic blow up of the Kretschmann scalar (tidal forces), favoured by many authors, is manifestly insufficient, as it only rules out $C^2$ extensions.

In a seminal paper, Poisson and Israel \cite{IsraelPoisson} identified the blow up of a renormalized version of the Hawking mass at the Cauchy horizon as a consequence of
the blue-shift mechanism (see also \cite{hiscock} and \cite{OriInner}), implying the blow up of the Kretchmann scalar. This scenario, 
which was confirmed by Dafermos in his celebrated non-linear analysis of the spherically symmetric Einstein-Maxwell-scalar field system with $\Lambda=0$ (see \cite{Dafermos1, Dafermos2}),
became known as mass inflation. As a consequence, the current expectation is that, for generic initial data in the context of black hole spacetimes, the metric extends 
beyond the Cauchy horizon in $C^0$ but not in $C^2$. To accomodate these developments, Christodoulou \cite{Christodoulou:2008} proposed a formulation of the SCCC that excludes the generic existence of extensions of the metric with square integrable connection coefficients. As already suggested by Chru\'sciel, this guarantees that (generically) the potential extensions will not be regular enough to solve the Einstein equations, even in a weak sense.
By now, there is strong evidence that this formulation of the SCCC, which we will refer to as the Christodoulou-Chru\'sciel criterion, holds for asymptotically flat black holes (see \cite{DafermosBlack, LukWeak}). 

 It turns out that for cosmological black holes, i.e.\ black hole solutions of the Einstein equations with a positive cosmological constant $\Lambda$, the instability mechanism is expected to be
weaker. Although the introduction of this term has a negligible impact on the causal structure of the black hole interior, where the blueshift occurs, it has dramatic consequences for the structure of the exterior. In particular, a new horizon, known as the cosmological horizon, is formed. This generates an extra redshift, which counteracts the blueshift effect.
Related to this is Price's law, which predicts that, in Eddington-Finkelstein coordinates, the decay of the perturbations along the event horizon is polynomial in the asymptotically flat case, but exponential in the cosmological case.
From the start, it was clear that these facts could have a strong influence on the issue of stability of the Cauchy horizon.
This was enough to generate a considerable amount of activity concerning the SCCC in the positive cosmological constant setting, raising the possibility that mass inflation might not occur (see \cite{BradyCauchy, BradyCauchyMass, BradyCosmic}). After intense debate, the consensus was reached that the SCCC would probably prevail, at least in its weaker $C^2$
formulation. Unfortunately, most of the reasoning leading to this conclusion was based on heuristic arguments or perturbative calculations, with, for instance, the back-reaction of the metric being
``put in by hand''. %This last fact by itself requires an analysis that takes into account the entire non-linear structure of the Einstein equations. 
This calls for a more detailed analysis, that takes into account the entire non-linear structure of the Einstein equations and is able to capture the fine regularity properties of potential extensions, especially in view of the growing popularity of the Christodoulou-Chru\'sciel criterion. This analysis has became even more pertinent since, in the meantime, it was discovered that supernovae observations are best fitted by models with $\Lambda>0$. 

This paper is the third part of a trilogy dedicated to the full non-linear evolution, inside a black hole, of the Einstein equations \eqref{Einstein} with nonvanishing cosmological constant $\Lambda$. The matter model consists of a massless scalar field $\phi$ and an electromagnetic field $F$, satisfying the Maxwell and wave equations
\begin{align}
& dF = d\star F = 0, \\
& \Box \phi = 0,
\end{align}
where $\star$ is the Hodge star operator and $\Box$ is the d'Alembertian (both depending on $g$). These equations couple to \eqref{Einstein} through the energy-momentum tensor
\begin{align}
& T_{\mu\nu} = \partial_\mu \phi \, \partial_\nu \phi - \frac12 \partial_\alpha \phi \, \partial^\alpha \phi \, g_{\mu\nu} + F_{\mu\alpha} F_{\nu}^{\,\,\alpha} - \frac14 F_{\alpha\beta} F^{\alpha \beta} g_{\mu\nu}. \label{Tmunu}
\end{align}
We choose this matter model because we wish to consider spherically symmetric perturbations of the Cauchy horizon of the Reissner-Nordstr\"{o}m spacetime; since Birkhoff's theorem implies that this is the only spherically symmetric electrovacuum solution, we also introduce a self-gravitating real massless scalar field to provide dynamical degrees of freedom with the same hyperbolic character.

More precisely, we study the following problem: 
given spherically symmetric characteristic
initial data for the Einstein-Maxwell-scalar
field system \eqref{Einstein}$-$\eqref{Tmunu},
with the data on the outgoing initial null
hypersurface given by a subextremal Reissner-Nordstr\"{o}m black
hole event horizon, 
and the remaining data otherwise free,
study the future extendibility of the corresponding
maximal globally hyperbolic development as a  ``suitably regular'' Lorentzian manifold. 
% (see the discussion in the Introduction of Part~1).
Strictly speaking, this problem does not address the strong cosmic censorship conjecture directly, 
because the data considered on the event horizon does not arise from the gravitational
collapse of generic Cauchy initial data. 
However, since the Price law for $\Lambda>0$ is widely expected to yield exponential decay of the scalar field along the event horizon 
(see for instance the linear analysis of~\cite{DafermosWavedeSitter, Dyatlov}),
we believe that our conclusions will provide valuable insights for this cosmological case.

In Part~1 \cite{relIst1} of this trilogy, we showed the equivalence (under appropriate regularity conditions for the initial data) between the Einstein-Maxwell-scalar
field equations \eqref{Einstein}$-$\eqref{Tmunu} under spherical symmetry
and the system of first order PDEs~\eqref{r_u}$-$\eqref{kappa_at_u}. We established existence, uniqueness and identified a breakdown criterion for solutions of this system.
In Part~2 \cite{relIst2}, we analyzed the properties of the solution
up to the Cauchy horizon, proving, in particular, the stability of the radius function. See Section~\ref{main-one-two} for a summary of our previous results, as well as for the definitions and notation that will be used henceforth.
We refer the reader to Parts 1 and 2 for the complete details.

In this paper we examine the behavior of the renormalized Hawking mass $\varpi$ (see~\eqref{bar_rafaeli}) and the scalar field at the Cauchy horizon. Depending on the control that we have on these quantities, we are able to construct extensions of the metric beyond the Cauchy horizon with different degrees of regularity. The quotient
\[
\Psi:=\frac{k_-}{k_+} > 1,
\]
of the surface gravities (see \eqref{ganhamos}) of the Cauchy and event horizons $r=r_-$ and $r=r_+$ in the reference Reissner-Nordstr\"{o}m black hole
plays an important role in our analysis.

We start by briefly recalling the strategy of Dafermos~\cite{Dafermos1, Dafermos2} to establish mass inflation (that is, blow-up of $\varpi$ at the Cauchy horizon), which naturally generalizes to the case of a non-vanishing cosmological constant. This requires the initial field $\zeta_0$
(see~\eqref{zeta} and~\eqref{RN1}) to satisfy (for constants $c>0$, $U>0$ and $u\in\left[0,U\right]$)
$$
\zeta_0(u)\geq cu^s \ {\rm for\ some} \ 0<s<\frac\Psi 2-1,
$$
see Theorem~\ref{inflation}. 
Since the mass is a scalar invariant involving first derivatives of the metric, its blow up excludes the existence
of spherically symmetric $C^1$ extensions.\footnote{In this paper we will only be concerned with extensions that are also spherically symmetric.} Moreover,
using the techniques of~\cite{DafermosBlack}, one can hope to prove that in this case the Christodoulou-Chru\'sciel inextendibility criterion holds, that is, there is no extension of the metric beyond the Cauchy horizon with Christoffel symbols in $L^2_{\rm loc}$.

The previous approach only allows us to explore a particular subregion of parameter space, corresponding to sufficiently subextremal reference solutions (see the figure below). We proceed by extending the analysis to the full parameter range. First we prove that if the field $\zeta_0$ satisfies the weaker hypothesis
$$
\zeta_0(u)\geq cu^s \ {\rm for\ some} \ 0<s<\Psi-1,
$$
then either the renormalized mass $\varpi$ or the field $\bigl|\frac\theta\lambda\bigr|$ 
(see~\eqref{lambda_0} and~\eqref{theta})
blow up at the Cauchy horizon (Theorem~\ref{inflation2}). As a consequence, the Kretschmann curvature scalar also blows up (Remark~\ref{RmkSCC1}).

On the other hand, when the initial field $\zeta_0$ satisfies
$$
|\zeta_0(u)|\leq cu^s \ {\rm for\ some} \ s>\frac{7\Psi}9-1 > 0,
$$
we show that the mass remains bounded (Theorem~\ref{thmNoMass}). This behavior is in contrast with the standard picture of spherically symmetric gravitational collapse.

The case where no mass inflation occurs is then analyzed in further detail. We construct $C^0$ spherically symmetric extensions of the metric beyond the Cauchy horizon with
the second mixed derivatives of $r$ continuous. There are two natural coordinate choices for the extension, corresponding to either $\lambda=-1$ or $\kappa=1$ (see \eqref{lambda_0} and \eqref{kappa_0})
on the outgoing null ray $u=U$.
Interestingly, these lead to inequivalent $C^2$ structures for the extended manifolds, a fact that is reflected on the behavior of the Christoffel symbols: when the initial data satisfies 
$$
c_2u^{s_2}\leq\zeta_0(u)\leq c_1u^{s_1} \ {\rm for\ some} \ \frac{7\Psi}9-1<s_1\leq s_2<\Psi-1,
$$
we prove that one of the Christoffel symbols blows up on $u=U$ at the Cauchy horizon in the $\lambda=-1$ coordinates, but not in the $\kappa=1$ coordinates. Moreover, for both coordinate systems this Christoffel symbol blows up along almost all outgoing null rays, which excludes the existence of $C^{0,1}$ extensions of the metric. Nonetheless, in the $\kappa=1$ coordinates the Christoffel symbols are in $L^2_{{\rm loc}}$ 
and the field $\phi$ is in $H^1_{{\rm loc}}$ for the whole range of initial data where there is no mass inflation (Corollary~\ref{cc-fails}). Therefore, the Christodoulou-Chru\'sciel inextendibility criterion for strong cosmic censorship does not hold in this setting.

% We use two
% coordinate systems. In the first one 
% $\tilde v=r(U,0)-r(U,v)$ is the outgoing null coordinate, 
% and the Cauchy horizon corresponds to the finite coordinate $\tilde v=V$.
% In the second one, $\hat{v}=\int_0^v\kappa(u,\bar{v})\,d\bar{v}$.
% The extensions of these two coordinate systems to the Cauchy horizon are not
% equivalent when $\bigl|\frac{\theta}{\lambda}\bigr|$ is unbounded.
% We construct $C^0$ spherically symmetric extensions of the metric beyond the Cauchy horizon with
% the second mixed derivatives of $r$ continuous (Remark~\ref{dias_l}).
%   In $(u,\hat{v})$ coordinate system, this can 
%  be done so that the Christoffel symbols of the metric are in $L^2_{{\rm loc}}$ 
%  and the field $\phi$ is in $H^1_{{\rm loc}}$ (Corollary~\ref{cc-fails}).
% So, the Christodoulou-Chru\'sciel inextendibility criterion 
% for strong cosmic censorship does not hold.

\iffalse
Assuming that
$$
c_2u^{s_2}\leq\zeta_0(u)\leq c_1u^{s_1}\ {\rm for\ some}\ \frac{7\Psi}9-1<s_1\leq s_2<\Psi-1,
$$
we prove that the Christoffel symbol $\Gamma^{\tilde v}_{\tilde v \tilde v}$ blows up at the Cauchy horizon, while all the other Christoffel symbols are bounded (Theorems~\ref{meio}
and~\ref{meiof}).
\fi
%It follows that in any other coordinate system that covers the Cauchy horizon the metric does 
%not have bounded Christoffel symbols (Remark~\ref{ubd}; compare 
%with~\cite[Conjecture~4]{DafermosBlack}). 

Finally, assuming that
$$
|\zeta_0(u)|\leq cu^s \ {\rm for\ some} \ s>\frac{13\Psi}9-1,
$$
we can bound the field $\frac{\theta}{\lambda}$ at the Cauchy horizon. This allows us to prove
that the solution of the first order system
extends, non-uniquely, to a classical solution beyond the Cauchy horizon (Theorem~\ref{para-la}). 
We then show that this solution corresponds to a classical solution of the Einstein
equations extending beyond the Cauchy horizon (Theorem~\ref{fine}).
The metric for this solution is $C^1$ and such that $r\in C^2$ and $\partial_u\partial_v\Omega$ 
(see~\eqref{metric}) exists and is continuous (Remark~\ref{fail}).
However, we emphasize that the metric does not have to be $C^2$, in spite of the
Kretschmann curvature scalar being bounded.
To the best of our knowledge, these are the first results where the generic existence of extensions as solutions of the Einstein equations is established.

It should be noted that these results, while valid for all signs of the cosmological constant $\Lambda$, only provide evidence against the strong cosmic censorship conjecture in the case $\Lambda > 0$, since, as discussed above, only in this case does one expect an exponential decay of perturbations along the event horizon, even in the absence of symmetry assumptions.

In summary, for a given $\Psi$ and $cu^s\leq\zeta_0(u)\leq Cu^s$, the behavior of the solution at the Cauchy horizon depends on the value of $s$ as described in the following figure. The lines represent the limits of the various strict inequalities above. 

\begin{center}
\begin{psfrags}
\psfrag{b}{\tiny $\!\frac 97$}
\psfrag{u}{\tiny $\!\frac 92$}
\psfrag{c}{\tiny $2$}
\psfrag{d}{\tiny $3$}
\psfrag{e}{\tiny $4$}
\psfrag{z}{\tiny $\!\!1$}
\psfrag{g}{\tiny $\!\!\!\frac 49$}
\psfrag{h}{\tiny $\!\!1$}
\psfrag{i}{\tiny $\!\!2$}
\psfrag{j}{\tiny $\!\!3$}
\psfrag{k}{\tiny $s=\frac \Psi 2-1$}
\psfrag{m}{\tiny $s=\Psi-1$}
\psfrag{l}{\tiny $s=\frac{7\Psi}9-1$}
\psfrag{n}{\tiny $s=\frac{13\Psi}9-1$}
\psfrag{r}{\tiny $\Psi$}
\psfrag{s}{\tiny $s$}
\psfrag{o}{\tiny $1$}
\psfrag{p}{\tiny $\!\!0$}
\includegraphics[scale=.7]{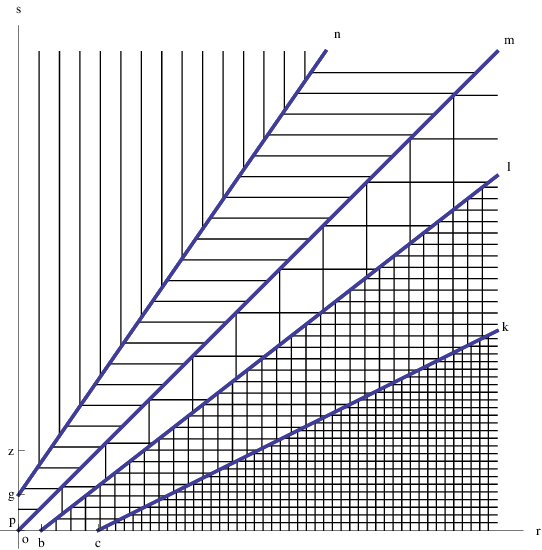}
\end{psfrags}
\end{center}
\begin{center}
\begin{psfrags}
\psfrag{e}{\tiny \!\!\!\!\!\!\!\!\!\!\!\!\!\!\!\!\!mass inflation}
\psfrag{u}{\tiny \!\!\!\!\!\!\!\!\!\!\!\!\!\!\!\!\!mass inflation}
\psfrag{b}{\tiny \!\!\!\!\!\!\!\!\!\!\!\!\!\!\!\!\!\!\!\!\!no mass inflation}
\psfrag{a}{\tiny \!\!\!\!\!\!\!\!\!\!\!\!\!\!\!\!\!\!\!\!\!no mass inflation}
\psfrag{c}{\tiny \!\!\!\!\!\!\!\!\!\!\!\!\!\!\!\!\!\!\!\!\!no mass inflation}
\psfrag{d}{\tiny \!\!\!\!\!\!\!\!\!\!\!\!\!\!\!$\theta/\lambda$ bounded}
\psfrag{f}{\tiny \!\!\!\!\!\!\!\!\!\!\!\!\!\!\!\!\!\!\!$\theta/\lambda$ unbounded}
\psfrag{v}{\tiny \!\!\!\!\!\!\!\!\!\!\!\!\!\!\!\!\!\!\!\!\!\!\!\! or $\theta/\lambda$ unbounded}
\psfrag{g}{\tiny \!\!\!\!\!\!\!\!\!\!\!\!\!\!\!\!\!\!\!\!\!smooth extension}
\psfrag{h}{\tiny \!\!\!\!\!\!\!\!\!\!\!\!\!\!\!\!\!\!\!\!\!\!\!\!\!\!\!\!\!beyond Cauchy horizon}
\includegraphics[scale=1.2]{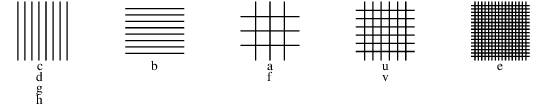}
\end{psfrags}
\end{center}

\bigskip

In Appendix~\ref{appendix-2} we explain how $\Psi$ depends on the physical parameters $r_+$, $r_-$ and $\Lambda$. In particular,
$\Psi$ is a function of $\frac{r_+}{r_-}$ and $\frac{\Lambda r_-^2}3$. 

\vspace{3mm}

\noindent {\bf Acknowledgments.} The authors thank M.~Dafermos for bringing
the {\em Epilogue}\/ of~\cite{DafermosBlack} to their attention.

%% file: sh_formulas_for_part_three.tex
\section{Framework and some results from Parts~1 and~2}\label{main-one-two}

\subsection*{The spherically symmetric Einstein-Maxwell-scalar field system with a cosmological constant}

Consider a spherically symmetric spacetime with metric
\begin{equation}\label{metric}
g=-\Omega^2(u,v)\,dudv+r^2(u,v)\,\sigma_{\mathbb{S}^2},
\end{equation}
where $\sigma_{\mathbb{S}^2}$ is the round metric on the 2-sphere. In this case, the Maxwell equations decouple from the system, since they can be immediately solved to yield
\[
F = - \frac{e \, \Omega^2(u,v)}{2 \, r^2(u,v)} \, du \wedge dv.
\]
Here $e$ is a real constant, corresponding to a total electric charge $4 \pi e$, and we have assumed zero magnetic charge without loss of generality.
The remaining equations can then be written as follows:
a wave equation for $r$,
\begin{equation}\label{wave_r} 
\partial_u\partial_vr=-\frac{\Omega^2}{4r} - \frac{\partial_ur\,\partial_vr}{r} + \frac{\Omega^2e^2}{4r^3} + \frac{\Omega^2 \Lambda r}{4},
\end{equation}
a wave equation for the massless scalar field $\phi$,
\begin{equation}\label{wave_phi} 
\partial_u\partial_v\phi=-\,\frac{\partial_ur\,\partial_v\phi+\partial_vr\,\partial_u\phi}{r},
\end{equation}
the Raychaudhuri equation in the $u$ direction,
\begin{equation}\label{r_uu} 
\partial_u\left(\frac{\partial_ur}{\Omega^2}\right)=-r\frac{(\partial_u\phi)^2}{\Omega^2},
\end{equation}
the Raychaudhuri equation in the $v$ direction,
\begin{equation}\label{r_vv} 
\partial_v\left(\frac{\partial_vr}{\Omega^2}\right)=-r\frac{(\partial_v\phi)^2}{\Omega^2},
\end{equation}
and a wave equation for $\ln\Omega$,
 \begin{equation}\label{wave_Omega} 
\partial_v\partial_u\ln\Omega=-\partial_u\phi\,\partial_v\phi-\,\frac{\Omega^2e^2}{2r^4}+\frac{\Omega^2}{4r^2}+\frac{\partial_ur\,\partial_vr}{r^2}.
\end{equation}

\subsection*{The first order system}

Given $r$, $\phi$ and $\Omega$, solutions of the Einstein equations,
we define the following quantities:
\begin{equation}\label{nu_0}
\nu:=\partial_u r,
\end{equation}
\begin{equation}\label{lambda_0}
\lambda:=\partial_v r,
\end{equation}
\begin{equation}\label{bar_rafaeli} 
\omega:=\frac{e^2}{2r}+\frac{r}{2}-\frac{\Lambda}{6}r^3+\frac{2r}{\Omega^2}\nu\lambda,
\end{equation}
\begin{equation}\label{mu_0} 
\mu:=\truemu,
\end{equation}
\begin{equation}\label{theta} 
\theta:=r\partial_v\phi,
\end{equation}
\begin{equation}\label{zeta} 
\zeta:=r\partial_u\phi
\end{equation}
and
\begin{equation}\label{kappa_0} 
 \kappa:=\frac{\lambda}{1-\mu}.
\end{equation}
Notice that we may rewrite~\eqref{bar_rafaeli} as
\begin{equation}\label{omega_sq}
\Omega^2=-\,\frac{4\nu\lambda}{1-\mu}=-4\nu\kappa.
\end{equation}
The Einstein equations imply the first order system for $(r,\nu,\lambda,\varpi,\theta,\zeta,\kappa)$
\begin{eqnarray} 
 \partial_ur&=&\nu\label{r_u},\\
 \partial_vr&=&\lambda\label{r_v},\\
 \partial_u\lambda&=&\nu\kappa\partial_r(1-\mu)\label{lambda_u},\\
 \partial_v\nu&=&\nu\kappa\partial_r(1-\mu),\label{nu_v}\\
 \partial_u\omega&=&\frac 12(1-\mu)\left(\frac\zeta\nu\right)^2\nu,\label{omega_u}\\
 \partial_v\omega&=&\frac 12\frac{\theta^2}{\kappa},\label{omega_v}\\
 \partial_u\theta&=&-\,\frac{\zeta\lambda}{r},\label{theta_u}\\
 \partial_v\zeta&=&-\,\frac{\theta\nu}{r},\label{zeta_v}\\
 \partial_u\kappa&=&\kappa\nu\frac 1r\left(\frac{\zeta}{\nu}\right)^2,\label{kappa_u}
\end{eqnarray}
with the restriction
\begin{equation}\label{kappa_at_u} 
\lambda=\kappa(1-\mu).
\end{equation}
Here by $\partial_r(1-\mu)$ we mean
$$
\partial_r(1-\mu) = \partial_r \left( 1 - \frac{2\omega}{r} + \frac{e^2}{r^2} - \frac{\Lambda}{3}r^2 \right) = \frac{2\omega}{r^2} - \frac{2e^2}{r^3} - \frac{2\Lambda}{3}r.
$$
Under appropriate regularity conditions for the initial data, the system of first order PDE~\eqref{r_u}$-$\eqref{kappa_at_u} also implies the spherically symmetric Einstein equations \eqref{wave_r}$-$\eqref{wave_Omega} (Part 1, Propositions 6.3 and 6.4).

\subsection*{Initial data}

We take the initial data on the $v$ axis to be the data on the event horizon of a subextremal Reissner-Nordstr\"{o}m solution with mass $M$, electric charge $e$ and cosmological constant $\Lambda$.
Therefore we choose initial data as follows:
\begin{equation}\label{RN1} 
\left\{
\begin{array}{lclcl}
 r(u,0)&=&r_0(u)&=&r_+-u,\\
 \nu(u,0)&=&\nu_0(u)&=&-1,\\
 \zeta(u,0)&=&\zeta_0(u),
\end{array}
\right.\qquad{\rm for}\ u\in[0,U],
\end{equation}
where $0 < U < r_+$, and
\begin{equation}\label{RN2} 
\left\{
\begin{array}{lclcl}
 \lambda(0,v)&=&\lambda_0(v)&=&0,\\
 \varpi(0,v)&=&\omega_0(v)&=&M,\\
 \theta(0,v)&=&\theta_0(v)&=&0,\\
 \kappa(0,v)&=&\kappa_0(v)&=&1,
\end{array}
\right.\qquad{\rm for}\ v\in[0,\infty[.
\end{equation}
Here $r_+>0$ is the radius of the event horizon. We assume that $\zeta_0$ is continuous with $\zeta_0(0)=0$,
and denote $M$ by $\varpi_0$.

\subsection*{Well posedness of the first order system and stability of the radius at the Cauchy horizon}

\begin{Thm} [Part 1, Theorem 4.4] \label{existence}
The characteristic initial value problem\/~\eqref{r_u}$-$\eqref{kappa_at_u}
with initial data\/~\eqref{RN1}$-$\eqref{RN2} %Assume $\zeta_0$ is continuous and $\zeta_0(0)=0$. 
has a unique solution defined on a maximal past set $\mysigma$
containing a neighborhood of $[0,\myumax]\times\{0\}\cup\{0\}\times[0,\infty[$.
\end{Thm}

\begin{Lem} [Part 2, Lemma 3.1] \label{sign} 
Let $(r,\nu,\lambda,\omega,\theta,\zeta,\kappa)$ be the maximal solution of the characteristic initial value problem\/~\eqref{r_u}$-$\eqref{kappa_at_u} with initial 
conditions\/~\eqref{RN1} and\/~\eqref{RN2}.
Then:
\begin{itemize}
\item $\kappa$ is positive;
\item $\nu$ is negative;
\item $\lambda$ is negative on ${\cal P}\setminus\{0\}\times[0,\infty[$;
\item $1 - \mu$ is negative on ${\cal P}\setminus\{0\}\times[0,\infty[$;
\item $r$ is decreasing with both $u$ and $v$;
\item $\omega$ is nondecreasing with both $u$ and $v$.
\end{itemize}
\end{Lem}

\begin{Thm} [Part 2, Theorem 1.1] \label{r-stability}
There exists $U>0$ such that
$\mysigma$ contains $[0,\myumax]\times[0,\infty[$. Moreover,
$$
\inf_{[0,\myumax]\times[0,\infty[}r>0
$$
and
\begin{equation}\label{r-menos} 
\lim_{u\searrow 0}r(u,\infty)=r_-.
\end{equation}
\end{Thm}
Here $r_->0$ is the radius of the Cauchy horizon of the Reissner-Nordstr\"om reference solution and $$r(u, \infty) = \lim_{v \to \infty} r(u,v)$$ 
(which exists and is decreasing). Similarly, we also define $$\varpi(u, \infty) = \lim_{v \to \infty} \varpi(u,v).$$

Following the same argument as in~\cite[Section~11]{Dafermos2}, Theorem~\ref{r-stability} implies that the spacetime is extendible across the Cauchy horizon with a $C^0$ metric.

\subsection*{Two effects of any nonzero field}

\begin{Thm} [Part 2, Theorem 8.1] \label{rmenos} 
 Suppose that there exists a positive sequence $(u_n)$ converging to~0 such that $\zeta_0(u_n)\neq 0$.
 Then $r(u,\infty)<r_-$ for all $u>0$.
\end{Thm}

\begin{Lem} [Part 2, Lemma 8.2] \label{l-kappa} 
 Suppose that there exists a positive sequence $(u_n)$ converging to~0 such that $\zeta_0(u_n)\neq 0$.
Then
\begin{equation}\label{integral-k} 
\int_{0}^\infty\kappa(u,v)\,dv<\infty\ {\rm for\ all}\ u>0.
\end{equation}
\end{Lem}
This lemma implies that the affine parameter of an outgoing null geodesic is finite at the Cauchy horizon.

\subsection*{Well posedness for the backwards problem}

In Section~\ref{CH}, we will extend the solutions of Einstein's equations beyond the Cauchy horizon.
For this we will need to solve a backwards problem, already discussed in Part 1.
The initial conditions will be prescribed as follows:
$$
({\rm I}_u)\qquad\left\{
\begin{array}{lcl}
 r(u,0)&=&r_0(u),\\
 \nu(u,0)&=&\nu_0(u),\\
 \zeta(u,0)&=&\zeta_0(u),
\end{array}
\right.\qquad{\rm for}\ u\in\left]0,U\right],
$$
$$
({\rm I}^v)\qquad\left\{
\begin{array}{lcl}
 \lambda(U,v)&=&\lambda_0(v),\\
 \omega(U,v)&=&\omega_0(v),\\
 \theta(U,v)&=&\theta_0(v),\\
 \kappa(U,v)&=&\kappa_0(v),
\end{array}
\right.\qquad{\rm for}\ v\in[0,V].
$$
We let $$\tilde{r}_0(v)=r_0(U)+\int_0^v\lambda_0(v')\,dv',$$ for $v\in[0,V]$.
We assume the regularity conditions:
\begin{align*}
\text{(h1)} \qquad & \text{the functions } \nu_0, \zeta_0, \lambda_0, \theta_0 \text{ and } \kappa_0 \text{ are continuous, and} \\
& \text{the functions } r_0 \text{ and } \omega_0 \text{ are continuously differentiable.}
\end{align*}
We assume the sign conditions:
$$
\hspace{-4.2cm}
\text{(h2)} \qquad 
\left\{
\begin{array}{ll}
r_0(u)>0&{\rm for}\ u\in\left]0,U\right],\\
\tilde{r}_0(v)>0&{\rm for}\ v\in[0,V],\\
\nu_0(u)<0&{\rm for}\ u\in\left]0,U\right],\\
\kappa_0(v)>0&{\rm for}\ v\in[0,V].
\end{array}\right.
$$
We assume the three compatibility conditions:
\begin{align} 
& r_0'=\nu_0,\label{initial_rp}\\
\text{(h3)} \qquad & \omega_0'=\frac 12\frac{\theta_0^2}{\kappa_0},\label{initial_v}\\
& \lambda_0=\kappa_0\left(\muuz\right).\label{kappa_at_zero} \hspace{3cm}
\end{align}

\begin{Thm} [Part 1, Theorem 4.5] \label{existence-alt}
The characteristic initial value problem
with initial conditions\/~{\rm (I$_u$)} and\/~{\rm (I$^v$)} satisfying\/ {\rm (h1)}$-${\rm (h3)}
has a unique solution defined on a maximal reflected past set\footnote{By reflected past set we mean a set $\mathcal{R}$ such that if $(u,v)\in \mathcal{R}$ then $[u,U]\times[0,v]\subset\mathcal{R}$.}
$\mathcal{R}$ containing a neighborhood of $]0,U]\times\{0\}\cup\{U\}\times[0,V]$.
\end{Thm}

\subsection*{Retrieving the Einstein equations from the first order system}

We consider the additional regularity condition
\[
\text{(h4)} \qquad \nu_0, \kappa_0 \text{ and } \lambda_0 \text{ are continuously differentiable.} \hspace{2cm}
\]

\begin{Lem} [Part 1, Proposition 6.2]  \label{regular} 
Suppose that $(r,\nu,\lambda,\omega,\theta,\zeta,\kappa)$ is the solution of the characteristic initial value problem, or the backwards problem,
with initial data satisfying\/ {\rm (h1)} to\/ {\rm (h4)}.
Then the function $r$ is $C^2$, and $\kappa$ is $C^1$.
\end{Lem}

\begin{Prop} [Part 1, Proposition 6.3] \label{einstein-2} 
Under the same hypotheses, the functions $r$, $\phi$ and $\Omega$ satisfy the spherically symmetric Einstein equations~\eqref{wave_r}, \eqref{wave_phi}, \eqref{r_uu} and~\eqref{r_vv}.
\end{Prop}

\begin{Prop} [Part 1, Proposition 6.4] \label{einstein-3} 
Under the same hypotheses, the first order system~\eqref{r_u}$-$\eqref{kappa_at_u} implies \eqref{wave_Omega}.
Since equations~\eqref{wave_r}$-$\eqref{r_vv} imply the first order system,
equations~\eqref{wave_r}$-$\eqref{r_vv} also imply~\eqref{wave_Omega}.
\end{Prop}

\subsection*{The partition of spacetime into four regions}
In Part~2,
we divide $[0,U]\times[0,\infty[$ into four disjoint regions, separated by three curves, $\cg_{\ckrp}$, $\cg_{\ckrm}$ and $\gam$, where different estimates can be obtained (see Appendix~\ref{Costa}).
Next we explain how these curves are constructed.

\begin{center}
\begin{turn}{45}
\begin{psfrags}
\psfrag{u}{{\tiny $u$}}
\psfrag{v}{{\tiny $v$}}
\psfrag{x}{{\tiny $\!\!\!\cg_{\ckrp}$}}
\psfrag{y}{{\tiny $\!\!\cg_{\ckrm}$}}
\psfrag{z}{{\tiny $\gam$}}
\psfrag{e}{{\tiny $\!\!\!\!\!U$}}
\includegraphics[scale=.7]{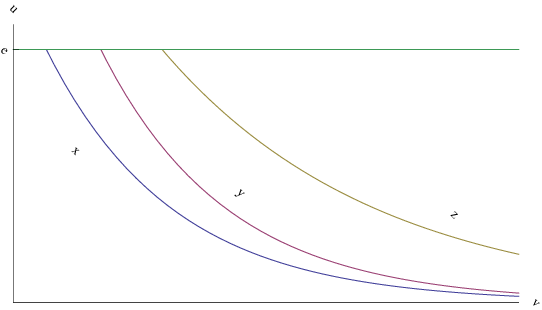}
\end{psfrags}
\end{turn}
\end{center}

\subsection*{The curves $\cg_{\ckr}$}  

We denote by $\cg_{\ckr}:=r^{-1}(\ckr)$ the level sets of the radius function.
These are spacelike curves and consequently may be parameterized by
$$v\mapsto (\uckr(v),v)\ \quad {\rm or}\ \quad u\mapsto (u,\vckr(u)).$$
We choose $\ckrp$ and $\ckrm$ sufficiently close to $r_+$ and $r_-$ with $r_-<\ckrm<\ckrp<r_+$.

\subsection*{The curve $\gam=\gamma_{\ckrm,\beta}$}
 Given $\ckrm$ as before and $\beta>0$, we define $\gamma$ to be the curve parametrized by 
$$ 
u\mapsto\big(u,(1+\beta)\vgr(u)\big)=:(u,v_{\gamma}(u)),\ {\rm for}\ u\in[0,U].
$$ 
This curve probes the region near the Cauchy horizon. The parameter $\beta$ measures the deviation of $\gamma$ with respect to $\cg_{\ckrm}$.

\subsection*{The choice of $\beta$ so that $r$ and $\varpi$ are controlled}
Choose any $\beta$ such that
\begin{equation}\label{beta} 
0<\beta<{\textstyle\frac 12\left(\sqrt{1-8\frac{\partial_r(1-\mu)(r_+,\varpi_0)}{\partial_r(1-\mu)(r_-,\varpi_0)}}-1\right)}.
\end{equation}
Let $0<\eps<\eps_0$. Then, by Lemma~6.1 and Corollary~6.2 of Part~2, there exists $U_\eps>0$ such that
\begin{equation}\label{espana}
r(u,v)\geq r_--\eps\qquad{\rm and}\qquad\varpi(u,v)\leq\varpi_0+\eps
\end{equation}
for $(u,v)\in J^-(\gam)\cap J^+(\cg_{\ckrm})$ and $0<u\leq U_\eps$, provided that the parameters $\ckrp$, $\eps_0$ and $\delta$ are chosen so that
\begin{equation}\label{success} 
\textstyle
\beta<\frac 12\left(
\sqrt{(1+\delta)^2-8
\frac{\mbox{\tiny $(\frac{\ckrp}{r_+})^{\hat\delta^2}$}\min_{r\in[\ckrp,r_+]}\partial_r(1-\mu)(r,\varpi_0)}{\partial_r(1-\mu)(r_--\eps_0,\varpi_0)}}-(1+\delta)
\right).
\end{equation}
Here $\hat\delta$ is a bound for $\bigl|\frac\zeta\nu\bigr|$ in $J^-(\ckrp)$.
Suppose that there exist positive constants $C$ and $s$ such that
$|\zeta_0(u)|\leq Cu^s$. Then, instead of choosing $\beta$ according to~\eqref{beta} we may choose
\begin{equation}\label{beta-s} 
0<\beta<{\textstyle\frac 12\left(\sqrt{1-8\frac{(1+s)\partial_r(1-\mu)(r_+,\varpi_0)}{\partial_r(1-\mu)(r_-,\varpi_0)}}-1\right)}.
\end{equation}
In this case, \eqref{success} should be replaced by
\begin{equation}\label{success-s} 
\textstyle
\beta<\frac 12\left(
\sqrt{(1+\delta)^2-8
\frac{\mbox{\tiny $[(\frac{\ckrp}{r_+})^{\hat\delta^2}$}+s]\min_{r\in[\ckrp,r_+]}\partial_r(1-\mu)(r,\varpi_0)}{\partial_r(1-\mu)(r_--\eps_0,\varpi_0)}}-(1+\delta)
\right).
\end{equation}

\subsection*{A note on the choice of parameters}
The estimates obtained in Part~2, some of which are listed in Appendix C, depend on the choice of a number of parameters, namely $\beta_-$, $\beta_+$, $\ckrm$, $\ckrp$, $\varepsilon_0$ and $U$. As a rule, these estimates hold if $\beta_- < \beta$ and $\beta_+ > \beta$ are chosen sufficiently close to $\beta$, $\ckrm > r_-$ is chosen sufficiently close to $r_-$, $\ckrp < r_+$ is chosen sufficiently close to $r_+$ and $\varepsilon_0>0$, $U > 0$ are chosen sufficiently small. The deviations of these parameters from $\beta$, $r_-$, $r_+$ and $0$, respectively, are controlled by a generic small parameter, typically denoted by $\delta$ or $\varepsilon$, where we will absorb all small quantities (so that $\delta$ or $\varepsilon$ may change from line to line).

%% file: sh_part_three.tex
\section{Mass inflation}\label{bup}

We denote the surface gravities of the Cauchy and black hole horizons in the reference subextremal Reissner-Nordstr\"{o}m black hole by
\begin{equation}\label{ganhamos}
k_- = -\,\frac12 \partial_r ( 1 - \mu ) (  r_-,\varpi_0  ) ,\qquad
k_+ = \frac12 \partial_r ( 1 - \mu ) ( r_+ ,\varpi_0 ),
\end{equation}
and define (see Appendix~\ref{appendix-2})
\begin{equation}\label{def_psi}
\Psi:=\frac{k_-}{k_+} > 1.
\end{equation}
This parameter measures how close the black hole is to being extremal, which corresponds to $\Psi = 1$.

% In this section we present a sufficient condition for
% the renormalized mass $\varpi$ to blow up identically on the Cauchy horizon (Theorem~\ref{inflation}).
% This condition requires
% \begin{equation}\label{ktk}
% \Psi>2
% \end{equation}
% and also that the field $\zeta_0$ satisfies
% $$
% \zeta_0(u)\geq cu^s\ \mbox{for\ some}\/\ 
% 0<s<\frac\Psi 2-1
% $$
% in a neighborhood of the origin. In Appendix~\ref{appendix-2}
% we see how condition~\eqref{ktk} translates into a relationship between $r_-$, $r_+$ and $\Lambda$. 

% We also prove (Theorem~\ref{inflation2}) that if the field $\zeta_0$ satisfies the weaker hypothesis
% $$
% \zeta_0(u)\geq cu^s\ {\rm for\ some}\/\ 0 < \frac\Psi 2-1\leq s<\Psi-1,
% $$
% then either the renormalized mass $\varpi$ or the field $\bigl|\frac\theta\lambda\bigr|$ blow up at the Cauchy horizon.
% Note that under either of the previous hypotheses (that is, $\zeta_0(u)\geq cu^s$ for some $0 < s <\Psi-1$) the Kretschmann scalar blows up
% at the Cauchy horizon (see Remark~\ref{RmkSCC1}).

% We start with a simple result.

We start by presenting a sufficient condition for the renormalized mass $\varpi$ to blow up identically on the Cauchy horizon.

\begin{Thm}[Mass inflation]\label{inflation}
Suppose that $\Psi>2$ and
\begin{equation}\label{zeta-0}
\zeta_0(u)\geq cu^s\ \mbox{for\ some}\/\ 
0<s<\frac\Psi 2-1,
\end{equation}
where $c>0$ and $u\in\left[0,U\right]$. Then 
\begin{equation}\label{infinito}
\varpi(u,\infty)=\infty\ {\it for\ each}\/\ u\in\left]0,U\right].
\end{equation}
\end{Thm}
In Appendix~\ref{appendix-2} we see how the condition $\Psi>2$ translates into a relationship between $r_-$, $r_+$ and $\Lambda$.

The proof of Theorem~\ref{inflation} generally follows the argument on pages~493--497 of~\cite{Dafermos2}, where the $\Lambda=0$ case was studied. Nonetheless, the introduction of a cosmological constant requires a different technical approach, in particular the use of a foliation by the level sets of the radius function; moreover, since later on we will need some of the estimates derived in the proof, we present the relevant details in Appendix~\ref{prova-massa}.

The previous techniques only allow us to explore the subregion of parameter space determined by \eqref{zeta-0}. The rest of this paper will be dedicated to the analysis of the full parameter range. The first result in that direction is
\begin{Thm}[Mass inflation or $\frac\theta\lambda$ unbounded]\label{inflation2}
Suppose that 
$$
\zeta_0(u)\geq cu^s\ {\rm for\ some}\ 0<s<\Psi-1,
$$
where $c>0$ and $u\in\left[0,U\right]$. If $U$ is sufficiently small and $\varpi(U,\infty)<\infty$ then, for each $0<\delta<U$,
$\bigl|\frac\theta\lambda\bigr|(u,v)$ tends to $+\infty$, uniformly for $u\in[\delta,U]$, as $v\nearrow\infty$.
\end{Thm}
\begin{proof}
Suppose that there exists $U>0$ such that $\varpi(U,\infty)<\infty$. Going through the proof of Theorem~\ref{inflation}, we see that Case~3.2 must occur.
So $-\lambda$ must be bounded above in $J^+(\gam)$ as in~\eqref{lambda_cima}. Furthermore, the lower estimate on $\zeta_0$ guarantees
the lower bound~\eqref{theta_below} for $\theta$ in $J^+(\gam)$. Combining~\eqref{lambda_cima} with~\eqref{theta_below}, we obtain
\begin{eqnarray}
\Bigl|\frac{\theta}{\lambda}\Bigr|(u,v)&\geq&\frac{C}{C(u)}\frac{e^{[-(s+1)\partial_r(1-\mu)(r_+,\varpi_0)-\tilde\varepsilon]v}}
{e^{[\partial_r(1-\mu)(r_-,\varpi_0)+\tilde\delta]v}}
\nonumber\\
&=& \frac 1{C(u)}e^{2k_+((\Psi-s-1)-\delta')v},\label{penelope}
\end{eqnarray}
for $(u,v)\in J^+(\gam)$. 
We choose $\ckrp$ sufficiently close to $r_+$, $\ckrm$ sufficiently close to $r_-$, $\beta^+$ and $\beta^-$ sufficiently close to $\beta$,
and $U$ sufficiently small so that $\delta'<\Psi-s-1$ (see the note in the end of Section~\ref{main-one-two}). 
Let $0<\delta<U$. The constant $C(u)$ is bounded above by $C(\delta)$ for $u\in[\delta,U]$.
Then, estimate~\eqref{penelope} shows that 
$\bigl|\frac\theta\lambda\bigr|(u,v)$ tends to $+\infty$, uniformly for $u\in[\delta,U]$, as $v\nearrow\infty$.
\end{proof}

In Remark~\ref{RmkSCC1} we will see that under the hypothesis of the previous theorem the Kretschmann scalar always blows up at the Cauchy horizon, as a consequence of either $\varpi$ or $\frac\theta\lambda$ blowing up.

\section{No mass inflation} 

In this section we will prove that mass inflation does not occur if $\zeta_0$ decays sufficiently fast as $u$ tends to zero. We also control the field $\zeta$
up to the Cauchy horizon.

\begin{Thm}[No mass inflation]
\label{thmNoMass}
Suppose that 
\begin{equation}\label{eq1}
|\zeta_0(u)|\leq cu^s\ {\rm for\ some\ nonnegative}\ s>\frac{7\Psi}9-1,
\end{equation}
where $c>0$ and $u\in\left[0,U\right]$. Then
$$ 
\varpi(u,\infty)<\infty\ {\rm for\ each}\ u\in\left]0,U\right],
$$ 
provided that $U$ is sufficiently small. Furthermore,\/ $\lim_{u\searrow 0}\varpi(u,\infty)=\varpi_0$.
\end{Thm}

Given $\eps_1>0$, define
\begin{eqnarray*}
{\cal D}=
{\cal D}_{\eps_1}=\left\{(u,v)\in J^+(\gamma):\ u\leq U\/\ {\rm and}\   
\int_{\vgam(u)}^v\Bigl| \frac{\theta^2}{\lambda} \Bigr|(u,\tilde v)\,d\tilde v\leq \eps_1\right\}.
\end{eqnarray*}
The set $\cD$ is connected and contains $\gamma$.
Our goal is to prove that, for $U$ small enough, $\cD=J^+(\gamma)$.
This is a consequence of
\begin{Lem}\label{L1} Assume that $\Psi$ and $\zeta_0$ are as in\/ {\rm Theorem~\ref{thmNoMass}}.
Then there exist $\eps_1>0$ and\/ $U>0$\! such that, for $(u,v)\in\cD$,
$$
\int_{\vgam(u)}^v\Bigl| \frac{\theta^2}{\lambda} \Bigr|(u,\tilde v)\,d\tilde v\leq \frac{\eps_1}2.
$$
\end{Lem}
Indeed, for $\eps_1$ and $U$ small enough, Lemma~\ref{L1} implies $\cD$ is open in $J^+(\gamma)$. 
Since $\cD$ is also closed in $J^+(\gamma)$, we conclude that $\cD=J^+(\gamma)$.

\begin{proof}[Proof of\/ {\rm Lemma~\ref{L1}}]
Our goal is to improve the upper estimate~\eqref{up-l} for $-\lambda$ in $\cD$, to obtain a lower estimate for $-\lambda$ in $\cD$, and to obtain
an upper estimate for $|\theta|$ in $\cD$. These will allow us to prove that $\frac{\theta^2}{-\lambda}(u,v)$ decays exponentially in $v$,
from which the conclusion of the lemma will easily follow. Note that the estimates used in this proof will be sharper than needed here, for use in Section~\ref{CH}.

Integrating~\eqref{omega_v} as a linear first order ODE for $\varpi$, starting from $\gamma$, leads to 
\begin{eqnarray}
&& \omega(u,v)=\omega(u,{\vgam(u)})e^{\int_{\vgam(u)}^v\bigl(\frac{\theta^2}{-\lambda}\frac{1}{r}\bigr)(u,\tilde v)\,d\tilde v} \nonumber
\\ 
&& \qquad +
\int_{\vgam(u)}^ve^{\int_{\tilde v}^v\frac{\theta^2}{-\lambda}\frac{1}{r}(u,\bar v)\,d\bar v}\left(\frac{1}{2}\left(1+\frac{e^2}{r^2}
-\frac{\Lambda}{3}r^2\right)\frac{\theta^2}{\lambda}\right)(u,\tilde v)\,d\tilde v. \label{Letizia}
\end{eqnarray}
Let $\tilde\eps>0$.  
If $\eps_1$ and $U$ are sufficiently small, we have from \eqref{espana} and \eqref{Letizia}
\begin{equation}\label{Felipe}
\left|\varpi(u,v)-\varpi_0\right|<{\tilde\eps},
\end{equation}
for  $(u,v)\in\cD$. On the other hand, we have $r<\ckrm$ in $J^+(\gam)$ and,
using~\eqref{r-menos}, we know $\lim_{u\searrow 0}r(u,\infty)=r_-$. Therefore,
$$ 
 -C_1\dMenos \leq \rmu\leq -C_1^{-1} \dMenos\ \ {\rm in}\ \cD,
$$ 
with $C_1>1$. The value of $C_1$ can be chosen as close to one as desired by decreasing $\eps_1$, $\ckrm-r_-$ and $U$.
Henceforth, $C_1$ will denote a constant greater than one, which can be made arbitrarily close to one by a convenient choice of parameters.
$C$ will denote a positive constant.

We start by recalling some estimates over $\gamma$.
Collecting~\eqref{alp_new}, \eqref{thetaLambdaGamma}, \eqref{P} and~\eqref{eq1}, we get
\begin{eqnarray}
\Bigl|\frac\theta\lambda\Bigr|(\ugam(v),v)&\leq&
 C|\ugam(v)|^se^{-2\left(\frac{k_+}{1+\beta}-k_-\beta-\delta\right)v}\nonumber\\
&\leq&Ce^{-2\left(\frac{k_+(s+1)}{1+\beta}-k_-\beta-\delta\right)v}.\label{eq3}
\end{eqnarray}
According to~\eqref{lambda-below} and~\eqref{lambda-above},
\begin{eqnarray}
ce^{-2\left(\frac{k_-\beta}{1+\beta}+\delta\right)v}\leq
-\lambda(\ugam(v),v)&\leq&Ce^{-2\left(\frac{k_-\beta}{1+\beta}-\delta\right)v}.\label{eq4}
\end{eqnarray}
Combining~\eqref{eq3} with~\eqref{eq4},
\begin{eqnarray}
|\theta|(\ugam(v),v)&\leq& Ce^{-2\left(\frac{k_+(s+1)}{1+\beta}-\,\frac{k_-\beta^2}{1+\beta}-\delta\right)v}.\label{bbtheta}
\end{eqnarray}
Finally, according to~\eqref{O} and~\eqref{P},
\begin{equation}\label{v_gamma}
\frac{1+\beta^-}{2k_+}\ln\Bigr(\frac cu\Bigr)\leq \vgam(u)\leq \frac{1+\beta^+}{2k_+}\ln\Bigr(\frac Cu\Bigr).
\end{equation}
Recall that $\beta^-<\beta<\beta^+$ can be chosen arbitrarily close to $\beta$.

We now improve the upper estimate~\eqref{up-l} for $-\lambda$ in $\cD$.
Taking into account~\eqref{integrating_factor}, for $(u,v)\in\cD$, we have
$$ 
e^{\frac 1{r(U,\infty)}\int_{\vgr(\bar u)}^v\bigl[\big|\frac\theta\lambda\bigr||\theta|\bigr](\bar u,\tilde v)\,d\tilde v}\leq 1+\delta,
$$ 
for $\bar u\in[\ugam(v),u]$.
Arguing as in~\eqref{argue2}, 
\begin{eqnarray*}
\int_{\ugam(v)}^u \frac{\nu}{1-\mu}(\tilde u,v)\, d\tilde u
&\geq&
\mbox{\tiny$
\frac {1-\delta}{1+\delta}\frac{\max_{\cg_{\ckrm}}(1-\mu)}{\min_{\cg_{\ckrm}}(1-\mu)}$}(\vgr\bigl(\ugam(v))-\vgr(u)\bigr)
\\
&\geq& 
\frac{C_1^{-1}}{(1+\beta)}\bigl(v-\vgam(u)\bigr)
\\
&\geq& 
\frac{C_1^{-1}q}{(1+\beta)}\bigl(v-\vgam(u)\bigr).
\end{eqnarray*} 
In the last inequality we introduced a parameter $0<q\leq 1$ whose importance will become apparent below.
  Equation~\eqref{lambda_above} together with~\eqref{eq4} now show that,  for $(u,v)\in\cD$,
\begin{eqnarray}
 -\lambda(u,v)&\leq& Ce^{-2\left(\frac{k_-\beta}{1+\beta}-\delta\right)v}e^{-\,\frac{C_1^{-1}q}{(1+\beta)}\dMenos\left(v-\vgam(u)\right)}\nonumber\\
&=&Ce^{-2\left(\frac{k_-(\beta+q)}{1+\beta}-\delta\right)v}e^{\frac{C_1^{-1}q}{(1+\beta)}\dMenos \vgam(u)}.\label{eq6}
\end{eqnarray}
The parameter $q$ makes the second exponential grow slower as $u\searrow 0$ (at the cost of making the first exponential decay slower).
Note that $q=0$ corresponds to~\eqref{up-l}. This is our improved estimate for $-\lambda$ from above.

We now obtain a lower estimate for $-\lambda$ in $\cD$.
Arguing as above, we have 
\begin{eqnarray}
\int_{\ugam(v)}^u \frac{\nu}{1-\mu}(\tilde u,v)\, d\tilde u
&\leq&
C_1(\vgr(\ugam(v))-\vgr(u))\nonumber
\\
&\leq& 
\frac{C_1}{(1+\beta)}v.\label{int-nu}
\end{eqnarray} 
Using~\eqref{lambda_above} together with~\eqref{eq4} once more, for $(u,v)\in\cD$, we obtain
\begin{equation}\label{l-baixo}
-\lambda(u,v)\geq
Ce^{-2\left(\frac{k_-\beta}{1+\beta}+\frac{k_-}{1+\beta}+\delta\right)v}=Ce^{-2\left(k_-+\delta\right)v}.
\end{equation}
This is our estimate for $-\lambda$ from below.

We will now control $\theta$ in $\cD$. Integrating~\eqref{theta_u} and~\eqref{zeta_v} from $\gamma$ leads to
$$\theta(u,v)=\theta(\ugam(v),v)-\int_{\ugam(v)}^u\frac{\zeta\lambda}{r}(\tilde u,v)\,d\tilde u$$
and 
$$\zeta(u,v)=\zeta(u,\vgam(u))-\int_{\vgam(u)}^{v}\frac{\theta\nu}{r}(u,\tilde v)\,d\tilde v.$$
It follows that
\begin{eqnarray}
\label{thetaD}
\theta(u,v)
&=&
\theta(\ugam(v),v)
-\int_{\ugam(v)}^u\zeta(\tilde u,\vgam(\tilde u)) \frac{\lambda}{r}(\tilde u,v)\,d\tilde u \nonumber
\\
&+&
\int_{\ugam(v)}^u\frac{\lambda}{r}(\tilde u,v) \int_{\vgam(\tilde u)}^{v}\frac{\theta\nu}{r}(\tilde u,\tilde v)\,d\tilde v \, d\tilde u\;.
\end{eqnarray}
We fix  $u\leq U$. 
Given $\bar u\in[\ugam(v),u]$, since $r$ is bounded below, from~\eqref{thetaD} we obtain
\begin{eqnarray}
\label{thetaD2}
\theta(\bar u,v)
&\leq&
\left |\theta(\ugam(v),v) \right|
+
C\int_{\ugam(v)}^{\bar u}\left | \zeta(\tilde u,\vgam(\tilde u))\right | \left| \lambda (\tilde u,v)\right|d\tilde u \nonumber
\\
&&+
C\int_{\ugam(v)}^{\bar u} \left| \lambda(\tilde u,v)\right|  
\int_{\vgam(\tilde u)}^{v} \left| \theta (\tilde u,\tilde v)\right|  \left| \nu(\tilde u,\tilde v) \right|d\tilde v \; d\tilde u\nonumber\\
&=:&\left |\theta(\ugam(v),v) \right|+I(\bar u,v)+I\!I(\bar u,v).
\end{eqnarray}
In the next two paragraphs we bound $I$ and $I\!I$.

Collecting~\eqref{alp_new}, \eqref{zetaNuGamma}, \eqref{P} and~\eqref{eq1}, we obtain
\begin{eqnarray}
\Bigl|\frac\zeta\nu\Bigr|(\tilde u,\vgam(\tilde u))&\leq&
 C|\tilde u|^se^{-2\left(\frac{k_+}{1+\beta}-k_-\beta-\delta\right)\vgam(\tilde u)}\nonumber\\
&\leq&Ce^{-2\left(\frac{k_+(s+1)}{1+\beta}-k_-\beta-\delta\right)\vgam(\tilde u)}.\label{eq7}
\end{eqnarray}
Using~\eqref{up-n}, \eqref{eq6} and~\eqref{eq7}, we have
\begin{eqnarray*}
I(\bar u,v)
&\leq&
Ce^{-2\left(\frac{k_-(\beta+q)}{1+\beta}-\delta\right)v}\int_{\ugam(v)}^{\bar u}
e^{-2\left(\frac{k_+(s+1)}{1+\beta}-\,\frac{k_-(\beta^2+\beta+q)}{(1+\beta)}-\delta\right)\vgam(\tilde u)}
\tilde u^p\,d\tilde u.
 \end{eqnarray*}
Here $p=\Psi\beta-1-\delta$.
Using~\eqref{v_gamma}, the integral above can be estimated as
\begin{eqnarray}
&&
\int_{\ugam(v)}^{\bar u}
e^{-2\left(\frac{k_+(s+1)}{1+\beta}-\,\frac{k_-(\beta^2+\beta+q)}{(1+\beta)}-\delta\right)\vgam(\tilde u)}
\tilde u^p\,d\tilde u\label{int1}\\
&&\qquad\leq
\int_{\ugam(v)}^{\bar u}
e^{-2\left(\frac{k_+(s+1)}{1+\beta}-\,\frac{k_-(\beta^2+\beta+q)}{(1+\beta)}-\delta\right)\frac{1+\beta}{2k_+}\ln\left(\frac{C}{\tilde u}\right)}
\tilde u^p\,d\tilde u\nonumber\\
&&\qquad\leq
\int_0^{\bar u} \left( \frac{C}{\tilde u}\right)^{-s-1+\Psi(\beta^2+\beta+q)+\delta} \tilde u^p\, d\tilde u\nonumber\\
&&\qquad\leq
C\bar u^{s+1-\Psi(\beta^2+q)-\delta},\nonumber
\end{eqnarray}
if
\begin{equation}
\label{pBound}
 s>\Psi(\beta^2+q)-1
\end{equation}
and if the parameters are chosen so that $\delta$ is sufficiently small.
Therefore, it is possible to bound $I$ as follows:
\begin{equation}\label{boundI}
I(\bar u,v)\leq C\bar u^{s+1-\Psi(\beta^2+q)}e^{-2\left(\frac{k_-(\beta+q)}{1+\beta}-\delta\right)v}.
\end{equation}

For $v\geq \vgam(u)$, we define
\begin{equation}\label{eqt}
{\cal T}_u(v):= \max_{\tilde u\in[\ugam(v),u]}  \left| \theta(\tilde u,v) \right|.
\end{equation}
We emphasize that the constants $C$ will not depend on $u$.
Using~\eqref{up-n}, \eqref{eq6} and~\eqref{eqt}, we see that 
\begin{eqnarray*}
I\!I(\bar u,v)
&\leq&
Ce^{-2\left(\frac{k_-(\beta+q)}{1+\beta}-\delta\right)v}\int_{\ugam(v)}^{ \bar u}e^{\frac{C_1^{-1}q}{(1+\beta)}\dMenos \vgam(\tilde u)}
 \int_{\vgam(\tilde u)}^{v}  {\cal T}_u(\tilde v) \tilde u^p\, d\tilde v d\tilde u\\
&\leq&
Ce^{-2\left(\frac{k_-(\beta+q)}{1+\beta}-\delta\right)v}\int_{\ugam(v)}^{\bar u}e^{\frac{C_1^{-1}q}{(1+\beta)}\dMenos \vgam(\tilde u)}\tilde u^p\,d\tilde u
 \int_{\vgam(u)}^{v}  {\cal T}_u(\tilde v)\,  d\tilde v,
\end{eqnarray*}
where we used the fact that $\vgam(u) \leq \vgam(\tilde u)$.  
Again, $p=\Psi\beta-1-\delta$.
Using~\eqref{v_gamma}, the first integral above can be estimated as
\begin{eqnarray}
&&
\int_{\ugam(v)}^{\bar u}
 e^{\frac{C_1^{-1}q}{(1+\beta)}\dMenos \vgam(\tilde u)}\tilde u^p\, d\tilde u\label{int2}
\\
&&\qquad\leq
\int_{\ugam(v)}^{\bar u}
 e^{\frac{C_1^{-1}(1+\beta^+)q}{(1+\beta)}\frac{\dMenos}{\dMais} \ln\left(\frac{C}{\tilde u}\right)}\tilde u^p\, d\tilde u\nonumber
\\
&&\qquad\leq
\int_0^{\bar u} \left( \frac{C}{\tilde u}\right)^{C_1^{-1}\frac{1+\beta^+}{1+\beta}\Psi q} \; \tilde u^p d\tilde u\nonumber\\
&&\qquad\leq
C\bar u^{\Psi(\beta-q)-\delta},\nonumber
\end{eqnarray}
if
\begin{equation}
 \label{beta2}
 \beta>q
\end{equation}
and if $\delta$ is sufficiently small.
Therefore it is possible to bound $I\!I$ as follows:
\begin{equation}
\label{II}
I\!I(\bar u,v)\leq C\bar u^{\Psi(\beta-q)-\delta}e^{-2\left(\frac{k_-(\beta+q)}{1+\beta}-\delta\right)v}
 \int_{\vgam(u)}^{v}  {\cal T}_u(\tilde v)\,  d\tilde v.
\end{equation}

For all $v\geq \vgam(u)$,
we estimate ${\cal T}_u(v)$ using~\eqref{bbtheta}, \eqref{thetaD2}, \eqref{boundI} and~\eqref{II}: 
\begin{eqnarray}
{\cal T}_u(v) 
&\leq& Ce^{-2\left(\frac{k_+(s+1)}{1+\beta}-\,\frac{k_-\beta^2}{1+\beta}-\delta\right)v}
   + Cu^{s+1-\Psi(\beta^2+q)-\delta}e^{-2\left(\frac{k_-(\beta+q)}{1+\beta}-\delta\right)v}  \nonumber
\\
&&+
Cu^{\Psi(\beta-q)-\delta}e^{-2\left(\frac{k_-(\beta+q)}{1+\beta}-\delta\right)v}
 \int_{\vgam(u)}^{v}  {\cal T}_u(\tilde v)\,  d\tilde v.\label{TGronwall}
\end{eqnarray}
We claim that
\begin{eqnarray}
{\cal T}_u(v)&\leq& Ce^{-2\left(\frac{k_+(s+1)}{1+\beta}-\,\frac{k_-\beta^2}{1+\beta}-\delta\right)v}\nonumber\\
&&+Cu^{s+1-\Psi(\beta^2+q)-\delta}e^{-2\left(\frac{k_-(\beta+q)}{1+\beta}-\delta\right)v}\label{eq10}
\end{eqnarray}
for
\begin{equation}
\label{beta6}
s>\Psi(\beta^2+\beta+q)-1
\end{equation}
and small $\delta$.

We impose~\eqref{beta6}; it can be checked that considering also the opposite inequality will not lead to an improvement of
the statement of Theorem~\ref{thmNoMass}
(for the choice of parameters that we make below).
\begin{proof}[Proof of the claim]
Inequality~\eqref{TGronwall}, with $u^{\rho(\beta-q)-\delta}$ bounded by a constant, is of the form
\begin{eqnarray*}
 {\cal T}_u(v)&\leq&Ce^{-Av}+Cu^be^{-av}+Ce^{-av}\int_{\vgam(u)}^{v}  {\cal T}_u(\tilde v)\,  d\tilde v,
\end{eqnarray*}
with
\begin{eqnarray}
 A&=&\textstyle 2\left(\frac{k_+(s+1)}{1+\beta}-\,\frac{k_-\beta^2}{1+\beta}-\delta\right),\label{Aa}\\
 a&=&\textstyle 2\left(\frac{k_-(\beta+q)}{1+\beta}-\delta\right),\label{aa}\\
 b&=&\textstyle s+1-\Psi(\beta^2+q)-\delta.\nonumber
\end{eqnarray}
Since we impose~\eqref{beta6},
$A>a>0$, for small $\delta$. Let $\tilde{\cal T}_u(v)=e^{av}{\cal T}_u(v)$. Then
\begin{eqnarray*}
 \tilde{\cal T}_u(v)&\leq&Ce^{-(A-a)v}+Cu^b+C\int_{\vgam(u)}^{v}  e^{-a\tilde v}\tilde{\cal T}_u(\tilde v)\,  d\tilde v.
\end{eqnarray*}
Applying Gronwall's inequality, we get
\begin{eqnarray*}
  \tilde{\cal T}_u(v)&\leq&Ce^{-(A-a)v}+Cu^b+C\int_{\vgam(u)}^{v}\bigl(e^{-A\tilde v}+u^be^{-a\tilde v}\bigr)\,d\tilde v\\
  &\leq&Ce^{-(A-a)v}+Cu^b+Ce^{-A\vgam(u)}+Cu^be^{-a\vgam(u)}\\
  &\leq&Ce^{-(A-a)v}+Cu^b+Cu^{A\frac{1+\beta}{2k_+}-\delta}+Cu^bu^{a\frac{1+\beta}{2k_+}-\delta}\\
  &\leq&Ce^{-(A-a)v}+Cu^b.
\end{eqnarray*}
To estimate $e^{-\vgam(u)}$ we used~\eqref{O}. We also used $A\frac{1+\beta}{2k_+}-\delta=s+1-\Psi\beta^2-\delta>b=s+1-\Psi(\beta^2+q)-\delta$,
for small $\delta$.
\end{proof}

\bigskip

Obviously, for $(u,v)\in\cD$ we have
\begin{equation}\label{theta-f}
|\theta(u,v)|\leq{\cal T}_u(v).
\end{equation}
Using~\eqref{l-baixo} and~\eqref{eq10}, we obtain
\begin{eqnarray}
\Bigl|\frac{\theta^2}{\lambda}\Bigr|(u,v)
&\leq& Ce^{-2\left(\frac{2k_+(s+1)}{1+\beta}-\,\frac{2k_-\beta^2}{1+\beta}-k_--\delta\right)v}\label{eq11}\\
&&+ Cu^{2\bigl(s+1-\Psi(\beta^2+q)-\delta\bigr)}e^{-2\left(\frac{k_-(\beta+2q-1)}{1+\beta}-\delta\right)v}.\label{eq12}
\end{eqnarray}
The exponent in~\eqref{eq11} can be made negative if
\begin{equation}
\label{beta5}
s>\Psi\Bigl(\beta^2+\frac \beta 2+\frac 12\Bigr)-1
\end{equation}
and the second exponent in~\eqref{eq12} can be made negative if
\begin{equation}
\label{beta8}
\beta>1-2q.
\end{equation}
In what follows we will not exploit the smallness of $u^{2b}$.

Below we will characterize a choice of parameters for which we have 
\begin{equation}\label{pescadinha}
\Bigl|\frac{\theta^2}{\lambda}\Bigr|(u,v)\leq C e^{-\Delta v},\ \ 
\end{equation}
for $(u,v)\in\cD$,
with $\Delta>0$.
However, before we do that, we note that estimate~\eqref{pescadinha} wraps up the bootstrap argument. Indeed,
as $\lim_{u\searrow 0}\vgam(u)=+\infty$,
we can choose $U$ such that 
$$\int_{\vgam(u)}^v\Bigl| \frac{\theta^2}{\lambda} \Bigr|(u,\tilde v)\,d\tilde v<\frac{\eps_1}{2},$$
 for $(u,v)\in\cD$. 

We now bring together the conditions that we must satisfy in order for the above argument to work, and we choose our parameters.
The number
$\beta$ is bounded above by~\eqref{beta-s} and bounded below by~\eqref{beta2} and~\eqref{beta8};
in addition, $s$ is bounded below  
by~\eqref{pBound}, \eqref{beta6} and~\eqref{beta5}.
In fact, the restrictions on $s$ can be stated in a simpler form:
inequality~\eqref{beta6} is stricter 
than~\eqref{pBound};
inequality~\eqref{beta8} implies that~\eqref{beta6} is stricter than~\eqref{beta5}.
So, all the restrictions on $s$ amount to saying that $s$ is bounded below by~\eqref{beta6}.

We now select the parameters $q$ and $\beta$.
The minimum of the maximum of the lower bounds for $\beta$ in~\eqref{beta2} and~\eqref{beta8} is obtained for $q=\frac 13$.
This is our choice of $q$. Inequality~\eqref{beta-s} can be satisfied when $s>\frac{2\Psi}9-1$ because
$$ 
{\frac 13}=\frac 12\left({\textstyle\sqrt{1+8\frac 29}}-1\right)<\beta<\frac 12\left({\textstyle\sqrt{1+\frac {8(1+s)}\Psi}}-1\right).
$$ 
For~\eqref{beta6} to be satisfied we impose $\frac{7\Psi}9-1<s$ because
$$ 
\left.\left[\Psi(\beta^2+\beta+q)-1\right]\right|_{\stackrel{\beta=\frac 13}{\mbox{\tiny{$q=\frac 13$}}}}=\frac{7\Psi}9-1<\left.\left[\Psi(\beta^2+\beta+q)-1\right]\right|_{q=\frac 13}<s,
$$
Obviously, $\frac{2\Psi}9-1<\frac{7\Psi}9-1$.
Therefore, if  $s>\frac{7\Psi}9-1$ and we choose $\beta=\frac 13+\eps$, with $\eps>0$ sufficiently small, both~\eqref{beta-s} and~\eqref{beta6} are satisfied.

Therefore, our parameters will be chosen in the following way.
Suppose that 
we are given initial data $\zeta_0$ satisfying~\eqref{eq1}. 
We choose $\beta>\frac 13$ (so that~\eqref{beta2} and~\eqref{beta8} hold with $q=\frac 13$) and such that~\eqref{beta-s} and~\eqref{beta6} hold.
When $$(\ckrp,\ckrm,\beta^+,\beta^-,\eps_0,\eps_1,U)\to(r_+,r_-,\beta,\beta,0,0,0),$$
the parameters $\delta$ above all converge to 0 (at the cost of increasing the constants $C$). So, we may
choose $\ckrp$ sufficiently close to $r_+$, $\ckrm$ sufficiently close to $r_-$, $\beta^+$ and $\beta^-$ sufficiently close to $\beta$,
and $\eps_0$, $\eps_1$ and $U$ sufficiently small 
so that~\eqref{success-s} holds,
the exponent 
in~\eqref{eq11} and the second exponent in~\eqref{eq12} are negative, the integrals~\eqref{int1} and~\eqref{int2} converge, and, finally,
such that the numbers $A$ in~\eqref{Aa} and $a$ in~\eqref{aa} satisfy $A>a$. This will guarantee~\eqref{pescadinha}, for a certain
positive $\Delta$ and (a maybe very large but finite value of) $C$.

The proof of Lemma~\ref{L1} is complete.
\end{proof}

Since \eqref{Felipe} holds in $\cD$, the fact that $\cD=J^+(\gamma)$, established as a consequence of Lemma~\ref{L1}, implies Theorem~\ref{thmNoMass}.

We finish this section by controlling the field $\zeta$ in the following result.
\begin{Lem}\label{zeta-control}
Suppose that 
$$ 
|\zeta_0(u)|\leq cu^s\ {\rm for\ some\ nonnegative}\ s>\frac{7\Psi}9-1.
$$ 
Then there exists a constant $C>0$ such that
\begin{equation}\label{eq24}
|\zeta(u,v)|\leq Cu^{s-\Psi\beta^2-\delta},
\end{equation}
for $(u,v)\in J^+(\gam)$, where $\delta>0$ can be chosen arbitrarily close to zero, provided that $U$ is sufficiently small.
\end{Lem}
\begin{proof}
Integrating~\eqref{zeta_v}, we have
\begin{equation}
\zeta(u,v)=\zeta(u,\vgam(u))-\int_{\vgam(u)}^v\frac{\theta\nu}{r}(u,\tilde v)\,d\tilde v.
\label{eq23}\end{equation}
Collecting~\eqref{up-n}, \eqref{v_gamma} and \eqref{eq7}, we get
\begin{eqnarray}
|\zeta(u,\vgam(u))|&\leq& Ce^{-2\left(\frac{k_+(s+1)}{1+\beta}-k_-\beta-\delta\right)\vgam(u)}u^{\Psi\beta-1-\delta}\nonumber\\
&\leq&C u^{s-\Psi\beta^2-\delta}.\label{eq21}
\end{eqnarray}
On the other hand, from~\eqref{up-n}, \eqref{v_gamma}, \eqref{eq10} and~\eqref{theta-f}, we obtain
\begin{eqnarray}
\int_{\vgam(u)}^v
\frac{|\theta\nu|}{r}(u,\tilde v)\,d\tilde v&\leq&
Ce^{-2\left(\frac{k_+(s+1)}{1+\beta}-\,\frac{k_-\beta^2}{1+\beta}-\delta\right)\vgam(u)}u^{\Psi\beta-1-\delta}\nonumber\\
&&+Cu^{s+1-\Psi(\beta^2+q)-\delta}
e^{-2\left(\frac{k_-(\beta+q)}{1+\beta}-\delta\right)\vgam(u)}u^{\Psi\beta-1-\delta}\nonumber\\
&\leq&C u^{s-\Psi\beta^2+\Psi\beta-\delta}+ Cu^{s-\Psi\beta^2+2\Psi\beta-\delta}\nonumber\\
&\leq&C u^{s-\Psi\beta^2+\Psi\beta-\delta}.\label{eq22}
\end{eqnarray}
Using~\eqref{eq21} and~\eqref{eq22} in~\eqref{eq23}, we obtain~\eqref{eq24}.
\end{proof}

\section{Extensions of the metric beyond the Cauchy horizon} 

% Interestingly, these lead to inequivalent $C^2$ structures for the extended manifolds, a fact which is reflected on the behavior of the Christoffel symbols: when the initial data satisfies 
% $$
% c_2u^{s_2}\leq\zeta_0(u)\leq c_1u^{s_1}\ {\rm for\ some}\ \frac{7\Psi}9-1<s_1\leq s_2<\Psi-1,
% $$
% we prove that one of the Christoffel symbols blows up identically at the Cauchy horizon in the $\lambda=-1$ case, but not in the $\kappa=1$ case. Moreover, in the second case the Christoffel symbols of the metric are in $L^2_{{\rm loc}}$ 
% and the field $\phi$ is in $H^1_{{\rm loc}}$ for the whole range of initial data where there is no mass inflation (Corollary~\ref{cc-fails}). Therefore, the Christodoulou-Chru\'sciel inextendibility criterion for strong cosmic censorship does not hold.

In this section we assume that the field $\zeta_0$ satisfies
$$ 
|\zeta_0(u)|\leq cu^s\ {\rm for\ some\ nonnegative}\ s>\frac{7\Psi}9-1,
$$ 
so that there is no mass inflation, and we examine the possibility of extending the metric beyond the Cauchy horizon.
We regard the $(u,v)$ plane, the domain of our first order system, as a $C^2$ manifold.
Since the Cauchy horizon corresponds to $v=\infty$, we must change this coordinate to one with a finite range.
There are two natural choices to do so: either resorting to the radius function along the outgoing null ray $u=U$ for the new coordinate (i.e. choosing $\lambda=-1$ on $u=U$), or setting $\kappa=1$ on the null ray $u=U$ (as was done for the initial data along the event horizon).

In the first coordinate system, $v$ is then replaced by $\tilde v=r(U,0)-r(U,v)$.
This is the coordinate system that we will later use in Section~\ref{CH};
it transforms the domain $[0,\infty[$ of $v$ into a bounded interval for $\tilde{v}$,
even when the field $\zeta_0$ is identically zero.
In the second coordinate system, which has finite range only when $\zeta_0$ is not identically
zero, $v$ is replaced by $\hat{v}:=\int_0^v\kappa(U,\bar{v})\,d\bar{v}$. 

Both maps $v\mapsto\tilde{v}$ and $v\mapsto\hat{v}$ are $C^2([0,\infty[)$
and have a non-zero derivative, the first one with range $[0,V[$, say, and the
second one with range $[0,\hat{V}[$. So the map $\hat{v}\mapsto\tilde{v}$
is $C^2$. This map extends to a $C^1$ map from the interval $[0,\hat{V}]$ to
the interval $[0,V]$ (see~\eqref{tildebar}).
By Remark~\ref{non-equivalent}, 
the two coordinate systems $(u,\tilde{v})$ and $(u,\hat{v})$
are not equivalent (as $C^2$ coordinate systems) when $\bigl|\frac{\theta}{\lambda}\bigr|$ is unbounded along $u=U$.

In both coordinate systems we can extend the metric continuously to the
Cauchy horizon, and consequently beyond the Cauchy horizon, with the second mixed 
derivatives of $r$ continuous. In the coordinate system $(u,\hat{v})$ this can 
be done so that the Christoffel symbols are in $L^2_{{\rm loc}}$ 
and the field $\phi$ is in $H^1_{{\rm loc}}$. Therefore, the Christodoulou-Chru\'sciel
inextendibility criterion for strong cosmic censorship does not hold.

\subsection{Coordinates with $v$ replaced by $\tilde v=r(U,0)-r(U,v)$}

If there exists a positive sequence $(u_n)$ converging to $0$ such that $\zeta_0(u_n)\neq 0$
then we choose $U$ such that $(1-\mu)(U,\infty)<0$. In the proof of Lemma~\ref{l-kappa} we showed that such a $U$ exists;
in Proposition~\ref{fechado-0} we will see that under the present assumptions $(1-\mu)(U,\infty)<0$ for any $U>0$, so that actually any choice of $U$
will do. If $\zeta_0$ vanishes in a right neighborhood of the origin then the solution is simply Reissner-Nordstr\"{o}m and we can choose any $U$.
We define $f:[0,\infty[\to\R$, by
\begin{equation}\label{smooth_coordinate}
f(v)=
r(U,0)-r(U,v),
\end{equation}
so that $$f'(v)=-\lambda(U,v),$$
and set
$$
V=f(\infty)=r(U,0)-r(U,\infty).
$$
We will change the $v$ coordinate to $$\tilde v=f(v).$$

The functions $\nu_0$, $\kappa_0$ and $\lambda_0$ for the original characteristc initial value problem
(equal to $-1$, $1$ and $0$, respectively, see \eqref{RN1} and \eqref{RN2}) satisfy hypothesis (h4) (see Section~\ref{main-one-two}). 
By Lemma~\ref{regular}, the function $r$ is $C^2$.
Moreover, $\lambda(U,\,\cdot\,)<0$. Therefore, the change of coordinates of the previous paragraph is admissible (that is, $C^2$).

We denote by $\tilde r$ the function $r$ written in the new coordinates, i.e.
$$
\tilde r(u,\tilde v)=\tilde r(u,f(v))=r(u,v).
$$
We let $\tilde\lambda=\partial_{\tilde v}\tilde r$ and $\tilde\nu=\partial_{u}\tilde r$, whence
$$
\tilde\lambda(u,\tilde v)=\frac{\lambda(u,v)}{f'(v)}
$$
and
$$
\tilde\nu(u,\tilde v)=\nu(u,v).
$$
In particular, $$\tilde\lambda(U,\tilde v)\equiv -1.$$
Similarly, we define
$$
-\tilde\Omega^2(u,\tilde v)\,dud\tilde v=-\tilde\Omega^2(u,f(v))f'(v)\,dudv=-\Omega^2(u,v)\,dudv.
$$
From~\eqref{bar_rafaeli} we then have
$$
\tilde\varpi(u,\tilde v)=\varpi(u,v),
$$
and from~\eqref{kappa_0}
$$
\tilde\kappa(u,\tilde v)=\frac{\kappa(u,v)}{f'(v)}.
$$
Finally, we also denote by $\tilde\phi$ the function $\phi$ written in the new coordinates,
$$
\tilde\phi(u,\tilde v)=\tilde\phi(u,f(v))=\phi(u,v),
$$
and from~\eqref{theta} and~\eqref{zeta}
$$
\tilde\theta(u,\tilde v)=\frac{\theta(u,v)}{f'(v)},\qquad \tilde\zeta(u,\tilde v)=\zeta(u,v).
$$
\begin{Rmk}
It is obvious that the functions  $\tilde r$, $\tilde\nu$, $\tilde\lambda$, $\tilde\varpi$, $\tilde\theta$, $\tilde\zeta$ and $\tilde\kappa$ 
satisfy the first order system\/~\eqref{r_u}$-$\eqref{kappa_at_u}, with respect to the new coordinates $(u,\tilde v)$.
\end{Rmk}

% As observed in the first paragraph of Section~6 of Part~1, ``Derivation of the Einstein equations from the first order system'',
% the fact that $\tilde\lambda(U,\,\cdot\,)$ is $C^1$ implies
% $\tilde\kappa(U,\,\cdot\,)$ is $C^1$. So, clearly we have
% \begin{Rmk}\label{our_choice}
% The regularity hypothesis\/ {\rm (h4)} holds for our data on\/ $[0,U]\times\{0\}\cup\{U\}\times[0,V[$.
% \end{Rmk}

\vspace{4mm}

\begin{Prop}\label{fechado-0}
Suppose that 
$$ 
|\zeta_0(u)|\leq cu^s\ {\rm for\ some\ nonnegative}\ s>\frac{7\Psi}9-1,
$$ 
where $c>0$ and $u\in\left[0,U\right]$.
If $U>0$ is sufficiently small then, for all\/ $0<\delta<U$, the functions $\tilde r$, $\tilde\nu$, $\tilde\lambda$, $\tilde\varpi$,
$\tilde\zeta$ and $\tilde\kappa$ (but not necessarily $\tilde{\theta}$)
admit continuous extensions to the closed rectangle\/ $[\delta,U]\times[0,V]$.
Equations~\eqref{r_u} to \eqref{omega_u}, \eqref{kappa_u}, and~\eqref{kappa_at_u} are satisfied on this set.
Finally, $\widetilde{(1-\mu)}(u,V)$ is negative for $u>0$, unless there exists a right neighborhood of the origin where $\zeta_0$ vanishes.
\end{Prop}
\begin{proof}
If $\zeta_0$ vanishes in a right neighborhood of the origin, then the conclusion is immediate since the functions are obtained from the Reissner-Nordstr\"om solution.

Assume that there exists a positive sequence $(u_n)$ converging to 0 such that $\zeta_0(u_n)\neq 0$.
We fix $0<\delta<U$, and proceed in three steps.

{\em Step 1}. We prove that our functions $\tilde r$, $\tilde\nu$, $\tilde\lambda$, $\tilde\varpi$, $\tilde\zeta$ and $\tilde\kappa$ converge uniformly as functions of $u \in [\delta,U]$ as $\tilde v \to V$.
The convergence of $\tilde r(\,\cdot\,,\tilde v)$ to $\tilde r(\,\cdot\,,V)$ is uniform on $[\delta,U]$ because
$$\int_{\tilde v}^{V}|\tilde\lambda|(u,\bar v)\,d\bar v=\int_{f^{-1}(\tilde v)}^{\infty}|\lambda|(u,\bar v)\,d\bar v\to 0$$
as $\tilde v\nearrow V$ (by~\eqref{up-l}). 

In view of \eqref{omega_v} and \eqref{Felipe}, the convergence of $\tilde\varpi(\,\cdot\,,\tilde v)$ to $\tilde\varpi(\,\cdot\,,V)$ is also uniform on $[\delta,U]$, because
$$\int_{\tilde v}^{V}\Bigl|\frac{\tilde\theta^2}{\tilde\lambda}\Bigr|(u,\bar v)\,d\bar v=\int_{f^{-1}(\tilde v)}^{\infty}\Bigl|\frac{\theta^2}{\lambda}\Bigr|(u,\bar v)\,d\bar v\to 0$$
as $\tilde v\nearrow V$ (by~\eqref{eq11}-\eqref{eq12}).

For $u\in[\delta,U]$, using~\eqref{integral-k},
\begin{equation}\label{uni-nu}
 \int_{\tilde v}^V\tilde\kappa(u,\bar v)\,d\bar v\leq \int_{\tilde v}^V\tilde\kappa(\delta,\bar v)\,d\bar v\to 0,\ {\rm as}\ \tilde v\to\infty.
\end{equation}
Integrating~\eqref{nu_v}, for $\tilde v\leq\tilde V<V$,
$$
\tilde\nu(u,\tilde v)-\tilde\nu(u,\tilde V)=\tilde\nu(u,\tilde v)\left(1-e^{\int_{\tilde v}^{\tilde V}[\tilde\kappa\widetilde{\partial_r(1-\mu)}](u,\bar v)\,d\bar v}\right).
$$
Using~\eqref{up-n}, \eqref{uni-nu}, $r \geq r(U,\infty)>0$ and $\varpi \leq \varpi(U,\infty)<\infty$, implying that $\partial_r(1-\mu)$ is bounded, we conclude that we may define $\tilde\nu(\,\cdot\,,V)$. Also, letting $\tilde V\nearrow V$,
the restriction of $\tilde\nu(\,\cdot\,,\tilde v)$ to $[\delta,U]$ converges uniformly to $\tilde\nu(\,\cdot\,,V)$ as $\tilde v\nearrow V$. 
Integrating~\eqref{nu_v} between $\tilde v$ and $V$, we conclude that
\begin{equation}\label{nu_negativo}
\tilde\nu(u,V)<0
\end{equation} 
for each $u>0$.

Integrating~\eqref{zeta_v},
$$ 
\tilde\zeta(u,\tilde V)=\tilde\zeta(u,\tilde v)-\int_{\tilde v}^{\tilde V}\frac{\tilde\theta\tilde\nu}{\tilde r}(u,\bar v)\,d\bar v.
$$ 
We use~\eqref{up-n} and
$$\int_{\tilde v}^{V}|\tilde\theta|(u,\bar v)\,d\bar v=\int_{f^{-1}(\tilde v)}^{\infty}|\theta|(u,\bar v)\,d\bar v\to 0$$
as $\tilde v\nearrow V$ (by~\eqref{eq10} and~\eqref{theta-f}). Note that the last convergence is uniform for $u\in[\delta,U]$.
Arguing as in the previous paragraph, we may define $\tilde\zeta(\,\cdot\,,V)$ 
as the uniform limit of $\tilde\zeta(\,\cdot\,,\tilde v)$ when $\tilde v\nearrow V$.

From $\tilde\kappa(U,\tilde v)=\frac{-1}{\widetilde{(1-\mu)}(U,\tilde v)}$ and~\eqref{kappa_u}, we get
\begin{equation}\label{acs}
\tilde\kappa(u,\tilde v)=\frac{-1}{\widetilde{(1-\mu)}(U,\tilde v)}e^{-\int_u^U\bigl(\frac{\tilde\zeta^2}{\tilde r\tilde\nu}\bigr)(\bar u,\tilde v)\,d\bar u}.
\end{equation}
Using $\widetilde{(1-\mu)}(U,V)=(1-\mu)(U,\infty)<0$ (recall the beginning of the current subsection), the uniform convergence of $\tilde r$, $\tilde\nu$ and $\tilde\zeta$ as $\tilde v\nearrow V$,
and the fact that $\tilde r$ and $\tilde\nu$ are bounded away from zero,
we see that we may define $\tilde\kappa(\,\cdot\,,V)$.
Furthermore, since we already proved uniform convergence of $\tilde r$, $\tilde\varpi$, $\tilde\nu$ and $\tilde\zeta$, it is clear that $\tilde\kappa(\,\cdot\,,V)$ is the uniform limit of $\tilde\kappa(\,\cdot\,,\tilde v)$ when $\tilde v\nearrow V$.
We have
$$
\tilde\kappa(u,V)\geq \tilde\kappa(U,V)=\frac{-1}{\widetilde{(1-\mu)}(U,V)}>0 
$$
for $u\in[\delta,U]$.

The function $\tilde\lambda$ clearly extends to a continuous function on $[\delta,U]\times[0,V]$ since
$\tilde\lambda=\tilde\kappa\widetilde{(1-\mu)}$.

{\em Step 2}. 
The functions $\tilde r$, $\tilde\nu$, $\tilde\lambda$, $\tilde\varpi$, $\tilde\zeta$ and $\tilde\kappa$ 
are continuous in the closed rectangle $[\delta,U]\times[0,V]$. Indeed, let $\tilde h$ denote one of these functions. We know $\tilde h(\,\cdot\,,V)$
is continuous because it is the uniform limit of continuous functions. Let $u\in[\delta,U]$ and $\eps>0$.
There exists $\tilde\delta>0$ such that $|\bar u-u|<\tilde\delta$ implies $|\tilde h(\bar u,V)-\tilde h(u,V)|<\frac\eps 2$. Furthermore, again by uniform convergence, there exists $\hat\delta>0$ such that
$|\tilde v-V|<\hat\delta$ implies $|\tilde h(\bar u,\tilde v)-\tilde h(\bar u,V)|<\frac\eps 2$ for all $\bar u\in[\delta,U]$. So, if $|\bar u-u|<\tilde\delta$ and $|\tilde v-V|<\hat\delta$, then
$|\tilde h(\bar u,\tilde v)-\tilde h(u,V)|<\eps$. This proves continuity of $\tilde h$ at $(u,V)$.

{\em Step 3}. 
It is clear that the system~\eqref{r_u} to~\eqref{kappa_u}, except~\eqref{omega_v}, \eqref{theta_u} and~\eqref{zeta_v}, is satisfied also on the segment $[\delta,U]\times\{V\}$.
Indeed, to obtain the equations that involve the derivative with respect to $u$, we use the fact that if $\tilde h(\,\cdot\,,v_n)$ converges uniformly to
$\tilde h(\,\cdot\,,V)$ and $\partial_u\tilde h(\,\cdot\,,v_n)$ converges uniformly to $\hat h(\,\cdot\,,V)$ as $v_n\nearrow V$ then $\partial_u\tilde h(\,\cdot\,,V)$ exists and is equal to
$\hat h(\,\cdot\,,V)$.

On the other hand, to obtain the equations that involve the derivative with respect to $v$, we write these equations in integrated form, say from $0$ to $\tilde v_n$,
and let $\tilde v_n\nearrow V$. From the (trivial) continuity of the indefinite integral of a continuous function and the Fundamental Theorem of Calculus, we deduce that
the equations are valid at $V$.

Obviously, \eqref{kappa_at_u} is satisfied  on the segment $[\delta,U]\times\{V\}$.

Finally, taking into account
$$
\frac{\nu(u,\infty)}{(1-\mu)(u,\infty)}\leq\frac{\nu(u,0)}{(1-\mu)(u,0)}<\infty
$$ 
(from~\eqref{ray_v_bis})
and that $\tilde\nu$ is negative on $[\delta,U]\times\{V\}$ (see \eqref{nu_negativo}), we conclude that $\widetilde{(1-\mu)}(u,V)$ is uniformly
bounded above by a negative constant on $[\delta,U]$.
\end{proof}

\vspace{4mm}

\noindent {\em The metric and the field.}\/ Recall that the reason to study our first order system is that its solutions 
allow the construction of spherically symmetric Lorentzian manifolds $({\cal M}, g)$ and fields $\tilde\phi$ which solve the Einstein equations.
Here ${\cal M} = {\cal Q}\times{\mathbb S}^2$, where ${\cal Q}$ admits the global null coordinate system
$(u,\tilde v)$ defined on $[0,U]\times[0,V]\setminus\{(0,V)\}$, and the metric is
$$ 
g = -\tilde\Omega^2(u,\tilde v)\,dud\tilde v+\tilde r^2(u,\tilde v)\,\sigma_{{\mathbb S}^2},
$$ 
with $\tilde\Omega^2=-4\tilde\nu\tilde\kappa$. We give ${\cal M}$ the structure of a $C^2$ manifold, i.e.\ we only allow $C^2$ changes of
coordinates. Although Proposition~\ref{fechado-0} guarantees that $r$ is $C^1$ on $[0,U]\times[0,V]\setminus\{(0,V)\}$, the regularity of the metric is
no better than $C^0$, since, as will become apparent in the proof of Proposition~\ref{meio}, $\partial_{\tilde v} \tilde\kappa$ may blow up on the Cauchy horizon.
This allows for $C^0$ extensions of the metric beyond the Cauchy horizon, by a similar construction as the one that will be used below for the coordinate system $(u,\hat{v})$.

The field $\tilde\phi$ is determined, after prescribing $\tilde\phi(0,0)$, by integrating~\eqref{theta} and~\eqref{zeta}. According to~\cite[Proposition~13.2]{Dafermos2} (with the choice $u_1=v_1=0$), $\int_0^v|\theta|(u,\bar v)\,d\bar v+\int_0^u|\zeta|(\bar u,v)\,d\bar u\leq\underline{C}$ (note that this result depends only on equations \eqref{theta_u} and \eqref{zeta_v}, and so does not depend on the presence of $\Lambda$). So, $\tilde\phi$ is well defined, bounded and continuous, with continuous partial derivative with respect to $u$ in $[0,U]\times[0,V]\setminus\{(0,V)\}$. 

The nonvanishing Christoffel symbols of the metric on ${\cal M}$ are
\begin{eqnarray} 
\tilde \Gamma^C_{AB},&&\nonumber\\
\tilde \Gamma^u_{AB}
&=&
2\tilde \Omega^{-2}\tilde r\tilde\lambda\,\sigma_{AB}=-\tilde r\frac{\widetilde{1-\mu}}{2\tilde\nu}\,\sigma_{AB},\nonumber
\\
\tilde \Gamma^{\tilde v}_{AB}
&=&
2\tilde \Omega^{-2}\tilde r\tilde\nu\,\sigma_{AB}=-\tilde r\frac{\widetilde{1-\mu}}{2\tilde\lambda}\,\sigma_{AB},\nonumber
\\
\tilde \Gamma^A_{B\tilde u}&=&\tilde\nu\tilde r^{-1}\delta^A_B,\nonumber\\
\tilde \Gamma^A_{B\tilde v}&=&\tilde\lambda\tilde r^{-1}\delta^A_B,\nonumber\\
\tilde \Gamma^u_{uu}
&=&
\tilde \Omega^{-2}\partial_u(\tilde \Omega^2)\
=\frac{\partial_u \tilde\nu}{\tilde\nu}+\frac{\partial_u \tilde\kappa}{\tilde \kappa},\nonumber
\\
\tilde \Gamma^{\tilde v}_{\tilde v\tilde v}
&=&
\tilde \Omega^{-2}\partial_{\tilde v}(\tilde \Omega^2)\
=\frac{\partial_{\tilde v} \tilde\nu}{\tilde\nu}+\frac{\partial_{\tilde v} \tilde\kappa}{\tilde \kappa}\label{chris}
\end{eqnarray}
(see~\cite[Appendix~A]{DafermosExtension}).

\begin{Prop}\label{meio}
Suppose that  
$$ 
c_2u^{s_2}\leq\zeta_0(u)\leq c_1u^{s_1}\ {\rm for\ some}\ \frac{7\Psi}9-1<s_1\leq s_2<\Psi-1,
$$ 
where $c_1,c_2>0$ and $u\in\left[0,U\right]$. For any $\delta>0$ the field $\bigl|\frac{\tilde\theta}{\tilde\lambda}\bigr|(u,\tilde v)$ tends to $+\infty$ as $\tilde v\nearrow V$, uniformly for\/ $u\in[\delta,U]$, provided that $U$ is sufficiently small.
For all $u\in[\delta,U]$, with one possible exception, 
$\tilde \Gamma^{\tilde v}_{\tilde v\tilde v}(u,\tilde v)$ is unbounded  
as $\tilde v\nearrow V$.
Moreover,\/ $\tilde \Gamma^{\tilde v}_{\tilde v\tilde v}(U,\tilde v)$ tends to $-\infty$  
as $\tilde v\nearrow V$.
\end{Prop}
\begin{proof}
The upper bound on $\zeta_0$ and Theorem~\ref{thmNoMass}
imply that $\varpi(u,\infty)<\infty$ for each $0<u\leq U$, provided that $U$ is sufficiently small.
Fix $0<\delta<U$. 
Using the lower bound on $\zeta_0$ 
together with 
Theorem~\ref{inflation2}, we know that $\bigl|\frac{\tilde\theta}{\tilde\lambda}\bigr|(u,\tilde v)$ tends to $+\infty$, uniformly for\/ $u\in[\delta,U]$, as $\tilde v\nearrow V$.
In particular,
\begin{equation}\label{jkowski}
|\tilde\theta(U,\tilde v)|\to+\infty,\ {\rm as}\  \tilde v\nearrow V.
\end{equation}

Suppose, by contradiction, that there exist $u_1<u_2$ in
$[\delta,U]$ for which $\tilde \Gamma^{\tilde v}_{\tilde v\tilde v}(u_1,\,\cdot\,)$
and $\tilde \Gamma^{\tilde v}_{\tilde v\tilde v}(u_2,\,\cdot\,)$
are bounded.
The wave equation \eqref{wave_Omega} for $\tilde{\Omega}^2$ can be written as
\begin{eqnarray*}
\partial_u\partial_{\tilde v}\log \tilde \Omega^2&=&
-\,\frac{2\tilde \theta \tilde \zeta}{\tilde r^2}
+\frac{4\tilde\kappa\tilde \nu e^2}{\tilde r^4}
-\frac{2\tilde\kappa\tilde \nu}{\tilde r^2}
+\frac{2\tilde \lambda\tilde \nu}{\tilde r^2}.
\end{eqnarray*}
Using~\eqref{chris} we know that 
$\tilde \Gamma^{\tilde v}_{\tilde v\tilde v}
=
\partial_{\tilde v}\ln(\tilde \Omega^2)$. Thus, there exists a constant $C>0$ such that,
for $\delta\leq u_1<u_2\leq U$ and $\tilde v\in[0,V[$,
\begin{equation} \label{difere}
\int_{u_1}^{u_2}\frac{2\tilde \theta\tilde \zeta}{\tilde r^2}
(\bar u,\tilde v)\,d\bar u\leq
\tilde \Gamma^{\tilde v}_{\tilde v\tilde v}(u_1,\tilde v)-
\tilde \Gamma^{\tilde v}_{\tilde v\tilde v}(u_2,\tilde v)+C.
\end{equation}
According to Proposition~\ref{fechado-0} there exists a positive constant
$c_\delta$ such that 
\begin{equation} \label{lambdapositive}
\tilde\lambda(u,\tilde v)\leq -c_\delta<0\ \ \ {\rm for}\ (u,\tilde v)\in[\delta,U]\times[0,V]
\end{equation}
because $1-\tilde{\mu}(u,V)<0$ and
$\tilde{\kappa}(u,V)\geq\tilde{\kappa}(U,V)>0$ for $u>0$. Our hypotheses and 
Lemma~\ref{sinal-theta} imply that $\tilde\theta$ and $\tilde\zeta$ are positive on
$]0,U]\times[0,\tilde V[$, and so, by \eqref{zeta_v} and the lower bound on $\zeta_0$, we have $\zeta \geq c_2u^{s_2}$. Hence,
\begin{equation} \label{diverge}
\int_{u_1}^{u_2}\frac{2\tilde \theta\tilde \zeta}{\tilde r^2}
(\bar u,\tilde v)\,d\bar u\geq
\frac{2c_\delta c_2\delta^{s_2}}{\tilde r^2(U,V)}
\int_{u_1}^{u_2}\frac{\tilde \theta}{|\lambda|}(\bar u,\tilde v)\,d\bar u\to
+\infty\ \ \ {\rm as}\ \tilde v\nearrow V.
\end{equation}
This shows that
$
\tilde \Gamma^{\tilde v}_{\tilde v\tilde v}(u_1,\,\cdot\,)-
\tilde \Gamma^{\tilde v}_{\tilde v\tilde v}(u_2,\,\cdot\,)
$
tends to $+\infty$,
which is a contradiction. Therefore, there is at most one $u\in[\delta,U]$
for which $\tilde \Gamma^{\tilde v}_{\tilde v\tilde v}(u,\,\cdot\,)$ is bounded.

Using~\eqref{chris} and \eqref{nu_v}, and differentiating~\eqref{acs}, we obtain
\begin{eqnarray}
\tilde \Gamma^{\tilde v}_{\tilde v\tilde v}(U,\tilde v)&=& -\,\frac{1}{{\widetilde{(1-\mu)^2}}(U,\tilde v)}\frac{2\partial_{\tilde v}\tilde\varpi(U,\tilde v)}{\tilde r(U,\tilde v)}
\frac{1}{\tilde\kappa(U,\tilde v)}\nonumber\\
&=& -\left(\frac{2\tilde\kappa\partial_{\tilde v}\tilde\varpi}{\tilde r}\right)(U,\tilde v)\nonumber\\
&=&-\left(\frac{\tilde\theta^2}{\tilde r}\right) (U,\tilde v). 
\label{uniforme2} 
\end{eqnarray}
From~\eqref{uniforme2} we see that $\tilde \Gamma^{\tilde v}_{\tilde v\tilde v}(U,\tilde v)$ tends to $-\infty$ as $\tilde v\nearrow V$.
\end{proof}

\begin{Rmk}\label{piao}
From \eqref{difere} and \eqref{diverge} we can draw the following conclusions:
\begin{enumerate}[{\rm (i)}]
\item if\/
$\tilde \Gamma^{\tilde v}_{\tilde v\tilde v}(u,\,\cdot\,)$ is bounded, then
$\tilde \Gamma^{\tilde v}_{\tilde v\tilde v}(u_1,\tilde v)\to+\infty$
as $\tilde v\to V$ when $u_1<u$, and
$\tilde \Gamma^{\tilde v}_{\tilde v\tilde v}(u_2,\tilde v)\to-\infty$
as $\tilde v\to V$ when $u_2>u$;
\item
if\/ 
$\tilde \Gamma^{\tilde v}_{\tilde v\tilde v}(u,\,\tilde v)\to-\infty$, then
$\tilde \Gamma^{\tilde v}_{\tilde v\tilde v}(u_2,\tilde v)\to-\infty$
as $\tilde v\to V$ when $u_2>u$; 
\item 
if\/
$\tilde \Gamma^{\tilde v}_{\tilde v\tilde v}(u,\,\tilde v)\to+\infty$, then
$\tilde \Gamma^{\tilde v}_{\tilde v\tilde v}(u_1,\tilde v)\to+\infty$
as $\tilde v\to V$ when $u_1<u$.
\end{enumerate}
\end{Rmk}

\subsection{Coordinates with $v$ replaced by $\hat{v}:=\int_0^v\kappa(U,\bar{v})\,d\bar{v}$}  
Assume there exists a positive sequence $(u_n)$ converging to 0 such that $\zeta_0(u_n)\neq 0$.
We change the $v$ coordinate to
\begin{equation}
\label{smooth_coordinate_f}
\hat{v}:=\int_0^v\kappa(U,\bar{v})\,d\bar{v}.
\end{equation}
According to \eqref{integral-k}, $\hat{V}:=\int_0^\infty\kappa(U,\bar{v})\,d\bar{v}<\infty$.
From Lemma~\ref{regular},
$\kappa$ is $C^1$; since $\kappa$ is also positive, this change of coordinates is
admissible ($C^2$).

We denote by $\hat{r}$, $\hat{\nu}$, $\hat{\lambda}$, $\hat{\varpi}$, $\hat{\theta}$, $\hat{\zeta}$ and $\hat{\kappa}$ the functions written in the coordinates $(u,\hat{v})$.
In particular,
$$
\hat{\kappa}(u,\hat{v})=\frac{\kappa(u,v)}{\kappa(U,v)}\qquad {\rm and}\qquad
\hat{\kappa}(U,\hat{v})\equiv 1.
$$
From $\hat\kappa(U,\hat v)\equiv 1$ and~\eqref{kappa_u}, we get
\begin{equation}\label{acsf}
\hat\kappa(u,\hat v)=e^{-\int_u^U\bigl(\frac{\hat\zeta^2}{\hat r\hat\nu}\bigr)(\bar u,\hat v)\,d\bar u}.
\end{equation}

\noindent {\it Relationship between the $\tilde{v}$ and the $\hat{v}$ coordinates.}
We now show that when $\bigl|\frac{\theta}{\lambda}\bigr|$ is unbounded
the change of coordinates from $\tilde{v}$ to $\hat{v}$ is not $C^2$ at the Cauchy horizon. From~\eqref{smooth_coordinate} and~\eqref{smooth_coordinate_f}, we write
$$
\frac{d\tilde v}{dv}(v)=-\lambda(U,v)
$$
and
$$
\frac{d\hat v}{dv}(v)=\kappa(U,v).
$$
So
\begin{equation}\label{tildebar}
\frac{d\tilde v}{d\hat v}(\hat v)=-(1-\mu)(U,v)=-(\widehat{1-\mu})(U,\hat v).
\end{equation}
Using the chain rule, \eqref{mu_0} and \eqref{omega_v}, we obtain
\begin{eqnarray}
\nonumber\frac{d^2\tilde v}{d\hat{v}^2}(\hat{v})&=&-\,\frac{d}{dv}(1-\mu)(U,v)\frac{dv}{d\hat{v}}(\hat{v})\\
\nonumber&=&\left(\frac{2}{r}\partial_v\varpi\right)(U,v)\frac{1}{\kappa(U,v)}+{\rm bounded\ terms}\\
\nonumber&=&\left(\frac{1}{r}\frac{\theta^2}{\kappa^2}\right)(U,v)+{\rm bounded\ terms}\\
\label{ddoisddois}&=&\left[\frac{(1-\mu)^2}{r}\left(\frac{\theta}{\lambda}\right)^2\right](U,v)+{\rm bounded\ terms.}
\end{eqnarray}
According to Proposition~\ref{fechado-0}, $(1-\mu)(U,\, \cdot \,)$ is bounded away from zero (note that this quantity does not depend on the choice of coordinate system).
So indeed, we have 
\begin{Rmk}\label{non-equivalent} 
The change of coordinates from $\hat{v}$ to $\tilde{v}$ is not $C^2$ 
at the Cauchy horizon when
$\bigl|\frac{\theta}{\lambda}\bigr|$ is unbounded along $u=U$.
\end{Rmk}

The next result is a direct translation of Proposition~\ref{fechado-0} to the new coordinates $(u,\hat v)$.

\begin{Prop}\label{fechado-0f}
Suppose that 
$$ 
|\zeta_0(u)|\leq cu^s\ {\rm for\ some\ nonnegative}\ s>\frac{7\Psi}9-1,
$$ 
where $c>0$ and $u\in\left[0,U\right]$,
and suppose there exists a positive sequence $(u_n)$ converging to 0 such that $\zeta_0(u_n)\neq 0$.
If $U>0$ is sufficiently small then, for all\/ $0<\delta<U$, the functions $\hat r$, $\hat\nu$, $\hat\lambda$, $\hat\varpi$,
$\hat\zeta$ and $\hat\kappa$ (but not necessarily $\hat{\theta}$)
admit continuous extensions to the closed rectangle\/ $[\delta,U]\times[0,\hat{V}]$.
Equations~\eqref{r_u} to \eqref{omega_u}, \eqref{kappa_u}, and~\eqref{kappa_at_u} are satisfied on this set.
Finally, $\widehat{(1-\mu)}(u,\hat{V})$ is negative for $u>0$.
\end{Prop}
\begin{proof}
This is a consequence of Propostion~\ref{fechado-0} and the fact that
the map $\hat{v}\mapsto\tilde{v}$ extends to a $C^1$ map from $[0,\hat{V}]$
to $[0,V]$. For example, to check~\eqref{r_v} at the Cauchy horizon, note that
from $\partial_{\tilde{v}}\tilde{r}(u,\tilde{v})=\tilde\lambda(u,\tilde{v})$
we conclude that
\begin{eqnarray*}
\partial_{\hat{v}}\hat{r}(u,\hat{v})&=&
\partial_{\tilde{v}}\tilde{r}(u,\tilde{v})\frac{d\tilde{v}}{d\hat{v}}(\hat{v})
=\tilde\lambda(u,\tilde{v})[-(1-\mu)(U,v)]\\
&=&\frac{\lambda(u,v)}{\lambda(U,v)}(1-\mu)(U,v)=
\frac{\lambda(u,v)}{\kappa(U,v)}=\hat\lambda(u,\hat{v}).
\end{eqnarray*}
\end{proof}

The spherically symmetric Lorentzian manifold
${\cal M}$ is now ${\hat{{\cal Q}}}\times{\mathbb S}^2$, where
$\hat{{\cal Q}}$ admits the global null coordinate system
$(u,\hat v)$ defined on $[0,U]\times[0,\hat{V}]\setminus\{(0,\hat{V})\}$, and the metric is
$$ 
g = -\hat\Omega^2(u,\hat v)\,dud\hat v+\hat r^2(u,\hat v)\,\sigma_{{\mathbb S}^2},
$$ 
with $\hat\Omega^2=-4\hat\nu\hat\kappa$. 
The field $\hat\phi(u,\hat{v})$ equals ${\phi}(u,v)$ and so $\tilde{\phi}(u,\tilde{v})$.
The nonvanishing Christoffel symbols of the metric on ${\cal M}$ are written
as the ones above, with tildes replaced by hats. For example, instead 
of~\eqref{chris}, we have
\begin{eqnarray} 
\hat \Gamma^{\hat v}_{\hat v\hat v}
&=&
\hat \Omega^{-2}\partial_{\hat v}(\hat \Omega^2)\
=\frac{\partial_{\hat v} \hat\nu}{\hat\nu}+\frac{\partial_{\hat v} \hat\kappa}{\hat \kappa}.\label{chrisf}
\end{eqnarray}

\begin{Prop}\label{meiof}
Suppose that  
$$ 
c_2u^{s_2}\leq\zeta_0(u)\leq c_1u^{s_1}\ {\rm for\ some}\ \frac{7\Psi}9-1<s_1\leq s_2<\Psi-1,
$$ 
where $c_1,c_2>0$ and $u\in\left[0,U\right]$.
For any $\delta>0$ the field $\bigl|\frac{\hat\theta}{\hat\lambda}\bigr|(u,\hat v)$ tends to $+\infty$ as $\hat v\nearrow \hat{V}$, uniformly for\/ $u\in[\delta,U]$, provided that $U$ is sufficiently small.
Moreover, $\hat \Gamma^{\hat v}_{\hat v\hat v}(u,\hat{v})\to +\infty$ as $ \hat{v}\nearrow\hat{V}$ for each $u\in[\delta,U[$.
\end{Prop}
\begin{proof}
The fact that $\bigl|\frac{\hat\theta}{\hat\lambda}\bigr|(u,\hat v)$ tends to $+\infty$, uniformly for\/ $u\in[\delta,U]$, as $\hat v\nearrow \hat V$ follows from
Proposition~\ref{meio} because this quantity is invariant under changes of coordinates.

By construction, the coordinate transformation from $\tilde v$ to $\hat v$ has regularity
$C^2([0,V[)$. Since by Proposition~\ref{fechado-0}  $(1-\mu)$ is bounded and bounded away from zero, we have
from~\eqref{tildebar} that both $\frac{d\hat v}{d\tilde v}$ and $\frac{d\tilde v}{d\hat v}$
are always different from zero. Therefore, this coordinate transformation has in fact regularity $C^1([0,V])\cap C^2([0,V[)$.
Hence, the fact that $\tilde\lambda$ is bounded and bounded away from zero in
$[\delta,U]\times[0,V]$, see Proposition~\ref{fechado-0} and \eqref{lambdapositive}, implies that
$\hat\lambda$ is bounded and bounded away from zero in $[\delta,U]\times[0,\hat V]$. In particular, we recover the fact that
$\hat\kappa$ is bounded in $[\delta,U]\times[0,\hat V]$.
Therefore, the proof of Proposition~\ref{meio} applies to the present case and for all $u\in[\delta,U]$,
with one possible exception, $\hat \Gamma^{\hat v}_{\hat v\hat v}(u,\hat{v})$ is unbounded as $ \hat{v}\nearrow\hat{V}$.

In this case we have that
$\hat \Gamma^{\hat v}_{\hat v\hat v}(U,\hat{v})$ 
is bounded because $\hat{\kappa}(U,\hat{v})\equiv 1$ (see \eqref{chrisf} and \eqref{nu_v}). It follows from Remark~\ref{piao} that
$$
\hat \Gamma^{\hat v}_{\hat v\hat v}(u,\hat{v})\to +\infty\ \ \ {\rm as}\ \ \ \hat{v}\to\hat{V}
$$
for each $u\in[\delta,U[$.
\end{proof}

\begin{Rmk}
Examining the proof of\/ {\rm Proposition~\ref{meiof}}, we see that we only used the
specific form of the coordinates $\hat v$ in the last paragraph.
So, if $\tilde v\mapsto\mathring v$ is any coordinate transformation with
regularity $C^1([0,V])\cap C^2([0,V[)$, we conclude that there is at most one value
$u>0$ for which $\mathring \Gamma^{\hat v}_{\mathring v\mathring v}(u,\mathring{v})$
is bounded. If\/ $\mathring \Gamma^{\mathring v}_{\mathring v\mathring v}(\bar u,\mathring{v})$
is bounded, then $\mathring \Gamma^{\mathring v}_{\mathring v\mathring v}(u,\mathring{v})\to-\infty$ as $\mathring v\to\mathring V$ for $u>\bar u$, and
$\mathring \Gamma^{\mathring v}_{\mathring v\mathring v}(u,\mathring{v})\to+\infty$ as $\mathring v\to\mathring V$ for $u<\bar u$.
This excludes the existence of $C^{0,1}$ extensions of the metric using
these coordinates.

\end{Rmk}

% \begin{Rmk}
% Examining the proof of\/ {\rm Proposition~\ref{meiof}}, we see that most of it applies to any
% coordinate transformation with regularity $C^1([0,V])\cap C^2([0,V[)$,
% since the specific form of the coordinates $\hat v$ is only used in the last paragraph.
% So, if $\tilde v\to\mathring v$ is any such coordinate transformation, 
% we conclude that there is at most one value
% $u>0$ for which $\mathring \Gamma^{\mathring v}_{\mathring v\mathring v}(u,\mathring{v})$
% is bounded. If\/ $\mathring \Gamma^{\mathring v}_{\mathring v\mathring v}(\bar u,\mathring{v})$
% is bounded, then $\mathring \Gamma^{\mathring v}_{\mathring v\mathring v}(u,\mathring{v})\to-\infty$ as $\mathring v\to\mathring V$ for $u>\bar u$, and
% $\mathring \Gamma^{\mathring v}_{\mathring v\mathring v}(u,\mathring{v})\to+\infty$ as $\mathring v\to\mathring V$ for $u<\bar u$.
% % \end{Rmk}

\begin{Rmk} 
Suppose that the hypotheses of\/ {\rm Proposition~\ref{meio}} hold. Then $\tilde \Gamma^{\tilde v}_{\tilde v\tilde v}(U,\tilde{v})$ tends
to $-\infty$ as $\tilde v\to V$, and $\hat \Gamma^{\hat v}_{\hat v\hat v}(U,\hat{v})$ is bounded. From elementary Riemannian geometry we have
\[
\hat\Gamma^{\hat v}_{\hat v\hat v} = \frac{d\tilde v}{d\hat v} \tilde\Gamma^{\tilde v}_{\tilde v\tilde v} + \frac{d\hat v}{d\tilde v} \frac{d^2\tilde v}{d\hat v^2},
\]
and so $\frac{d^2\tilde v}{d\hat v^2}$ must blow up at the Cauchy horizon (as was already shown in \eqref{ddoisddois} by direct computation). This again shows that the two coordinate systems $(u,\tilde v)$ and $(u,\hat v)$ are not $C^2$ compatible. More generally, the same reasoning can be applied to show the $C^2$ incompatibility of any two coordinate systems whose Christoffel symbols $\Gamma^{v}_{vv}$ have different asymptotic behavior at the Cauchy horizon. In particular, different choices of $U$ yield incompatible $(u,\hat v)$ coordinates (when $\frac{\theta}{\lambda}$ is unbounded).
\end{Rmk}

It turns out that, although unbounded, the Christoffel symbols of the $(u,\hat v)$ coordinates are in $L^2$.

\begin{Prop}\label{CC}
Suppose that  
$$ 
|\zeta_0(u)|\leq cu^s\ {\rm for\ some\ nonnegative}\ s>\frac{7\Psi}9-1,
$$ 
where $c>0$ and $u\in\left[0,U\right]$.
For any $0<\delta<U$, the Christoffel symbols $\hat \Gamma^C_{AB}$, $\hat \Gamma^u_{AB}$, $\hat \Gamma^{\hat v}_{AB}$, 
$\hat \Gamma^A_{B\hat u}$, $\hat \Gamma^A_{B\hat v}$ and~$\hat \Gamma^u_{uu}$  
are bounded in $[\delta,U]\times[0,\hat{V}]$, provided that $U$ is sufficiently small.
Furthermore,
$\int_0^{\hat{V}}|\hat \Gamma^{\hat v}_{\hat v\hat v}|^2(u,\hat v)\,d\hat v$ and
$\int_0^{\hat{V}}|\hat\theta|^2(u,\hat v)\,d\hat v$
are bounded for $u\in[\delta,U]$.
Consequently, the Christoffel symbols and $\hat\theta$ (and also $\hat\zeta$) belong to 
$L^2({\cal M}_\delta)$, with ${\cal M}_\delta$ the preimage of $[\delta,U]\times[0,\hat{V}]$ by the double null coordinate system $(u,\hat v)$.
\end{Prop}
\begin{proof}
In the proof of Proposition~\ref{fechado-0f} we showed that all the functions in the first order system except $\hat\theta$,
i.e.\ the functions $\hat r$, $\hat\nu$, $\hat\lambda$, $\hat\varpi$, $\hat\zeta$ and $\hat\kappa$, extend to 
continuous functions in $[\delta,U]\times[0,\hat{V}]$, with $\hat r>0$, $\hat\nu<0$ and $\hat\kappa>0$. In addition, we proved that all
the equations of the first order system~\eqref{r_u}$-$\eqref{kappa_at_u}, except~\eqref{omega_v}, \eqref{theta_u} and~\eqref{zeta_v},
are satisfied in $[\delta,U]\times[0,\hat{V}]$; in particular~\eqref{kappa_u} (the equation for $\partial_u\hat\kappa$) is satisfied in this rectangle; moreover, 
the expression for $\partial_u\hat\nu$ is obtained from
$$
\hat\nu(u,\hat v)=-e^{-\int_0^{\hat{v}}\left(2\hat\kappa\frac{1}{\hat{r}^2}\left(
\frac{e^2}{\hat r}+\frac{\Lambda}{3}\hat{r}^3-\hat\omega
\right)\right)(u,\bar v)\,d\bar v}.
$$
Therefore, $\hat \Gamma^C_{AB}$, $\hat \Gamma^u_{AB}$, $\hat \Gamma^{\hat v}_{AB}$,
$\hat \Gamma^A_{B\hat u}$, $\hat \Gamma^A_{B\hat v}$,
and $\hat \Gamma^u_{uu}$ are bounded in $[\delta,U]\times[0,\hat{V}]$.

Proposition~\ref{fechado-0} and \eqref{lambdapositive} give a positive lower bound for $\hat{\kappa}$ in $[\delta,U]\times[0,\hat{V}]$.
By~\eqref{omega_v}, we then know that
$\int_0^{\hat{V}}|\hat\theta|^2(u,\hat v)\,d\hat v$ is bounded for $u\in[\delta,U]$.
Differentiating both sides of~\eqref{acsf} with respect to $\hat{v}$, and using~\eqref{nu_v} and~\eqref{zeta_v}, we get
\begin{eqnarray*}
\frac{\partial_{\hat v}\hat\kappa}{\hat\kappa}(u,\hat{v})&=&-
2\int_u^U\frac{\hat{\zeta}{\partial_{\hat v}\hat\zeta}}{\hat r\hat{\nu}}(\bar u,\hat v)\,d\bar u
+\int_u^U\frac{\hat{\zeta}^2\hat{\lambda}}{\hat{r}^2\hat{\nu}}(\bar u,\hat v)\,d\bar u\\
&&+\int_u^U\frac{\hat{\zeta}^2\partial_{\hat v}\hat{\nu}}{\hat{r}\hat{\nu}^2}(\bar u,\hat v)\,d\bar u\\
&=&
2\int_u^U\frac{\hat{\zeta}\hat{\theta}}{\hat r^2}(\bar u,\hat v)\,d\bar u
+\int_u^U\frac{\hat{\zeta}^2\hat{\lambda}}{\hat{r}^2\hat{\nu}}(\bar u,\hat v)\,d\bar u\\
&&+\int_u^U\frac{\hat{\zeta}^2\hat{\kappa}\partial_{\hat{r}}(\widehat{1-\mu})}{\hat{r}\hat{\nu}}(\bar u,\hat v)\,d\bar u.
\end{eqnarray*}
Let $0<\delta<U$. From \eqref{nu_negativo} we have $\hat \nu(u,\hat V) \leq -c < 0$
for $u \in [\delta, U]$ (note that $\hat{\nu}(u,\hat{v})=\tilde{\nu}(u,\tilde{v}(\hat{v}))$). The previous equality, \eqref{nu_v} and~\eqref{chrisf} then imply that there exists a $C>0$
such that
\[
|\hat \Gamma^{\hat v}_{\hat v\hat v}|(u,\hat v) \leq C\left(1+
\int_u^U|\hat\theta|(\bar u,\hat v)\,d\bar u
\right)
\]
for $(u,\hat v)\in[\delta,U]\times[0,\hat{V}[$, and so, using H\"older's inequality, 
\begin{eqnarray*}
|\hat \Gamma^{\hat v}_{\hat v\hat v}|^2(u,\hat v)
&\leq& C\left(1+
\int_u^U|\hat\theta|^2(\bar u,\hat v)\,d\bar u
\right).
\end{eqnarray*}
Therefore,
\begin{eqnarray*}
\int_0^{\hat{V}}|\hat \Gamma^{\hat v}_{\hat v\hat v}|^2(u,\hat v)\,d\hat{v}
&\leq& C\left(1+
\int_0^{\hat{V}}\int_u^U|\hat\theta|^2(\bar u,\hat v)\,d\bar ud\hat{v}
\right)\\
&=&C\left(1+
\int_u^U\int_0^{\hat{V}}|\hat\theta|^2(\bar u,\hat v)\,d\hat{v}d\bar u
\right)\\
&\leq&C,
\end{eqnarray*}
for $u\in[\delta,U]$.

Finally, note that the square of the $L^2$ norm of a function $\hat h$ on ${\cal M}_\delta$ is given by
$$
\int_{{\cal M}_\delta}{\hat h}^2\,dV_4=4\pi\int_{[\delta,U]\times[0,\hat{V}]}\left[\hat r^2\frac{\hat\Omega^2}2\hat h^2\right](u,\hat v)\,dud\hat v.
$$
Since the functions $\hat r$ and $\hat\Omega^2=-4\hat\nu\hat\kappa$ are bounded in $[\delta,U]\times[0,\hat{V}]$, we conclude that the Christoffel symbols and $\hat\theta$ are in $L^2({\cal M}_\delta)$.
\end{proof}

So, in our framework the Christodoulou-Chru\'sciel formulation of strong cosmic censorship (see~\cite{Christodoulou:2008}) does not hold:
\begin{Cor}\label{cc-fails} 
Suppose that
$$ 
|\zeta_0(u)|\leq cu^s\ {\rm for\ some\ nonnegative}\ s>\frac{7\Psi}9-1,
$$ 
where $c>0$ and $u\in\left[0,U\right]$.
Then, provided that $U$ is sufficiently small,
$({\cal M},g)$ and $\hat\phi$ extend across the Cauchy horizon (in a non-unique way) to
spherically symmetric 
$(\check{\cal M},\check g)$ and $\check \phi$, with
$\check{\cal M}=\check{\cal Q}\times{\mathbb S}^2$ a $C^2$ manifold and
$$
\check g=-\check\Omega^2(u,\hat v)\,dud\hat v+\check r^2(u,\hat v)\,\sigma_{{\mathbb S}^2}
$$
a $C^0$~metric on $\check{\cal M}$.
 Here $\check{\cal Q}$ has a global null coordinate system
$(u,\hat v)$ defined on $[0,U]\times[0,\hat{V}]\setminus\{(0,\hat{V})\}\cup{\cal V}$, with ${\cal V}$ a neighborhood of $]0,U]\times\{\hat{V}\}$.
Furthermore,
\begin{equation}\label{Gerrard}
\check \Gamma \in L^2_{{\rm loc}}\quad{\rm and}\quad\check \phi \in H^1_{{\rm loc}}.
\end{equation}
\end{Cor}
\begin{proof}
For $(u,\hat v)$ with $u>0$ and $\hat v>\hat{V}$, define
$$
\check\Omega^2(u,\hat v)=\hat\Omega^2(u,\hat{V}),\qquad\check\phi(u,\hat v)=\hat\phi(u,\hat{V}),
$$
and
$$
\check r(u,\hat v)=\hat r(u,\hat{V})+\hat\lambda(u,\hat{V})(\hat v-\hat{V}).
$$
Choose a neighborhood ${\cal V}$ of $]0,U]\times\{\hat{V}\}$ such that $\check r>0$ on $[0,U]\times[0,\hat{V}]\setminus\{(0,\hat{V})\}\cup{\cal V}$.
The extensions $\check\Omega^2$, $\check\phi$ and $\check r$ of $\hat\Omega^2$, $\hat\phi$ and $\hat r$ are continuous.
For $u>0$ and $\hat v>\hat{V}$, we get
$$
\partial_u\check\Omega^2(u,\hat v)=\partial_u\hat\Omega^2(u,\hat{V}),\qquad\partial_{\hat v}\check\Omega^2(u,\hat v)=0,
$$
$$
\partial_u\check\phi(u,\hat v)=\partial_u\hat\phi(u,\hat{V}),\qquad\partial_{\hat v}\check\phi(u,\hat v)=0,
$$
$$
\check\nu(u,\hat v)=\hat\nu(u,\hat{V})+\partial_u\hat\lambda(u,\hat{V})(\hat v-\hat{V})
$$
and
$$
\check\lambda(u,\hat v)=\hat\lambda(u,\hat{V}).
$$
Clearly, $\partial_u\check\Omega^2$, $\check\lambda$ and~$\check\nu$ are also continuous.
Therefore, $\check \Gamma^C_{AB}$, $\check \Gamma^u_{AB}$, $\check \Gamma^{\hat v}_{AB}$,
$\check \Gamma^A_{B\hat u}$, $\check \Gamma^A_{B\hat v}$,
and~$\check\Gamma^u_{uu}$
are continuous, and so is the field $\check\zeta$. Finally, $\check \Gamma^{\hat v}_{\hat v\hat v}$ and $\check\theta$ are zero for $\hat v>\hat{V}$.
It would be easy to construct other extensions of $({\cal M},g)$ and $\hat\phi$ satisfying \eqref{Gerrard}.
\end{proof}
Note that there is no guarantee that the extensions above satisfy the Einstein equations.
Moreover, the function $\hat\theta$ may not admit a continuous extension to the Cauchy horizon.

\begin{Rmk}\label{dias_l} Since in the previous extension
$$
\partial_{\hat v}\check\nu(u,\hat v)=\partial_u\check\lambda(u,\hat v)=\partial_u\hat\lambda(u,\hat{V}),
$$
for $\hat v>\hat{V}$, we constructed a $C^0$ extension of the metric such that (\/$
\check \Gamma \in L^2_{{\rm loc}}$, $\check \phi \in H^1_{{\rm loc}}$ and)
 the second mixed derivatives of $\check r$ are continuous. This would not be possible if $\hat\varpi(\,\cdot\,,\hat{V})$ were $+\infty$ (see~\eqref{lambda_u} and~\eqref{nu_v}).
In\/~{\rm \cite[Theorem~11.1]{Dafermos2}} M.~Dafermos constructs $C^0$ extensions of the metric without assuming any restriction on the continuous function $\zeta_0$, so
without any control on $\hat\varpi(\,\cdot\,,\hat{V})$. 
%The reader should note that choosing $\hat\lambda(U,\hat v)\equiv -1$ (as we do) or
%$\hat\kappa(U,\hat v)\equiv 1$ (as in\/~{\rm \cite[Theorem~11.1]{Dafermos2}}) is 
%similar when $\hat\varpi(U,\hat{V})<\infty$. But these choices are quite different when
%$\hat\varpi(U,\hat{V})=\infty$ because then $\widetilde{(1-\mu)}(u,\hat{V})=-\infty$.
\end{Rmk}

\section{Extensions of solutions beyond the Cauchy horizon}\label{CH}

It is clear that in order to improve on the results of the previous section we need to control the field $\frac\theta\lambda$. 
In view of Proposition~\ref{meio}, this requires a stronger restriction on the exponent $s$.
Once the field is controlled, it turns out to be possible to construct smooth extensions of our spacetime which in fact are solutions of the Einstein equations.

More precisely, in
the main part of this section we assume that  
$$ 
|\zeta_0(u)|\leq cu^s\ {\rm for\ some}\ s>\frac{13\Psi}9-1,
$$ 
where $c>0$ and $u\in\left[0,U\right]$.
In Lemma~~\ref{theta-control}, we obtain the desired bound for $\frac\theta\lambda$ in $J^+(\gam)$.
We then start by proving that our solution of the first order system~\eqref{r_u}$-$\eqref{kappa_at_u} can be extended to the closed rectangle $[\delta,U]\times[0,V]$, for any $0<\delta<U$, while still satisfying~\eqref{r_u}$-$\eqref{kappa_at_u}. By taking the values of the functions at the Cauchy horizon as initial data on $[\delta,U]\times\{V\}$, and choosing new initial data on $\{U\} \times [V,V+\eps]$, we can build (non-unique) extensions of the solution beyond the Cauchy horizon. The new initial data can be chosen with the required regularity so that we obtain classical solutions of the Einstein equations.
We finish the section by analyzing the behavior of the Kretschmann scalar at the Cauchy horizon, under the hypotheses used in this and in the previous sections.

\begin{Lem}[Bounding $\frac\theta\lambda$]\label{theta-control}
Suppose that  
$$ 
|\zeta_0(u)|\leq cu^s\ {\rm for\ some}\ s>\frac{13\Psi}9-1,
$$ 
where $c>0$ and $u\in\left[0,U\right]$.
Then there exists a constant $C>0$ such that
\begin{equation}\label{eq34}
\Bigl|\frac\theta\lambda\Bigr|(u,v)\leq C,
\end{equation}
for $(u,v)\in J^+(\gam)$, provided that $U$ is sufficiently small.
Furthermore, 
\begin{equation}\label{canto}
\lim_{(u,v)\to(0,\infty)}\,\Bigl|\frac\theta\lambda\Bigr|(u,v)=0.
\end{equation}
\end{Lem}
\begin{proof}
Integrating~\eqref{theta_lambda}, we obtain
\begin{eqnarray}
\frac\theta\lambda(u,v)&=&\frac\theta\lambda(\ugam(v),v)e^{-\int_{\ugam(v)}^{u}\bigl[\frac{\nu}{1-\mu}\partial_r(1-\mu)\bigr](\tilde u,v)\,d\tilde u}\nonumber\\
&&-\int_{\ugam(v)}^{u}\frac\zeta r(\tilde u,v)e^{-\int_{\tilde u}^{u}\bigl[\frac{\nu}{1-\mu}\partial_r(1-\mu)\bigr](\bar{u},v)\,d\bar{u}}\,d\tilde u.\label{field_26}
\end{eqnarray}
By Theorem~\ref{thmNoMass}, we know that we have $|\partial_r(1-\mu)+2k_-|<\delta$ in $J^+(\gam)$ for sufficiently small $U$.
Using~\eqref{eq3} and~\eqref{int-nu},
\begin{eqnarray}
&&\Bigl|\frac\theta\lambda\Bigr|(\ugam(v),v)e^{-\int_{\ugam(v)}^{u}\bigl[\frac{\nu}{1-\mu}\partial_r(1-\mu)\bigr](\tilde u,v)\,d\tilde u}\nonumber\\
&&\qquad\qquad\leq
 Ce^{-2\left(\frac{k_+(s+1)}{1+\beta}-k_-\beta-\delta\right)v}e^{2\left(\frac{k_-}{1+\beta}+\delta\right)v}\nonumber\\
&&\qquad\qquad\leq
Ce^{-\,\frac{2k_+}{1+\beta}(s+1-\Psi(\beta^2+\beta+1)-\delta)v}.\label{eq40}
\end{eqnarray}
This exponent can be made negative for
\begin{equation}\label{sb}
s>\Psi(\beta^2+\beta+1)-1.
\end{equation}
Now, according to \eqref{ray_v_bis} and \eqref{nu_mu_bis}
$$
\frac{\nu}{1-\mu}(\bar u,v)\leq\frac{\nu}{1-\mu}(\bar u,v_{\ckrm}(\bar u))\leq\frac{1+\delta}{2k_+\bar u},
$$
due to the monotonicity of $\frac{\nu}{1-\mu}$. Thus,
\begin{equation}\label{exph}
e^{-\int_{\tilde u}^{u}\bigl[\frac{\nu}{1-\mu}\partial_r(1-\mu)\bigr](\bar{u},v)\,d\bar{u}}\leq e^{(\Psi+\delta)\ln\left(\frac u{\tilde u}\right)}=\left(\frac u{\tilde u}\right)^{\Psi+\delta}
\leq \left(\frac U{\tilde u}\right)^{\Psi+\delta}.
\end{equation}
Combining this with~\eqref{eq24}, if $s>\Psi(\beta^2+1)-1$ and if the parameters are chosen appropriately, we get
\begin{eqnarray}
&&\int_{\ugam(v)}^{u}\frac{|\zeta|} r(\tilde u,v)e^{-\int_{\tilde u}^{u}\bigl[\frac{\nu}{1-\mu}\partial_r(1-\mu)\bigr](\bar{u},v)\,d\bar{u}}\,d\tilde u\nonumber\\
&&\qquad\qquad\leq C\int_{\ugam(v)}^{u}\tilde u^{s-\Psi\beta^2-\delta}\tilde u^{-\Psi-\delta}\,d\tilde u\nonumber\\
&&\qquad\qquad\leq Cu^{s+1-\Psi(\beta^2+1)-\delta}.\label{eq41}
\end{eqnarray}

Using~\eqref{eq40} and~\eqref{eq41} in~\eqref{field_26}, taking into account that the right-hand side of~\eqref{sb} would be $\frac{13\Psi}{9}-1$ if $\beta$ were $\frac 13$,
and recalling that we can choose $\beta=\frac 13+\eps$, we obtain~\eqref{eq34}.

To prove the last assertion, notice that for $(u,v)\in J^-(\gam)\cap J^+(\cg_{\ckrm})$ the estimate on the right-hand side of~\eqref{eq3} applies since $u\leq\ugam(v)$.
Also, recall~\eqref{canto-part}. 
All this information, together with~\eqref{field_26} and the bounds~\eqref{eq40} and~\eqref{eq41}, implies~\eqref{canto}.
\end{proof}

\begin{Thm}[Extending the solution of the first order system up to the Cauchy horizon]\label{fechado}
Suppose that  
$$ 
|\zeta_0(u)|\leq cu^s\ {\rm for\ some}\ s>\frac{13\Psi}9-1,
$$ 
where $c>0$ and $u\in\left[0,U\right]$.
If $U>0$ is sufficiently small then, for all\/ $0<\delta<U$, the functions $\tilde r$, $\tilde\nu$, $\tilde\lambda$, $\tilde\varpi$, $\tilde\theta$, $\tilde\zeta$ and $\tilde\kappa$ satisfy
the first order system~\eqref{r_u}$-$\eqref{kappa_at_u} on the closed rectangle\/ $[\delta,U]\times[0,V]$.
\end{Thm}
\begin{Rmk}\label{candice}
{\rm Theorem~\ref{inflation2}} and\/~{\rm Proposition~\ref{meio}} imply that
there is no hope of lowering the constant $\frac{13}9$ below\/ $1$.
\end{Rmk}

\begin{proof}[Proof of\/ {\rm Theorem~\ref{fechado}}]
We fix $0<\delta<U$.
We already did most of the work in~Proposition~\ref{fechado-0}. So, we just need to prove the assertion for $\tilde\theta$ and that~\eqref{omega_v}, \eqref{theta_u} and~\eqref{zeta_v} are 
satisfied on $[\delta,U]\times[0,V]$.
 As before, we proceed in three steps.

{\em Step 1}. We prove that $\tilde\theta(\cdot,\tilde{v})$ converges uniformly to $\tilde\theta(\cdot,V)$ in $[\delta,U]$ as $\tilde{v}\to V$, that is,
\begin{equation}\label{continuity-vt}
\forall_{\eps>0}\, \exists_{\tilde\delta>0}\, \forall_{u\in[\delta,U]}\ \ |\tilde v-V|<\tilde\delta\ \Rightarrow\ |\tilde\theta(u,\tilde v)-\tilde\theta(u,V)|<\eps.
\end{equation}
We let $v\nearrow\infty$ in~\eqref{field_26}. Taking into account the estimate~\eqref{eq40} for the first term on the right-hand side,
and using Lebesgue's Dominated Convergence Theorem and~\eqref{eq41} for the second term on the right-hand side, we conclude that
\begin{eqnarray*}
\frac\theta\lambda(u,\infty)&=&
-\int_{0}^{u}\frac\zeta r(\tilde u,\infty)e^{-\int_{\tilde u}^{u}\bigl[\frac{\nu}{1-\mu}\partial_r(1-\mu)\bigr](\bar{u},\infty)\,d\bar{u}}\,d\tilde u.
\end{eqnarray*}
Hence
$$
\widetilde{\left(\frac{\theta}{\lambda}\right)}(u,V)=\frac{\theta}{\lambda}(u,\infty)
$$
is well defined and 
$\tilde\theta(u,V)=\widetilde{\Bigl(\frac{\theta}{\lambda}\Bigr)}(u,V)\tilde\lambda(u,V)$
is also well defined.
We now wish to prove uniform convergence of $\frac\theta\lambda(\,\cdot\,,v)$ to  $\frac\theta\lambda(\,\cdot\,,\infty)$, as $v\nearrow\infty$.
We write
\begin{eqnarray*}
\frac\theta\lambda(u,v)-\frac\theta\lambda(u,\infty)&=&\frac\theta\lambda(\ugam(v),v)e^{-\int_{\ugam(v)}^{u}\bigl[\frac{\nu}{1-\mu}\partial_r(1-\mu)\bigr](\tilde u,v)\,d\tilde u}\nonumber\\
&&-\int_{{\hat\delta}}^{u}\frac\zeta r(\tilde u,v)e^{-\int_{\tilde u}^{u}\bigl[\frac{\nu}{1-\mu}\partial_r(1-\mu)\bigr](\bar{u},v)\,d\bar{u}}\,d\tilde u\\
&&+\int_{{\hat\delta}}^{u}\frac\zeta r(\tilde u,\infty)e^{-\int_{\tilde u}^{u}\bigl[\frac{\nu}{1-\mu}\partial_r(1-\mu)\bigr](\bar{u},\infty)\,d\bar{u}}\,d\tilde u\\
&&- 
\int_{\ugam(v)}^{{\hat\delta}}\frac\zeta r(\tilde u,v)e^{-\int_{\tilde u}^{u}\bigl[\frac{\nu}{1-\mu}\partial_r(1-\mu)\bigr](\bar{u},v)\,d\bar{u}}\,d\tilde u\\
&&+ 
\int_{0}^{{\hat\delta}}\frac\zeta r(\tilde u,\infty)e^{-\int_{\tilde u}^{u}\bigl[\frac{\nu}{1-\mu}\partial_r(1-\mu)\bigr](\bar{u},\infty)\,d\bar{u}}\,d\tilde u\\
&=:&I+I\!I+I\!I\!I+I\!V+V.
\end{eqnarray*}
Suppose that we are given $\eps>0$. 
Notice the upper limits of the integrals in $I\!V$ and $V$: the outer integrals have upper limit ${\hat\delta}$, while the inner integrals
have upper limit $u$. Nevertheless, we may do computations similar to~\eqref{eq41}, using~\eqref{exph}, to conclude that
we may choose ${\hat\delta}>0$ so that $|I\!V|+|V|<\frac\eps 3$, for all $u\in[\delta,U]$.
We fix such a ${\hat\delta}$. By~\eqref{eq40}, there exists $\tilde V_\eps>0$ such that for $v\geq\tilde V_\eps$ we have $|I|<\frac\eps 3$, again for all $u\in[\delta,U]$.
When estimating $|I\!I+I\!I\!I|$ we replace the upper limit of integration $u$ by $U$ (after gathering this difference into a single integral and taking absolut values).
Finally, by uniform convergence of the functions in the integral $I\!I$ to the functions in the integral $I\!I\!I$, in $[{\hat\delta},U]$, there exists $V_\eps\geq\tilde V_\eps$
such that $|I\!I+I\!I\!I|<\frac\eps 3$, for $v\geq V_\eps$. 
So for $v\geq V_\eps$ and for all $u\in[\delta,U]$, we have
$$
\Bigl|\frac\theta\lambda(u,v)-\frac\theta\lambda(u,\infty)\Bigr|<\eps.
$$
This establishes the desired uniform convergence.

{\em Step 2}. As in Step~2 of the proof of Proposition~\ref{fechado-0}, we conclude that $\tilde\theta$
is continuous in the closed rectangle $[\delta,U]\times[0,V]$. 

{\em Step 3}. As in Step~3 of the proof of Proposition~\ref{fechado-0}, we conclude 
that~\eqref{omega_v}, \eqref{theta_u} and~\eqref{zeta_v} are satisfied also on the segment $[\delta,U]\times\{V\}$. 
\end{proof}
\begin{Rmk}
The function $\tilde\varpi(U,\,\cdot\,)$ is continuously differentiable on $[0,V]$ due to~\eqref{omega_v}.
\end{Rmk}

\vspace{4mm}

\noindent {\em On the choice of initial data beyond the Cauchy horizon.} Fix $0<\hat\eps<\tilde r(U,V)$, and consider the continuous extension $\tilde\lambda(U,\,\cdot\,)\equiv -1$ to the interval\/ $[0,V+\hat\eps]$.
According to this choice, define
$$
\tilde r(U,\tilde v)=\tilde r(U,V)+\int_V^{\tilde v}\tilde\lambda(U,\bar v)\,d\bar v=\tilde r(U,V)-(\tilde v-V),
$$
for $\tilde v\in\,]V,V+\hat\eps]$.
The upper bound on $\hat\eps$ is imposed to guarantee that $$\tilde r(U,V+\hat\eps)=\tilde r(U,V)-\hat\eps>0.$$
Choose a continuously differentiable extension of $\tilde\varpi(U,\,\cdot\,)$ to the interval\/ $[0,V+\hat\eps]$, with $\partial_{\tilde v}\tilde\varpi(U,\,\cdot\,)\geq 0$,
for $\tilde v\in\,]V,V+\hat\eps]$. Since $\widetilde{(1-\mu)}(U,V)<0$, by continuity, there exists $0<\eps\leq\hat\eps$ such that
$$
\widetilde{(1-\mu)}(U,\tilde v)=\left(1-\frac{2\tilde\varpi}{\tilde r}+\frac{e^2}{\tilde r^2}-\frac\Lambda 3\tilde r^2\right)(U,\tilde v)<0,
$$
for $\tilde v\in\,]V,V+\eps]$. For $\tilde v\in\,]V,V+\eps]$, define
\begin{equation}\label{kappa_big_u}
\tilde\kappa(U,\tilde v)=\frac{-1}{\widetilde{(1-\mu)}(U,\tilde v)}
\end{equation}
and
\begin{equation}\label{teta}
\tilde\theta(U,\tilde v)={\rm sign}\,\tilde\theta(U,V)\sqrt{2\tilde\kappa\partial_{\tilde v}\tilde\varpi}\,(U,\tilde v).
\end{equation}
Take the sign of $\tilde\theta(U,V)$ to be $+1$ if $\tilde\theta(U,V)\geq 0$, and $-1$ if $\tilde\theta(U,V)<0$.
These choices guarantee~\eqref{initial_v} and~\eqref{kappa_at_zero}. Together with the values of $\tilde r(u,V)$,
$\tilde\nu(u,V)$ and $\tilde\zeta(u,V)$, they provide initial data for 
the first order system~\eqref{r_u}$-$\eqref{kappa_at_u} on $]0,U]\times\{V\}\cup\{U\}\times[V,V+\eps]$.

\begin{Thm}[Extending the solution of the first order system beyond the Cauchy horizon]\label{para-la} 
Suppose that  
$$  
|\zeta_0(u)|\leq cu^s\ {\rm for\ some}\ s>\frac{13\Psi}9-1,
$$ 
where $c>0$ and $u\in\left[0,U\right]$.
Then, provided that $U$ is sufficiently small, there exist (non-unique) extensions of the solution of the first order system~\eqref{r_u}$-$\eqref{kappa_at_u} beyond the Cauchy horizon, which are still solutions \linebreak of~\eqref{r_u}$-$\eqref{kappa_at_u}.
\end{Thm}

\begin{proof}
Choose any continuously differentiable extension of $\tilde\varpi(U,\,\cdot\,)$, with $\partial_{\tilde v}\tilde\varpi(U,\,\cdot\,)\geq 0$.
As described above, this determines initial data for 
the first order system~\eqref{r_u}$-$\eqref{kappa_at_u} on $]0,U]\times\{V\}\cup\{U\}\times[V,V+\eps]$, for some $\eps>0$.
According to\/ {\rm Theorem~\ref{existence-alt}}, there exists a unique solution defined on a maximal reflected past set $\mathcal{R}$ containing a neighborhood of $]0,U]\times\{V\}\cup\{U\}\times[V,V+\eps]$.
This is an extension of the original solution beyond the Cauchy horizon: as explained in the discussion preceding Theorem 4.2 of Part 1, solutions of~\eqref{r_u}$-$\eqref{kappa_at_u} can be glued along a common edge of two rectangles provided that all functions coincide on that edge, since the extended functions are clearly continuous and the equations imply the continuity of the relevant partial derivatives.
\end{proof}

\begin{Rmk}
The original version of the existence and uniqueness theorem in\/ {\rm Part~1} could have been used here instead of\/ {\rm Theorem~\ref{existence-alt}}, by defining our coordinate $\tilde v$ using the values of $r(\delta,\,\cdot\,)$
instead of the values $r(U,\,\cdot\,)$, i.e.\ we could have replaced~\eqref{smooth_coordinate} by
$$
\tilde v=f(v)=r(\delta,0)-r(\delta,v).
$$
In this case, we should 
consider the first order system with initial data on\/ $[\delta,U]\times\{V\}\cup\{\delta\}\times[V,V+\eps]$.
We would then obtain an extension of the solution to\/ $[\delta,U]\times[V,V+\tilde\eps]$ for some $0 < \tilde\eps < \eps$. However, our approach above, using\/ {\rm Theorem~\ref{existence-alt}}, guarantees
that the domain of our extended solution contains a neighborhood of the whole Cauchy horizon\/ $]0,U]\times\{V\}$.
If we had insisted on using the original existence and uniqueness theorem in\/ {\rm Part~1}, we would only have known
that there existed a solution whose domain contained a neighborhood of $]\delta,U]\times\{V\}$, for $\delta$ arbitrarily small; but if $\delta$ changed, the solution might change,
because we would have to change the initial data.
\end{Rmk}

We now wish to see that the solution of our first order system corresponds to a solution of the Einstein equations.
Using Propositions~\ref{einstein-2} and~\ref{einstein-3}, 
we know that this is the case provided that the regularity hypothesis (h4) (see Section~\ref{main-one-two}) is satisfied, which
it is. Indeed, the extended solution is a solution of the backward problem where $\tilde\lambda(U,\tilde v) \equiv -1$ and $\tilde\kappa(U, \tilde v)$ are $C^1$ on $[0,V+\eps]$ by our choice of initial data. On the other hand, $\tilde\nu(u,0)= \nu_0(u) \equiv -1$.
Hence, we proved
\begin{Thm}[Extending the solution of the Einstein equations beyond the Cauchy horizon]\label{fine} 
Under the hypotheses of\/ {\rm Theorem~\ref{para-la}}, there exists a neighborhood  ${\cal V}$ of $]0,U]\times\{V\}$ such that
the extended functions $\check r$, $\check\phi$ and $\check\Omega$ are (classical) solutions of the Einstein equations 
\eqref{wave_r}, \eqref{wave_phi}, \eqref{r_uu}, \eqref{r_vv} and~\eqref{wave_Omega} in  $[0,U]\times[0,V]\setminus\{(0,V)\}\cup{\cal V}$.
\end{Thm}

\begin{Rmk}\label{fail}
By\/ {\rm Lemma~\ref{regular}} we conclude that $\check r$ is $C^2$, and $\check\nu$ and $\check\kappa$ are $C^1$.
Therefore, ${\check\Omega}^2$ is $C^1$, and so \label{metric_field} the metric is also $C^1$. %Note that in any other coordinate system the metric is $C^1$.
The field $\check\phi$ is also $C^1$ because $\check\theta$ and $\check\zeta$ are continuous.
Furthermore, $\partial_u\partial_{\tilde v}\check\Omega^2$
exists and is continuous in this $(u,\tilde v)$ chart. We emphasize that $\check\Omega^2$ does not have to be $C^2$ in this $(u,\tilde v)$ chart. Indeed, 
$$
\check\Omega^2(u,0)=-4\nu_0(0)\kappa(u,0)=4e^{-\int_0^u\frac{\zeta_0^2(u')}{r(u',0)}\,du'}.
$$
This implies
\begin{equation}\label{lip}
\partial_u\check\Omega^2(u,0)=-4\frac{\zeta_0^2(u)}{r(u,0)}e^{-\int_0^u\frac{\zeta_0^2(u')}{r(u',0)}\,du'},
\end{equation}
and
\begin{eqnarray*}
\partial^2_u\check\Omega^2(u,0)&=&4\left(-\,\frac{\zeta_0^2(u)}{r^2(u,0)}+\frac{\zeta_0^4(u)}{r^2(u,0)}-\frac{(\zeta_0^2)'(u)}{r(u,0)}\right)e^{-\int_0^u\frac{\zeta_0^2(u')}{r(u',0)}\,du'},
\end{eqnarray*}
with $r(u,0)=r_+-u$. So, if\/ $0\leq u\leq U$ is a point where $\zeta_0^2$ is not differentiable, then $\partial^2_u\check\Omega^2(u,0)$ does not exist.
\end{Rmk}

\vspace{4mm}

\noindent {\em The Kretschmann scalar.} Consider ${\cal M}$ as a $C^3$ manifold.
We finish with some remarks about the behavior of the Kretschmann scalar
$$
R_{\alpha\beta\gamma\delta}R^{\alpha\beta\gamma\delta},
$$
whose blowup prevents the existence of $C^2$ extensions of the metric across the Cauchy horizon.
A~straightforward, though lengthy, computation shows that 
\begin{eqnarray*}
 R_{\alpha\beta\gamma\delta}R^{\alpha\beta\gamma\delta}
 &=&\frac{16}{r^6}\left[\left(\varpi-\,\frac{3e^2}{2r}+\frac\Lambda 6r^3\right)+\frac{r(1-\mu)}2\Bigl(\frac{\zeta}{\nu}\Bigr)\Bigl(\frac{\theta}{\lambda}\Bigr)\right]^2\\
&&+\frac{16}{r^6}\left(\varpi-\,\frac{e^2}{2r}+\frac\Lambda 6r^3\right)^2
+\frac{16}{r^6}\left(\varpi-\,\frac{e^2}{r}-\,\frac\Lambda 3r^3\right)^2\\
&&+4\frac{(1-\mu)^2}{r^4}\Bigl(\frac{\zeta}{\nu}\Bigr)^2\Bigl(\frac{\theta}{\lambda}\Bigr)^2
\end{eqnarray*}
(see~\cite[Section~2]{RendallNature}, \cite[Equation~(5)]{Dafermos1} and~\cite[Appendix~A]{DafermosStrong}). Note that if $e=\Lambda=0$, the Kretschmann scalar reduces to
$\frac{48\varpi^2}{r^6}$, the well known value for the Schwarzschild metric.
\begin{Rmk}[Kretschmann scalar]\ \label{RmkSCC1}
\begin{enumerate}[{\rm (i)}]
\item
Under the hypotheses of\/ {\rm Theorem~\ref{inflation2}}, for each $0<u\leq U$,
$$ (R_{\alpha\beta\gamma\delta}R^{\alpha\beta\gamma\delta})(u,\tilde v)\to\infty,\ {\rm as}\ \tilde v\nearrow V.$$
\item
Under the hypotheses of\/ {\rm Theorem~\ref{fechado}}, 
$$\exists_{C>0}\ |R_{\alpha\beta\gamma\delta}R^{\alpha\beta\gamma\delta}|\leq C.$$
\end{enumerate}
\end{Rmk}
\begin{proof}
In case (i), the conclusion is immediate if $\tilde\varpi(u,V)=\infty$. When $\tilde\varpi(u,V)<\infty$,
we know that $\tilde\varpi(u,V)$ is close to $\varpi_0$ for small $u$.
We have estimates~\eqref{up-n} and~\eqref{piece-2}, for $-\nu$ from above and for $\zeta$ from below, respectively, and
also that $(1-\mu)(u,\,\cdot\,)$ is bounded from above by a negative constant (see the proof of Proposition~\ref{fechado-0};
it applies to the present situation because we only need $\tilde\varpi(u,V)$ to be close to $\varpi_0$ to show that $\tilde\nu(u,V)<0$).
Therefore, the result follows from $\bigl|\frac{\tilde\theta}{\tilde\lambda}\bigr|(u,\tilde v)\to+\infty$, as $\tilde v\nearrow V$, for $u>0$.

In case (ii), the renormalized mass $\varpi$ and
$\bigl|\frac{\tilde\theta}{\tilde\lambda}\bigr|$ are bounded (see~\eqref{eq34}).
\end{proof}

% \begin{Rmk} \label{RmkSCC2}
% In the case of\/ {\rm Theorem~\ref{inflation2}}, \eqref{bar_rafaeli} implies that the metric does not admit 
% a $C^1$ extension across the Cauchy horizon if $\varpi(u,\infty)=\infty$ for small positive $u$.
% If  $\varpi(u,\infty)<\infty$ for small positive $u$, then the metric does not admit a $C^1$ extension across the Cauchy horizon,
% because, in any coordinate system that covers ${\cal M}_\delta$ the 
% metric does not have bounded Christoffel symbols. Indeed, in this case, the proof of\/ {\rm Proposition~\ref{meio}} applies to show
% that $\tilde \Gamma^{\tilde v}_{\tilde v\tilde v}(U,\tilde v)$ tends to $-\infty$  
% as $\tilde v\nearrow V$. 

% In the case of\/ {\rm Theorem~\ref{fechado}}, we know that the metric admits a $C^1$ extension across the Cauchy horizon 
% (see\/~{\rm Remark~\ref{fail}}).
% \end{Rmk}

\appendix

\section{On the choice of the parameters and its consequences}\label{appendix-2} 

The objective of this appendix is to study the behavior of $\Psi$ (defined in~\eqref{def_psi},
the quotient of the surface gravities at $r_-$ and at $r_+$ of the reference subextremal Reissner-Nordstr\"{o}m black hole) as a function of
the parameters $\Lambda$, $\varpi_0$ and $e$. It turns out that it is easiest to express $\Psi$ in terms of the new parameters $\sigma$ and $\Upsilon$,
defined in~\eqref{sigma-upsilon}. The formula for $\Psi$ in terms of $\sigma$ and $\Upsilon$ is given in~\eqref{dois}. At the end of this appendix,
the reader can find a figure showing the behavior of $\Psi$ in the $(\sigma,\Upsilon)$ plane.

We consider the fourth order polynomial 
$$
p(r):=r^2(1-\mu)(r,\varpi_0)=-\,\frac\Lambda 3r^4+r^2-2\varpi_0r+e^2.
$$
Since we assume $p$ has zeros at $r_-$ and $r_+$, it can be factored as
$$
p(r)=[r^2-(r_++r_-)r+r_-r_+]\left[{\textstyle-\,\frac\Lambda 3r^2+cr+\frac{e^2}{r_-r_+}}\right].
$$
The constant $c$ can be computed by imposing that the coefficient of $p$ in $r^3$ is equal to zero.
We obtain $c=-\,\frac\Lambda 3(r_-+r_+)$. Hence, $p$ can be factored as 
$$
p(r)=[r^2-(r_++r_-)r+r_-r_+]\left[{\textstyle-\,\frac\Lambda 3r^2-\,\frac\Lambda 3(r_-+r_+)r+\frac{e^2}{r_-r_+}}\right].
$$
Since the coefficient of $p$ in $r$ is equal to $-2\varpi_0$, we must have
$$
\varpi_0=\frac{e^2}{2r_-}+\frac{e^2}{2r_+}+\frac\Lambda 6(r_-+r_+)r_-r_+.
$$
On the other hand, since the coefficient of $p$ in $r^2$ is equal to $1$, we must have
$$
\frac{e^2}{r_-r_+}=1-\frac{\Lambda}{3}(r_-^2+r_-r_++r_+^2).
$$
We define 
\begin{equation}\label{sigma-upsilon}
\sigma:=\frac{r_+}{r_-}\qquad{\rm and}\qquad\Upsilon:=\frac{\Lambda r_-^2}{3}.
\end{equation}
Then
$$
\frac{e^2}{r_-r_+}=1-\Upsilon(\sigma^2+\sigma+1).
$$
A simple computation shows that
$$
\frac{\varpi_0}{r_-}=\frac 12(\sigma+1)[1-\Upsilon(\sigma^2+1)].
$$
Of course, we could think of $\Lambda$, $\varpi_0$ and $e$ as the independent parameters, and use the equation $p(r)=0$ to determine $r_-$ and $r_+$. Instead,
we think of $r_-$, $r_+$ and $\Lambda$ as the independent parameters, and $\varpi_0$ and $e$ as the dependent ones. More precisely,
we regard $r_-$, $\sigma$ and $\Upsilon$ as the independent parameters and $\frac{e^2}{r_-r_+}$ and $\frac{\varpi_0}{r_-}$ as the dependent ones.
Clearly, $\sigma>1$.

When $\Lambda>0$, the polynomial $p$ has a third positive root $r_c$, the radius of the Reissner-Nordstr\"{o}m de Sitter cosmological event horizon. This is the positive solution of
$$
r^2+(r_-+r_+)r-\frac{3e^2}{\Lambda r_-r_+}=0.
$$
The value of $r_c$ is given by
$$
r_c=\frac{-(r_-+r_+)+\sqrt{(r_-+r_+)^2+\frac{12e^2}{\Lambda r_-r_+}}}{2}\,.
$$
The fact that $r_+<r_c$ imposes a restriction on our independent parameters, namely
$$
\frac{3e^2}{\Lambda r_-r_+}>2r_+^2+r_-r_+.
$$
In terms of $\sigma$ and $\Upsilon$, this can be written as
$$
\frac{1-\Upsilon(\sigma^2+\sigma+1)}{\Upsilon}>2\sigma^2+\sigma,
$$
or
\begin{equation}\label{upsilon}
\Upsilon<\frac{1}{3\sigma^2+2\sigma+1}.
\end{equation}
If $\Lambda\leq 0$, condition~\eqref{upsilon} is also trivially satisfied.
We say that a choice of parameters $(\sigma,\Upsilon)$ is admissible if $\sigma>1$ and~\eqref{upsilon} holds.

Now we compute $\Psi$ as defined in~\eqref{def_psi}, obtaining
\begin{eqnarray}
\Psi&=&\left(\frac{r_+}{r_-}\right)^2\frac{\frac{e^2}{r_-}+\frac\Lambda 3r_-^3-\varpi_0}{-\,\frac{e^2}{r_+}-\,\frac\Lambda 3r_+^3+\varpi_0}\nonumber\\
&=&\left(\frac{r_+}{r_-}\right)^2\frac{\frac{e^2}{r_-r_+}\frac{r_+}{r_-}+\frac{\Lambda r_-^2}{3}-\frac{\varpi_0}{r_-}}{-\,\frac{e^2}{r_-r_+}-\,\frac{\Lambda r_-^2}{3}\frac{r_+^3}{r_-^3}+\frac{\varpi_0}{r_-}}\nonumber\\
&=&\sigma^2\frac{(1-\Upsilon(\sigma^2+\sigma+1))\sigma+\Upsilon-\frac 12(\sigma+1)[1-\Upsilon(\sigma^2+1)]}{-(1-\Upsilon(\sigma^2+\sigma+1))-\Upsilon\sigma^3+\frac 12(\sigma+1)[1-\Upsilon(\sigma^2+1)]}\nonumber\\
&=&\sigma^2\frac{1-\Upsilon(\sigma^2+2\sigma+3)}{1-\Upsilon(3\sigma^2+2\sigma+1)}.\label{dois}
\end{eqnarray}
Taking into account~\eqref{upsilon}, in the region of interest, the condition $\Psi>1$ is equivalent to
$$
\Upsilon<\frac{1}{3\sigma^2+2\sigma+1}\quad{\rm and}\quad\Upsilon<\frac 1{(\sigma+1)^2}.
$$
As the first upper bound is smaller than the second, we conclude that for all admissible choices of parameters we have $\Psi>1$, that is
$$ 
-\partial_r(1-\mu)(r_-,\varpi_0)>\partial_r(1-\mu)(r_+,\varpi_0).
$$ 

We prove mass inflation in the region $\Psi>2$. Using~\eqref{upsilon} and~\eqref{dois}, the condition $\Psi>2$ 
is equivalent to 
$$
\frac{\sigma^2-2}{\sigma^4+2\sigma^3-3\sigma^2-4\sigma-2}<\Upsilon<\frac{1}{3\sigma^2+2\sigma+1}
$$
if
$$
\sigma<\sigma_0:=\frac 12\left(-1+{\textstyle\sqrt{9+4\sqrt{6}}}\right)\approx 1.66783.
$$
The value $\sigma_0$ is the only positive solution of $\sigma^4+2\sigma^3-3\sigma^2-4\sigma-2=0$.
For $\sigma\geq\sigma_0$, the condition $\Psi>2$, with the restriction~\eqref{upsilon}, is always satisfied.
Indeed, for $\sigma>\sigma_0$, we have
$$
\frac{\sigma^2-2}{\sigma^4+2\sigma^3-3\sigma^2-4\sigma-2}>\frac{1}{3\sigma^2+2\sigma+1}
$$
because the difference
$$
\frac{\sigma^2-2}{\sigma^4+2\sigma^3-3\sigma^2-4\sigma-2}-\frac 1{(\sigma+1)^2}
$$
is equal to
$$
\frac{2\sigma^2}{(\sigma^4+2\sigma^3-3\sigma^2-4\sigma-2)(\sigma+1)^2},
$$
and this is positive for $\sigma>\sigma_0$.

In the next figure we sketch part of the $(\sigma,\Upsilon)$-plane. As we just saw, the restriction $r_+<r_c$ 
translates into~\eqref{upsilon} and this region (shaded in the figure) is the only relevant one for our purposes.
We remark that the limit value of $\Psi$ on the line $\sigma=1$ is one.

\vspace{1mm}

% \begin{center}
% \begin{psfrags}
% \psfrag{q}{\tiny $\!\!\!\!\!\!\!\sqrt{2}$}
% \psfrag{r}{\tiny $\!\sigma_0$}
% \psfrag{s}{\tiny $\sigma=\frac{r_+}{r_-}$}
% \psfrag{L}{\tiny $\!\!\!\!\!\!\!\!\!\Upsilon=\frac{\Lambda r_-^2}{3}$}
% \psfrag{t}{\tiny $\!\!\!\frac 16$}
% \psfrag{f}{\tiny $\!\!\!\frac 14$}
% \psfrag{b}{\tiny $r_+=r_c$}
% \psfrag{c}{\tiny $\Psi=1$}
% \psfrag{z}{\tiny $\!\!\Psi=1$}
% \psfrag{y}{\tiny $\!\!\Psi=2$}
% \psfrag{d}{\tiny $\Psi=2$}
% \psfrag{1}{\tiny $1$}
% \psfrag{p}{\tiny $\!\!\Psi=+\infty$}
% \psfrag{m}{\tiny $\Psi=-\infty$}
% \includegraphics[scale=1]{fig15.eps}
% \end{psfrags}
% \end{center}

\begin{center}
\begin{psfrags}
\psfrag{q}{\tiny $\!\!\!\!\!\!\!\sqrt{2}$}
\psfrag{r}{\tiny $\!\sigma_0$}
\psfrag{s}{\tiny $\sigma=\frac{r_+}{r_-}$}
\psfrag{L}{\tiny $\!\!\!\!\!\!\!\!\!\Upsilon=\frac{\Lambda r_-^2}{3}$}
\psfrag{t}{\tiny $\!\!\!\frac 16$}
\psfrag{f}{\tiny $\!\!\!\frac 14$}
\psfrag{b}{\tiny $\Psi=+\infty \Leftrightarrow r_+=r_c$}
\psfrag{c}{\tiny $\Psi=1$}
\psfrag{z}{\tiny $\!\!\Psi=1$}
\psfrag{y}{\tiny $\!\!\Psi=2$}
\psfrag{d}{\tiny $\Psi=2$}
\psfrag{1}{\tiny $1$}
\psfrag{p}{}%{\tiny $\!\!\Psi=+\infty$}
\psfrag{m}{}%{\tiny $\Psi=+\infty$}
\includegraphics[scale=1]{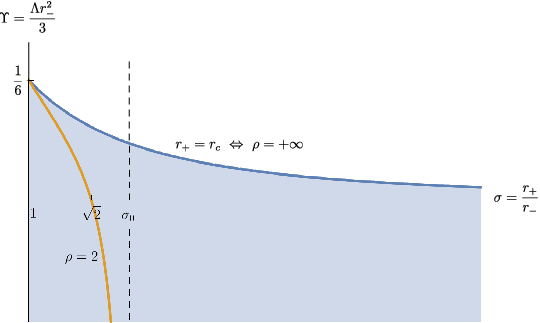}
\end{psfrags}
\end{center}

\section{Proof of Theorem~\ref{inflation}}\label{prova-massa}

We start by establishing the following useful result.

\begin{Lem}\label{sinal-theta} 
 Assume that $\zeta_0(u)>0$ for $u>0$. Then $\theta>0$ and $\zeta>0$ in ${\cal P}\setminus\{0\}\times[0,\infty[$.
\end{Lem}

\begin{proof}
The proof proceeds in three steps.

{\em Step 1}. If $\theta_0>0$ and $\zeta_0>0$, then $\theta>0$ and $\zeta>0$ in ${\cal P}$.
Otherwise, there would exist a point $(u,v)\in{\cal P}$ such that $\theta(u,v)=0$ or $\zeta(u,v)=0$
but $\theta>0$ and $\zeta>0$ in $J^-(u,v)$. Integrating~\eqref{theta_u} and~\eqref{zeta_v},
we obtain a contradiction.

{\em Step 2}. Since in Part~1 we proved continuous dependence of the solution on $\theta_0$ and $\zeta_0$, 
if $\theta_0\geq 0$ and $\zeta_0\geq 0$, then $\theta\geq 0$ and $\zeta\geq 0$.

{\em Step 3}.
 Suppose that $(u,v)\in{\cal P}\setminus\{0\}\times[0,\infty[$.
Since $\zeta_0(u)>0$ for $u>0$, \eqref{zeta_v} implies that $\zeta(u,v)>0$, because, from the previous step, $\theta\geq 0$.
So $\zeta>0$ in ${\cal P}\setminus\{0\}\times[0,\infty[$. Now~\eqref{theta_u} implies that $\theta$ is positive on ${\cal P}\setminus\{0\}\times[0,\infty[$
because $\lambda$ is negative on this set.
\end{proof}

\begin{Cor}\label{Jim}
Under the hypotheses of\/ {\rm Lemma~\ref{sinal-theta}}, for $u>\bar{u}$ and $v>\bar{v}$,
we have
$$
\varpi(u,v)-\varpi(u,\bar{v})\geq \varpi(\bar{u},v)-\varpi(\bar{u},\bar{v}).
$$
\end{Cor}
\begin{proof}
This is an easy consequence of the fact that 
$$\partial_u\partial_v\varpi=-\frac{\theta\zeta\lambda}{\kappa r}-\frac{\theta^2\zeta^2}{2\kappa r \nu}\geq 0.$$
\end{proof}

\begin{proof}[Proof of\/ {\rm Theorem~\ref{inflation}}]
We follow the argument on pages~493--497 of~\cite{Dafermos2}. 
We consider the same three cases as in the proof of Lemma~\ref{l-kappa}, presented in Part~2.

{\em Case 1}. If~\eqref{infinito} holds, there is nothing to prove.

{\em Case 2}. If
\begin{equation}
\label{welch}
\lim_{u\searrow 0}\varpi(u,\infty)>\varpi_0,
\end{equation}
then~\eqref{infinito} holds. This was proven on page~494 of~\cite{Dafermos2} and is repeated here for the convenience of the reader. Suppose that~\eqref{welch} holds.
Then there exists $\eps>0$ such that
$$
\lim_{u\searrow 0}\varpi(u,\infty)>\varpi_0+3\eps.
$$
Since $\lim_{v\to\infty}\varpi(u_\gamma(v),v)=\varpi_0$ (see~\eqref{espana}), we have
$$
\varpi(u_\gamma(v),v)<\varpi_0+\eps
$$
for, say, $v\geq V$. Hence, if $u$ is sufficiently small, there exists $v=v_{\gamma,\eps}(u)$ so that $\varpi(u,v_{\gamma,\eps}(u))-\varpi(u,v_\gamma(u))=\eps$. We construct a sequence
$(u_n,v_n)$ in the following way. Starting with $(u_0,v_0)=(u_\gamma(V),V)$, we let
$(u_{n+1},v_{n+1})$ be such that $v_{n+1}=v_{\gamma,\eps}(u_n)$ and
$u_{n+1}=u_\gamma(v_{n+1})$. According to Lemma~\ref{Jim},
$$
\varpi(u_n,v_{n+1})-\varpi(u_n,v_n)\geq \varpi(u_{n+1},v_{n+1})-\varpi(u_{n+1},v_n)=\eps
$$
for all $n$,
$$
\varpi(u_n,v_{n+2})-\varpi(u_n,v_{n+1})\geq \varpi(u_{n+1},v_{n+2})-\varpi(u_{n+1},v_{n+1})\geq\eps
$$
for all $n$,
and
$$
\varpi(u_n,v_{n+k})-\varpi(u_n,v_{n+k-1})\geq \varpi(u_{n+1},v_{n+k})-\varpi(u_{n+1},v_{n+k-1})\geq\eps
$$
for all $n$ and all $k\geq 1$.
Hence,
$$
\varpi(u_n,v_{n+k})-\varpi(u_n,v_n)\geq k\eps
$$
for all $n$ and all $k\geq 1$.
This implies~\eqref{infinito} because $\partial_v\varpi\geq 0$.

%{\em Case 2}. If $\lim_{u\searrow 0}\varpi(u,\infty)>\varpi_0$, then~\eqref{infinito} holds (see page~494 of~\cite{Dafermos2}).

{\em Case 3}. Suppose now that $\lim_{u\searrow 0}\varpi(u,\infty)=\varpi_0$. 
As in the proof of Lemma~\ref{l-kappa}, we have $\left(\vlinha\right)(u,v)\geq 0$ for $(u,v) \in J^+(\cg_{\ckrm})$ and $u$ sufficiently small. 
Then, from~\eqref{lambda_u} it follows that $\partial_u(-\lambda)\leq 0$ in $ J^+(\cg_{\ckrm})$, whereas from Lemma~\ref{sinal-theta} it follows $\partial_u\theta\geq 0$.
As a consequence, the integral
$$
I(u):=\int_{\vgr(u)}^\infty\left[\frac{\theta^2}{-\lambda}\right](u,\tilde v)\,d\tilde v
$$
is a nondecreasing function of $u$. Therefore we have two alternatives to consider.

{\em Case 3.1.} $I(u)=+\infty$ for all small $u$, say $0<u\leq U$. Consider such a~$u$. We observe that the following limit exists and is finite:
$$
\lim_{v\nearrow\infty}(1-\mu)(u,v)=:(1-\mu)(u,\infty)=1-\frac{2\varpi(u,\infty)}{r(u,\infty)}+\frac{e^2}{r^2(u,\infty)}-\frac\Lambda 3r^2(u,\infty).
$$
Equation~\eqref{nu_v} and $\left(\vlinha\right)(u,v)\geq 0$ imply that $v\mapsto\nu(u,v)$ is a nondecreasing function in $J^+(\gam)$. So we may define
\begin{eqnarray*}
\nu(u,\infty)&=&\lim_{v\nearrow+\infty}\nu(u,v).
\end{eqnarray*}
Integrating~\eqref{ray_v_bis} we get
$\lim_{v\nearrow\infty}\frac{\nu}{1-\mu}(u,v)=0$. Therefore, $\nu(u,\infty)=0$. Let $0<\delta<u\leq U$. Clearly,
$$
r(u,v)=r(\delta,v)+\int_\delta^u\nu(s,v)\,ds.
$$
 Thus, by Lebesgue's Monotone Convergence Theorem, 
$$ 
r(u,\infty)=r(\delta,\infty)+\int_\delta^u\nu(s,\infty)\,ds=r(\delta,\infty).
$$ 
Letting $\delta$ decrease to zero, due to~\eqref{r-menos}, we obtain $r(u,\infty)\equiv r_-$. 
This contradicts Theorem~\ref{rmenos}.

{\em Case 3.2.} $I(u)<+\infty$ for all small $u$, say $0<u\leq U$. Arguing as in pages 495--496 of~\cite{Dafermos2}, we know $\lim_{u\searrow 0}I(u)=0$.
We will use this information to improve our upper bound on $-\lambda$ in the region $J^+(\gam)$. Then we will obtain a lower bound for
$\theta$ in this region. Finally, we use these bounds to arrive at the contradiction that $I(u)=+\infty$.

Let $\eps>0$. As $\lim_{u\searrow 0}I(u)=0$, we may choose $U>0$ sufficiently small, so that for all $(u,v)\in J^+(\gam)$ with $0<u\leq U$,
\begin{equation}\label{integrating_factor2}
e^{\frac 1{r(U,\infty)}\int_{\vgr(\bar u)}^v\bigl[\big|\frac\theta\lambda\bigr||\theta|\bigr](\bar u,\tilde v)\,d\tilde v}\leq 1+\eps,
\end{equation}
for $\bar u\in[\ugam(v),u]$. 
\begin{center}
\begin{turn}{45}
\begin{psfrags}
\psfrag{f}{{\tiny \!\!\!\!\!\!\!\!\!\!\!\!\!\!$\ugam(v)$}}
\psfrag{y}{{\tiny $\gam$}}
\psfrag{u}{{\tiny $u$}}
\psfrag{a}{{\tiny $v$}}
\psfrag{x}{{\tiny $\ \cg_{\ckrm}$}}
\psfrag{d}{{\tiny \!\!\!\!$u$}}
\psfrag{e}{{\tiny \!\!\!\!$u$}}
\psfrag{v}{{\tiny $v$}}
\psfrag{b}{{\tiny $\vgr(u)$}}
\psfrag{c}{{\tiny $\vgr(\ugam(v))$}}
\psfrag{p}{{\tiny $(u,v)$}}
\includegraphics[scale=.7]{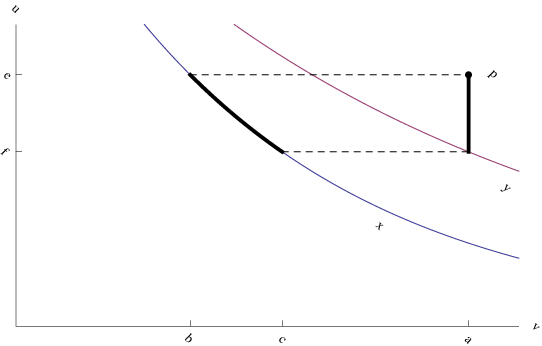}
\end{psfrags}
\end{turn}
\end{center}
\noindent
Next we use~\eqref{ray_v_bis}, \eqref{D} and~\eqref{integrating_factor2}.
We may bound the integral of $\nu$ along $\cg_{\ckrm}$ in terms of the integral of $\frac{\nu}{1-\mu}$ on the segment
$\bigl[\ugam(v),u\bigr]\times\{v\}$ in the following way:
\begin{eqnarray}
&&-\int_{\ugam(v)}^u\nu(\tilde u,\vgr(\tilde u))\,d\tilde u \leq\nonumber \\
&&\qquad\qquad\qquad-\min_{\cg_{\ckrm}}(1-\mu)
\int_{\ugam(v)}^u\frac\nu{1-\mu}(\tilde u,\vgr(\tilde u))\,d\tilde u\leq\nonumber\\
&&\qquad\qquad\qquad-(1+\eps)\min_{\cg_{\ckrm}}(1-\mu)
\int_{\ugam(v)}^u\frac\nu{1-\mu}(\tilde u,v)\,d\tilde u.\label{EE}
\end{eqnarray}
Applying successively~\eqref{EE}, \eqref{integral}, \eqref{F}, and~\eqref{kappaquaseum},
\begin{eqnarray}
&&\int_{\ugam(v)}^u\frac\nu{1-\mu}(\tilde u,v)\,d\tilde u\nonumber\\
&&\qquad\qquad\geq \mbox{\tiny$\frac{1}{-(1+\eps)\min_{\cg_{\ckrm}}(1-\mu)}$}\int_{\ugam(v)}^u-\nu(\tilde u,\vgr(\tilde u))\,d\tilde u\nonumber\\
&&\qquad\qquad= \mbox{\tiny$\frac{1}{-(1+\eps)\min_{\cg_{\ckrm}}(1-\mu)}$}\int_{\vgr(u)}^{\vgr(\ugam(v))}-\lambda(\ugr(\tilde v),\tilde v)\,d\tilde v\nonumber\\
&&\qquad\qquad\geq \mbox{\tiny$\frac{\max_{\cg_{\ckrm}}(1-\mu)}{(1+\eps)\min_{\cg_{\ckrm}}(1-\mu)}$}\int_{\frac{\vgam(u)}{1+\beta}}^{\frac v{1+\beta}}\kappa(\ugr(\tilde v),\tilde v)\,d\tilde v\nonumber\\
&&\qquad\qquad\geq \mbox{\tiny$\frac{(1-\eps)\max_{\cg_{\ckrm}}(1-\mu)}{(1+\eps)\min_{\cg_{\ckrm}}(1-\mu)}$}\bigl({\textstyle\frac{v-\vgam(u)}{1+\beta}}\bigr).\label{argue2}
\end{eqnarray}
Thus,
\begin{eqnarray}
&&e^{\int_{\ugam(v)}^{u}\bigl[\frac{\nu}{1-\mu}\partial_r(1-\mu)\bigr](\tilde u,v)\,d\tilde u}\nonumber\\
&&\qquad\qquad\leq
e^{\bigl[\max_{J^+(\gam)}\partial_r(1-\mu)\bigr]
\int_{\ugam(v)}^{u}\frac\nu{1-\mu}(\tilde u,v)\,d\tilde u}\nonumber\\
&&\qquad\qquad\leq
e^{\bigl[\max_{J^+(\gam)}\partial_r(1-\mu)\bigr]
\mbox{\tiny$\frac{(1-\eps)}{(1+\eps)}\frac{\max_{\cg_{\ckrm}}(1-\mu)}{\min_{\cg_{\ckrm}}(1-\mu)}$}\, \bigl({\textstyle\frac{v-\vgam(u)}{1+\beta}}\bigr)}\nonumber\\
&&\qquad\qquad\leq
e^{\bigl[\partial_r(1-\mu)(\ckrm,\varpi_0)+\max_{J^+(\gam)}\mbox{\tiny$\frac{2(\varpi-\varpi_0)}{r^2}$}\bigr]
\mbox{\tiny$\frac{(1-\eps)}{(1+\eps)}\frac{\max_{\cg_{\ckrm}}(1-\mu)}{\min_{\cg_{\ckrm}}(1-\mu)}$}\,\bigl({\textstyle\frac{v-\vgam(u)}{1+\beta}}\bigr)}\nonumber\\
&&\qquad\qquad\leq
e^{\bigl[\partial_r(1-\mu)(\ckrm,\varpi_0)+\mbox{\tiny$\frac\eps{(r_--\eps_0)^2}$}\bigr]
\mbox{\tiny$\frac{(1-\eps)}{(1+\eps)}\frac{\max_{\cg_{\ckrm}}(1-\mu)}{\min_{\cg_{\ckrm}}(1-\mu)}$}\, \bigl({\textstyle\frac{v-\vgam(u)}{1+\beta}}\bigr)}.\label{III}
\end{eqnarray}
We integrate~\eqref{lambda_u} and we use~\eqref{lambda-above} and~\eqref{III} to obtain
\begin{eqnarray}
-\lambda(u,v)&=&
-\lambda(\ugam(v),v)e^{\int_{\ugam(v)}^{u}\bigl[\frac{\nu}{1-\mu}\partial_r(1-\mu)\bigr](\tilde u,v)\,d\tilde u}\label{lambda_above}\\
&\leq&C
e^{(1-\tilde\delta)\partial_r(1-\mu)(\ckrm,\varpi_0)\bigl(\frac{\beta v}{1+\beta}+\frac v{1+\beta}-\frac{\vgam(u)}{1+\beta}\bigr)}\nonumber\\
&=&C(u)
e^{(1-\tilde\delta)\partial_r(1-\mu)(\ckrm,\varpi_0)v}.\label{lambda_cima}
\end{eqnarray}
The value of $\tilde\delta$ can be made
small by choosing $U$ sufficiently small. 
Here $C(u)=Ce^{-(1-\tilde\delta)\partial_r(1-\mu)(\ckrm,\varpi_0)\frac{\vgam(u)}{1+\beta}}$.
This is the desired upper estimate for $-\lambda$.

Now we turn to obtaining the lower estimate for $\theta$.
Combining~\eqref{bound_lambda2} with~\eqref{region_delta}, for $(u,v)\in\cg_{r_+-\delta}$ we have
\begin{eqnarray}
-\lambda(u,v)&\geq&\Bigl(\frac{r_+-\delta}{r_+}\Bigr)\frac{\partial_r(1-\mu)(r_+,\varpi_0)}{1+\eps}u\,e^{[\partial_r(1-\mu)(r_+,\varpi_0)-\eps]v}\nonumber\\
&\geq&Cu\,e^{[\partial_r(1-\mu)(r_+,\varpi_0)-\eps]v}.\label{piece-1}
\end{eqnarray}
Note that $C$ can be chosen independently of $\delta$.
Using~\eqref{zeta_v}, Lemma~\ref{sinal-theta} and~\eqref{zeta-0},
\begin{equation}\label{piece-2}
\zeta(u,v)\geq cu^s\ {\rm for\ all}\/\ (u,v).
\end{equation}
We take into account that~\eqref{piece-1} and~\eqref{piece-2} are valid for arbitrary $\delta$, small,
and that $J^-(\cg_{\ckrp})$ is foliated by curves $\cg_{r_+-\delta}$ for $0 < \delta < r_+ - \ckrp$. Therefore,
integrating~\eqref{theta_u}, 
for $(u,v)\in\cg_{\ckrp}$ we have
\begin{eqnarray}
\theta(u,v)&\geq& Cu^{s+2}e^{[\partial_r(1-\mu)(r_+,\varpi_0)-\eps]v}\nonumber\\
&\geq&Ce^{[-(s+1)\partial_r(1-\mu)(r_+,\varpi_0)-\tilde\eps]v}.\label{theta_below}
\end{eqnarray}
For the last inequality, we used~\eqref{region_delta}. The constant $C$ depends on $\ckrp$. The value of $\tilde\eps$ can be made
small by choosing $\ckrp$ sufficiently close to $r_+$. We know that $\partial_u\theta\geq 0$. Thus, \eqref{theta_below} also holds in $J^+(\cg_{\ckrp})$.
This is the desired lower estimate for $\theta$.

We can now obtain a lower bound for $I(u)$ using~\eqref{lambda_cima} and~\eqref{theta_below}:
\begin{eqnarray*}
I(u)&\geq&\int_{\vgam(u)}^\infty\left[\frac{\theta^2}{-\lambda}\right](u,\tilde v)\,d\tilde v\\
&\geq&C(u)\int_{\vgam(u)}^\infty\frac{e^{[-2(s+1)\partial_r(1-\mu)(r_+,\varpi_0)-2\tilde\eps]\tilde v}}{e^{(1-\tilde\delta)\partial_r(1-\mu)(\ckrm,\varpi_0)\tilde v}}\,d\tilde v.
\end{eqnarray*}
This integral is infinite if
$$
-2(s+1)\partial_r(1-\mu)(r_+,\varpi_0)-\partial_r(1-\mu)(\ckrm,\varpi_0)>0,
$$
or, equivalently,
\begin{equation} \label{Denise}
s<\frac 12\frac{-\partial_r(1-\mu)(\ckrm,\varpi_0)}{\partial_r(1-\mu)(r_+,\varpi_0)}-1,
\end{equation}
provided that $\tilde\eps$ and $\tilde\delta$ are chosen sufficiently small (which we can achieve by decreasing $U$ and $\delta$, if necessary).
To complete the proof of Theorem~\ref{inflation} we just have to note that given $s<\frac\Psi 2-1$ we can always choose $\ckrm$ so that \eqref{Denise} holds, contradicting $I(u)<\infty$.
\end{proof}

%% file: sh_appendix_for_part_three.tex
\section{Some useful formulas}\label{Costa}

Here we collect some formulas that were obtained in Part~2 and that are needed to study the behavior of the solution at the Cauchy horizon.

\subsection*{The Raychaudhuri equations written in terms of $\kappa$ and $\frac\nu{1-\mu}$}

Using equations (\ref{r_v}), (\ref{nu_v}), (\ref{omega_v}) and (\ref{kappa_at_u}), we get
\begin{equation}\label{ray_v_bis} 
\partial_v\left(\frac{\nu}{1-\mu}\right)=\frac{\nu}{1-\mu}\left(\frac{\theta}{\lambda}\right)^2\frac{\lambda}{r}.
\end{equation}
The equations~\eqref{kappa_u} and~\eqref{ray_v_bis} are the Raychaudhuri equations.

\subsection*{Evolution equations for $\frac\theta\lambda$ and $\frac\zeta\nu$} 

Using equations \eqref{lambda_u}, \eqref{theta_u} and \eqref{nu_v}, \eqref{zeta_v} we obtain, respectively,
 \begin{equation}\label{theta_lambda} 
  \partial_u\frac\theta\lambda=
	-\,\frac\zeta r-\frac\theta\lambda\frac{\nu}{1-\mu}\partial_r(1-\mu),
 \end{equation}
 \begin{equation}\label{zeta_nu} 
  \partial_v\frac\zeta\nu=
	-\,\frac\theta r-\frac\zeta\nu\frac{\lambda}{1-\mu}\partial_r(1-\mu).
 \end{equation}

\subsection*{The integrals of $\nu$ and $\lambda$ along a curve $\cg_{\ckr}$} Equation (118) in Part~2 is

\begin{equation} 
 \int_{\uckr(v)}^u\nu(\tilde u,\vckr(\tilde u))\,d\tilde u=\int_{\vckr(u)}^v\lambda(\uckr(\tilde v),\tilde v)\,d\tilde v.
\label{integral}
\end{equation}

\subsection*{Estimates in $J^-(\cg_{\ckrp})$} Estimates (42) and (50) in Part~2 are
$$
\left|\frac\zeta\nu\right|(u,v)\leq C\max_{\bar u\in[0,u]}|\zeta_0|(\bar u).
$$
\begin{eqnarray} 
-\,\frac{\lambda}{1-\mu}\partial_r(1-\mu)&\leq&
-\Bigl(\frac{\ckrp}{r_+}\Bigr)^{\hat\delta^2}\min_{r\in[\ckrp,r_+]}\partial_r(1-\mu)(r,\varpi_0)\nonumber\\
&=&-\alpha<0,\label{alp_new}
\end{eqnarray}
where $\hat\delta$ is a bound for $\bigl|\frac\zeta\nu\bigr|$ in $J^-(\ckrp)$.

\subsection*{Estimates for $(u,v)\in\cg_{r_+-\delta}$} Estimates (82) and (84) in Part~2 are

\begin{eqnarray} 
-\,\Bigl(\frac{r_+}{r_+-\delta}\Bigr)
\left(\frac{\partial_r(1-\mu)(r_+,\varpi_0)}{1-\eps}+\frac{4\tilde{\delta}}{r_+^2}\right)\,\delta\leq\lambda\qquad\qquad\qquad\qquad&&\label{bound_lambda2}\\
\leq
-\Bigl(\frac{r_+-\delta}{r_+}\Bigr)\frac{\partial_r(1-\mu)(r_+,\varpi_0)}{1+\eps}\,\delta,&&\nonumber
\end{eqnarray}

\begin{equation}\label{region_delta} 
\delta\, e^{-[\partial_r(1-\mu)(r_+,\varpi_0)+\eps]\,v}
\leq u \leq
\delta\, e^{-[\partial_r(1-\mu)(r_+,\varpi_0)-\eps]\,v},
\end{equation}
for $\delta > 0$ sufficiently small, where $\eps>0$ and $\tilde{\delta}>0$ can be chosen arbitrarily close to zero if $\delta$ is small enough.

\subsection*{Estimate in $J^-(\cg_{\ckrm})\cap J^+(\cg_{\ckrp})$} Estimate (79) in Part~2 is
\begin{equation}\label{nu_mu_bis} 
\Bigl(\frac{\ckrm}{r_+}\Bigr)^{\hat\delta^2}
\frac{1-\eps}{\partial_r(1-\mu)(r_+,\varpi_0)\,u}\leq\frac{\nu}{1-\mu}(u,v)\leq\frac{1+\eps}{\partial_r(1-\mu)(r_+,\varpi_0)\,u},
\end{equation}
where $\hat{\delta}>0$ can be chosen arbitrarily close to zero if $U$ is small enough.

\subsection*{Estimate in $J^-(\cg_{\ckrm})$} Equation (93) in Part~2 is
\begin{equation}\label{canto-part}
\lim_{\stackrel{(u,v)\to(0,\infty)}{{\mbox{\tiny{$(u,v)\in J^-(\ckrm)$}}}}}\,\Bigl|\frac\theta\lambda\Bigr|(u,v)=0.
\end{equation}

\subsection*{Relation between the integrals of $\lambda$ and $\kappa$ along the curve $\cg_{\ckrm}$} Estimates (119) and (120) in Part~2 are
\begin{eqnarray} 
&&-\max_{\cg_{\ckrm}}(1-\mu)
\int_{\vgr(u)}^v\kappa(\ugr(\tilde v),\tilde v)\,d\tilde v\label{F}\\
&&\qquad\qquad\leq -\int_{\vgr(u)}^v\lambda(\ugr(\tilde v),\tilde v)\,d\tilde v\leq\nonumber \\
&&\qquad\qquad\qquad\qquad-\min_{\cg_{\ckrm}}(1-\mu)
\int_{\vgr(u)}^v\kappa(\ugr(\tilde v),\tilde v)\,d\tilde v.\label{A}
\end{eqnarray}

\subsection*{Relation between the integrals of $\nu$ and $\frac\nu{1-\mu}$ along the curve $\cg_{\ckrm}$} Estimates (121) and (122) in Part~2 are
\begin{eqnarray} 
&&-\max_{\cg_{\ckrm}}(1-\mu)
\int_{\ugr(v)}^u\frac\nu{1-\mu}(\tilde u,\vgr(\tilde u))\,d\tilde u\label{B}\\
&&\qquad\ \ \leq-\int_{\ugr(v)}^u\nu(\tilde u,\vgr(\tilde u))\,d\tilde u\leq\nonumber \\
&&\qquad\qquad\qquad -\min_{\cg_{\ckrm}}(1-\mu)
\int_{\ugr(v)}^u\frac\nu{1-\mu}(\tilde u,\vgr(\tilde u))\,d\tilde u.\label{D}
\end{eqnarray}

\subsection*{Relation between the integrals of $\frac\nu{1-\mu}$ and $\kappa$ along the curve $\cg_{\ckrm}$} Estimates (130) and (123) in Part~2 are
\begin{eqnarray} 
&&\mbox{\tiny$\frac{\max_{\cg_{\ckrm}}(1-\mu)}{(1+\eps)\min_{\cg_{\ckrm}}(1-\mu)}$}\int_{\vgr(u)}^v\kappa(\ugr(\tilde v),\tilde v)\,d\tilde v\label{argue}\\
&&\qquad\ \ \leq\int_{\ugr(v)}^u\frac\nu{1-\mu}(\tilde u,v)\,d\tilde u\leq\nonumber\\
&&\qquad\qquad\qquad \mbox{\tiny$\frac{\min_{\cg_{\ckrm}}(1-\mu)}{\max_{\cg_{\ckrm}}(1-\mu)}$}\int_{\vgr(u)}^v\kappa(\ugr(\tilde v),\tilde v)\,d\tilde v.\label{argue3} 
\end{eqnarray}
where $\varepsilon>0$ can be chosen arbitrarily close to zero for appropriate choices of the parameters $\beta_-$, $\beta_+$, $\ckrm$, $\ckrp$, $\varepsilon_0$ and $U$.

\subsection*{Estimates in $J^-(\gam)\cap J^+(\cg_{\ckrm})$} Estimates (101), (105), (126) and (127) in Part~2 are
\begin{eqnarray} 
\Bigl|\frac\zeta\nu\Bigr|(u,v)
&\leq&C\sup_{[0,u]}|\zeta_0|e^{-\bigl(\frac\alpha{1+\beta^+}+\partial_r(1-\mu)(r_--\eps_0,\varpi_0)\beta\bigr) v},\label{zetaNuGamma}
\end{eqnarray}
\begin{eqnarray} 
\kappa(u,v) \geq 1 - \varepsilon,\label{kappaquaseum}
\end{eqnarray}
\begin{eqnarray} 
&&\Bigl|\frac\theta\lambda\Bigr|(u,v)
\leq C\sup_{[0,u]}|\zeta_0|e^{-\left(\frac\alpha{1+\beta^+}+
\partial_r(1-\mu)(r_--\eps_0,\varpi_0)\mbox{\tiny$\frac{\min_{\cg_{\ckrm}}(1-\mu)}{\max_{\cg_{\ckrm}}(1-\mu)}$}\beta\right)v},\label{thetaLambdaGamma}
\end{eqnarray}
\begin{equation}\label{integrating_factor} 
e^{\frac 1{r_--\eps_0}\int_{\vgr(\bar u)}^v\bigl[\big|\frac\theta\lambda\bigr||\theta|\bigr](\bar u,\tilde v)\,d\tilde v}\leq 1+\eps,
\end{equation}
where $\varepsilon>0$ can be chosen arbitrarily close to zero for appropriate choices of the parameters $\beta_-$, $\beta_+$, $\ckrm$, $\ckrp$, $\varepsilon_0$ and $U$.

\subsection*{Estimates for $(u,v)\in\gam$} Estimates (131), (109), (110), (135) and (136) in Part~2 are
\begin{eqnarray} 
&&e^{\int_{\ugr(v)}^{u}\bigl[\frac{\nu}{1-\mu}\partial_r(1-\mu)\bigr](\tilde u,v)\,d\tilde u}\nonumber\\
&&\qquad\qquad\leq
e^{\bigl[\partial_r(1-\mu)(\ckrm,\varpi_0)+\mbox{\tiny$\frac\eps{(r_--\eps_0)^2}$}\bigr]
\mbox{\tiny$\frac{(1-\eps)}{(1+\eps)}\frac{\max_{\cg_{\ckrm}}(1-\mu)}{\min_{\cg_{\ckrm}}(1-\mu)}$}\, \frac{\beta}{1+\beta} v},\label{I}
\end{eqnarray}
 \begin{eqnarray} 
 &&
 \tilde ce^{(1+\delta)\partial_r(1-\mu)(r_--\eps_0,\varpi_0)\frac\beta{1+\beta}\,v}\label{lambda-below}\\
 &&\qquad\qquad\qquad\qquad\leq
  -\lambda(u,v)\leq\nonumber\\
  &&\qquad\qquad\qquad\qquad\qquad\qquad
  \tilde Ce^{(1-\delta)\partial_r(1-\mu)(\ckrm,\varpi_0)\frac\beta{1+\beta}\,v},\label{lambda-above}
 \end{eqnarray}
\begin{eqnarray} 
&&ce^{-\partial_r(1-\mu)(r_+,\varpi_0)\frac{v}{1+\beta^-}}\label{O}\\
&&\qquad\qquad\qquad\qquad\leq u\leq \nonumber\\
&&\qquad\qquad\qquad\qquad\qquad Ce^{-\partial_r(1-\mu)(r_+,\varpi_0)\frac{v}{1+\beta^+}},\label{P}
\end{eqnarray}
for $\varepsilon<\varepsilon_0$.
The bound~\eqref{P} is actually valid in $J^-(\gam)\cap J^+(\cg_{\ckrm})$.

\subsection*{Estimates in $J^+(\gam)$} Lemmas~7.1 and 7.2 in Part~2 imply
\begin{eqnarray} 
-\lambda(u,v)&\leq&Ce^{(1-\delta)\partial_r(1-\mu)(\ckrm,\varpi_0)\frac\beta{1+\beta}\,v},\label{up-l}\\
-\nu(u,v)
&\leq&Cu^{\mbox{\tiny$-\,\frac{1+\beta^-}{1+\beta^+}\frac{\partial_r(1-\mu)(\ckrm,\varpi_0)}{\partial_r(1-\mu)(r_+,\varpi_0)}\beta$}\,-1},\label{up-n}
\end{eqnarray}
for appropriate choices of the parameters $\beta_-$, $\beta_+$, $\ckrm$, $\ckrp$, $\varepsilon_0$ and $U$.